\documentclass[sn-nature,pdflatex]{sn-jnl}

\usepackage{graphicx}
\usepackage{multirow}
\usepackage{amsmath,amssymb,amsfonts}
\usepackage{amsthm}
\usepackage{mathrsfs}
\usepackage[title]{appendix}
\usepackage{textcomp}
\usepackage{manyfoot}
\usepackage{booktabs}
\usepackage{xcolor}
\usepackage{algorithm}
\usepackage{algorithmicx}
\usepackage{algpseudocode}
\usepackage{listings}

\usepackage{upgreek}
\usepackage{bm}
\usepackage{bbm}
\usepackage{physics}
\usepackage[inline,shortlabels]{enumitem}
\usepackage[super]{nth}
\usepackage{titlesec}
\usepackage{cleveref}

\newtheorem{theorem}{Theorem}
\newtheorem{lemma}[theorem]{Lemma}
\newtheorem{proposition}[theorem]{Proposition}
\newtheorem{corollary}[theorem]{Corollary}

\newtheorem{definition}{Definition}

\newcommand{\ansatze}{ans\"{a}tze}           % anstaze
\newcommand{\ansatzes}{ans\"{a}tzes}         % ansatzes
\newcommand{\ConstE}{\mathrm{e}}             % For Euler's number
\newcommand{\ConstI}{\mathrm{i}}             % For the imaginary unit
\newcommand{\ConstPi}{\mathrm{\pi}}          % For pi
\newcommand{\NumE}{n_{\mathrm{e}}}           % Number of electrons
\newcommand{\PermB}{S^{\mathrm{b}}}          % Permutation of bits
\newcommand{\period}{.}                      % Deceive latexindent
\newcommand{\openone}{\mathbbm{1}}           % Blackboard 1
\newcommand{\SM}[1]{Supplementary #1}        % Supplementary *
\renewcommand{\vec}[1]{\bm{\mathrm{#1}}}     % Vector
\newcommand{\cf}[1]{(cf. \Cref{#1})}         % cf
\newcommand{\mth}[1]{#1\textsuperscript{th}} % nth, but for math expression
\newcommand{\conj}[1]{#1^{*}}                % Conjugate symbol
\newcommand{\orthc}[1]{#1^{\perp}}           % Orthogonal
\DeclareDocumentCommand{\E}{o m}{\mathbb{E}\IfValueT{#1}{_{#1}}\left[#2\right]} % expectation
\DeclareDocumentCommand{\Var}{o m}{\operatorname{Var}\IfValueT{#1}{_{#1}}\left(#2\right)} % variance
\DeclareDocumentCommand{\mso}{o m m}{\mathbb{E}\IfValueT{#1}{_{#1}}\left[#2^{\otimes #3,#3}\right]} % Moment SuperOperator

\colorlet{I1}{black}
\colorlet{I2}{red}
\colorlet{X1}{green!70!black}
\colorlet{X2}{blue}

\allowdisplaybreaks
 
\usepackage{geometry}
\begin{document}

\title{Towards determining the presence of barren plateaus in some chemically inspired variational quantum algorithms}

\author[1,2]{\fnm{Mao} \sur{Rui}}\email{maorui21b@ict.ac.cn}

\author[1,2]{\fnm{Tian} \sur{Guojing}}\email{tianguojing@ict.ac.cn}

\author*[1,2]{\fnm{Sun} \sur{Xiaoming}}\email{sunxiaoming@ict.ac.cn}

\affil[1]{\orgdiv{State Key Laboratory of Processors}, \orgname{Institute of Computing Technology, Chinese Academy of Sciences}, \orgaddress{\city{Beijing}, \postcode{100190}, \country{China}}}

\affil[2]{\orgdiv{School of Computer Science and Technology}, \orgname{University of Chinese Academy of Sciences}, \orgaddress{\city{Beijing}, \postcode{100049}, \country{China}}}
\abstract{
  In quantum chemistry, the variational quantum eigensolver (VQE) is a promising algorithm for molecular simulations on near-term quantum computers.
  However, VQEs using hardware-efficient circuits face scaling challenges due to the barren plateau problem.
  This raises the question of whether chemically inspired circuits from unitary coupled cluster (UCC) methods can avoid this issue.
  Here we provide theoretical evidence indicating they may not.
  By examining alternated dUCC \ansatzes{} and relaxed Trotterized UCC \ansatzes{}, we find that in the infinite depth limit, a separation occurs between particle-hole one- and two-body unitary operators.
  While one-body terms yield a polynomially concentrated energy landscape, adding two-body terms leads to exponential concentration.
  Numerical simulations support these findings, suggesting that popular 1-step Trotterized unitary coupled-cluster with singles and doubles (UCCSD) \ansatze{} may not scale.
  Our results emphasize the link between trainability and circuit expressiveness, raising doubts about VQEs' ability to surpass classical methods.
}
 
\maketitle

\section{Introduction}

In recent years, there has been significant interest in developing quantum algorithms to harness the capabilities of noisy intermediate-scale quantum (NISQ) devices \cite{preskill2018quantum}, hoping to solve classically intractable computational problems in the near term.
Computational chemistry is anticipated to be one of the first domains that benefit from such progress \cite{feynman1971feynman}.
A promising NISQ algorithm for molecular simulation in chemistry is the variational quantum eigensolver (VQE), aiming to find the ground energy of a given Hamiltonian.
VQE is a hybrid quantum-classical algorithm --- a quantum computer is used to prepare a parameterized trial state, an \ansatze{}, and perform measurements to get the expectation value of the given Hamiltonian, the cost function.
A classical computer is then used to iteratively train the parameters such that the cost function is minimized, according to the Rayleigh-Ritz variational principle.

While VQE has been successfully demonstrated for various small molecules \cite{peruzzo2014variational,kandala2017hardware,o2016scalable,nam2020ground,colless2018computation}, it faces several challenges.
From a practical point of view, the number of measurements required can be too large to afford, and the circuit depth is limited due to hardware noise.
More critically, the theoretical understanding of VQE is lacking, as it is a heuristic algorithm with no guarantee of producing a more accurate solution than classical methods.
Moreover, it has been found that for variational quantum algorithms, the ability to produce a better solution often comes at the cost of the barren plateau (BP) problem \cite{holmes2022connecting}, i.e., the gradients vanish exponentially concerning system size.
This means even if the global minimum is better than classical methods, VQE may fail to find it \cite{cerezo2021higher,arrasmith2021effect}, despite the efforts devoted to mitigating BP \cite{sack2022avoiding, grant2019initialization, zhang2022escaping, friedrich2022avoiding, binkowski2023barren, volkoff2021large, verdon2019learning, skolik2021layerwise, mastropietro2023fleming}.
It is unclear whether there exists VQE that are both more accurate than classical methods, or at least cannot be classically simulated, and trainable.

The \ansatzes{} employed in VQEs can be broadly categorized into three types \cite{fedorov2022vqe, barkoutsos2018quantum}: chemically inspired \ansatze{} \cite{peruzzo2014variational}, hardware-efficient \ansatze{} (HEA \cite{kandala2017hardware}), and Hamiltonian variational \ansatze{} (HVA \cite{wecker2015progress}).
The first category originates from well-established classical quantum chemistry methods \cite{szabo2012modern}, thereby offering higher accuracy.
The unitary coupled cluster (UCC) \ansatzes{} are a leading class of chemically inspired \ansatzes{} \cite{taube2006new,lee2018generalized,peruzzo2014variational}.
The second category is designed to be more compatible with NISQ devices, while the third category lies somewhere in between.
It is known that HEA \cite{mcclean2018barren} and HVA \cite{larocca2022diagnosing} in general suffer from BP.
There remains hope that chemically inspired \ansatzes{} can avoid BP, based on the intuition that the space explored by chemically inspired \ansatzes{} is restricted \cite{cerezo2021variational}.

Among chemically inspired \ansatzes{}, the alternated disentangled UCC (dUCC) offers the advantage over UCC of being provably able to express the exact FCI state using only single and double excitations as $k\to\infty$ \cite{evangelista2019exact}.
Several commonly discussed \ansatzes{} can be considered as examples of alternated dUCC \ansatzes{}, including $k$ products of unitary pair coupled cluster with generalized singles and doubles ($k$-UpCCGSD) \ansatze{} \cite{lee2018generalized}, basis rotation \ansatze{} (BRA \cite{google2020hartree}, since the Givens rotations are equivalent to single excitation rotations when acting on neighboring qubits), and the 1-step Trotterized unitary coupled-cluster with singles and doubles (UCCSD) \ansatze.

In this work, we contribute to the theoretical understanding of VQE by studying the trainability of a class of chemically inspired VQEs.
We focus on a common setting in VQE, where the Hamiltonian is chosen to be an electronic structure Hamiltonian, and the initial guess is chosen to be the uncorrelated Hartree-Fock state \cite{barkoutsos2018quantum}.
The \ansatze{} we study is a relaxed version of $k$-step Trotterized UCC \ansatzes{}, which we refer to as the alternated disentangled UCC (dUCC) \ansatze{}.
By ``relaxed'' we mean that the parameters across the $k$ alternations in the $k$-steps Trotterized UCC \ansatzes{} become independent.
We proved, under the assumption of $k \to \infty$ and the presence of
sufficient single excitation rotations connecting all the qubits,
the following results:
\begin{enumerate}
  \item
        If the \ansatze{} comprises solely single excitation rotations, the cost function concentrates around its mean polynomially concerning the qubit number $n$.
  \item
        If the \ansatze{} incorporates additional double excitation rotations, the concentration of the cost function scales inversely with $\binom{n}{\NumE}$, where $\NumE$ represents the number of electrons.
\end{enumerate}
Here, the mean and variance are defined in the sense of random parameter initialization.
It is well known that BP has a close relationship with cost concentration \cite{arrasmith2022equivalence}.
For our specific setting, we also prove a quantitative relationship between them, which indicates that cost concentration implies BP.
Furthermore, we conducted numerical simulations to explore the scenario where $k$ is small, which is more practical.
The findings indicate that the relative error between the cost variance for finite $k$ and its asymptotic value decreases exponentially as $k$ increases, suggesting that the assumption $k \to \infty$ can be significantly relaxed.
For $k$ products of UCCSD ($k$-UCCSD), our predictions are quite accurate even at $k = 2$ for qubit numbers ranging from 4 to 24.
When $k = 1$, the variance of the cost function also exhibits an exponential decrease as the number of qubits grows.
These numerical results show that our results are instructive for practical applications.
We also study the ``qubit'' version of alternated dUCC \ansatzes{} in \SM{Information}, where all Pauli Z terms are trimmed after the Jordan-Wigner transformation of excitation operators \cite{tang2021qubit}.
Remark that the Givens rotations can be viewed as qubit single excitation rotations \cite{google2020hartree}.
The most important finding is that the qubit version of $k$ products of unitary coupled cluster with singles ($k$-UCCS) \ansatze is exactly the same as $k$-UCCSD in terms of cost variance, which may indicate the connection between qubit UCCS and standard UCCSD.
Our results can be interpreted in two aspects.
Firstly, it is known that \ansatzes{} composed of single excitation rotations only can be simulated classically \cite{thouless1960stability} while \ansatzes{} composed of single and double excitation rotations cannot \cite{mcardle2020quantum}.
Hence, our results question the hope that VQE may surpass classical methods.
Secondly, our results show the trade-off between trainability and expressibility for VQE, {aligning with the findings that expressibility can induce BP} \cite{holmes2022connecting,larocca2022diagnosing,ragone2023unified,fontana2023adjoint,diaz2023showcasing}, since single excitation rotations only explore a polynomial space while single and double excitations together explore a $\binom{n}{\NumE}$ one.
\section{Results}

\subsection{Notations}

In this work, we study the scalability of alternated disentangled UCC (dUCC) \ansatzes{}, which can be viewed as a relaxed version of Trotterized UCC.
Recall that in UCC theory
\cite{taube2006new,lee2018generalized,peruzzo2014variational}, the trial wave
function is parameterized as the unitary exponentiation of excitation operators
$\hat{T}(\vec{\uptheta}) = \sum_{ia} \theta_{ia}\hat{a}_{a}^{\dagger}\hat{a}_{i}
  + \sum_{ijab}
  \theta_{ijab}\hat{a}_{a}^{\dagger}\hat{a}_{b}^{\dagger}\hat{a}_{i}\hat{a}_{j} +
  \dots$ acting upon a reference state $\ket{\psi_0}$:
\begin{equation}
  \label{eq:0nty} \ket{\psi_{\text{UCC}}(\vec{\uptheta})} =
  \ConstE^{\hat{T}(\vec{\uptheta})-\hat{T}^{\dagger}(\vec{\uptheta})} \ket{\psi_0}.
\end{equation}
Here, $\hat{a}$ ($\hat{a}^{\dagger}$) represents the annihilation (creation) operator \cite{szabo2012modern}, and we use the convention that $i,j,\dots$ ($a,b,\dots$) label the occupied (virtual) orbitals.
Generalized excitations without the occupation constraint are also allowed, and we will use $p,q,r,s,\dots$ to label any orbital to emphasize the difference.
The widely used UCCSD variant \cite{peruzzo2014variational, barkoutsos2018quantum} corresponds to truncated $\hat{T}$ up to second excitations.
For circuit implementation, it is unclear how to exactly implement the large exponentiation in \Cref{eq:0nty} efficiently (i.e., to a polynomial number of one and two-qubit gates), except for the case when there are only single excitations.
The canonical approach is to take a $k$-steps Trotter approximation, with systematic error up to $O(k^{-1})$.
In contrast, for the ease of our analysis, we will incorporate an additional
relaxation step after Trotterization, as described below:
\begin{equation}
  \label{eq:nzaj} \ConstE^{\sum _{j=1}^{m} \theta_{j} \qty(\hat{\tau }_{j} -\hat{\tau
    }_{j}^{\dagger })} \xrightarrow{\text{Trotter}}\prod _{i=1}^{k}\prod _{j=1}^{m}
  \ConstE^{\frac{\theta _{j}}{k}\left(\hat{\tau }_{j} -\hat{\tau }_{j}^{\dagger }\right)}
  \xrightarrow{\text{Relax}}\prod _{i=1}^{k}\prod _{j=1}^{m} \ConstE^{\theta _{j}^{(
      i)}\left(\hat{\tau }_{j} -\hat{\tau }_{j}^{ \dagger }\right)},
\end{equation}
where
$\hat{\tau}_j\in\qty{\hat{a}_{p}^{\dagger}\hat{a}_{q},\hat{a}_{p}^{\dagger}
    \hat{a}_{q}^{\dagger}\hat{a}_{r}\hat{a}_{s},\dots}$.
Remark that $\theta_j$ and $\theta_j^{(i)}$ are different sets of parameters, and we implicitly make the assumption that all parameters are real (this is valid if we only care about real wave functions).
Obviously, such relaxation can offer higher variational accuracy, at the cost of higher training overhead.
We refer to the last unitary in \Cref{eq:nzaj} as alternated dUCC \ansatze{}, which is formally defined below.
The term ``dUCC'' \cite{evangelista2019exact} was used to refer to any sequence of excitation rotations (similar to a 1-step Trotterization but conceptually different), and was proved to be able to express the exact FCI state in the limit $k\to\infty$ even when only single and double excitations are involved.
We add the ``alternated'' prefix to emphasize the alternative structure.
To avoid bothering with the fermionic ladder operators $\hat{a}, \hat{a}^{\dagger}$, we will identify $\hat{a}_p,\hat{a}^{\dagger}_p$ with $Q_p \prod_{a<p} Z_a, Q^{\dagger}_p \prod_{a<p} Z_a$ respectively, by the Jordan-Wigner transformation.
The Jordan-Wigner transformation is also employed in the canonical implementation.
Other transformations such as the Bravyi-Kitaev's may be investigated in future work.
Here, we define $Q = \op{0}{1}$.
A ``qubit'' version of UCC theory, where the $Z$-string is eliminated after the Jordan-Wigner transformation, is also considered in this work.

\begin{definition}[Alternated (qubit) dUCC \ansatze{}]\label{def:62cp}
  Call the unitary $U^{\vec{R}}_{k}(\vec{\uptheta})$ ($k\in\mathbb{N}_+$) an
  alternated (qubit) dUCC \ansatze{}, if it can be written as
  \begin{equation}
    \label{eq:t2ru}
    U^{\vec{R}}_{k}(\vec{\uptheta})=\prod_{i=1}^{k}\prod_{j=1}^{m}R_j(\theta^{(i)}_j),
  \end{equation}
  where $\vec{R}=(R_1,\dots,R_m)$ is a sequence of (qubit) excitation rotations, $R_j(\theta)=\exp(\theta(\hat{\tau}_j-\hat{\tau}_j^\dagger))$ and $\hat{\tau}_j\in\left\{\hat{a}_{p}^{\dagger}\hat{a}_{q},\hat{a}_{p}^{\dagger} \hat{a}_{q}^{\dagger}\hat{a}_{r}\hat{a}_{s},\dots\right\}$ (or $\hat{\tau}_j\in\left\{Q_p^\dagger Q_q,Q_p^\dagger Q_q^\dagger Q_r Q_s,\dots\right\}$ in qubit version).
\end{definition}

While \Cref{def:62cp} encompasses any (qubit) excitations, our primary focus lies on single and double (qubit) excitations.
This is particularly relevant in scenarios like UCCSD.
To this end, we will introduce dedicated notation for single and double (qubit) excitation rotations, as follows.
\begin{align}
  A_{pq}(\theta )                 & = \exp(\theta (Q^{\dagger}_p Q_q - h.c.)
  \prod_{a=q+1}^{p-1}
  Z_a),                                                                                        \\
  B_{pqrs}(\theta )               & = \exp(\theta (Q^{\dagger}_p Q^{\dagger}_q Q_r Q_s - h.c.)
  \prod_{b=s+1}^{r-1}
  Z_b \prod_{a=q+1}^{p-1} Z_a),                                                                \\
  A_{pq}^{\text{qubit}}(\theta)   & = \exp(\theta (Q^{\dagger}_p Q_q - h.c.)
  ),
  \\
  B_{pqrs}^{\text{qubit}}(\theta) & = \exp(\theta (Q^{\dagger}_p Q^{\dagger}_q
    Q_r Q_s - h.c.)
  ),
\end{align}
where ``h.c.'' stands for Hermitian conjugation, and $p>q>r>s$.
Examples of alternated (qubit) dUCC \ansatzes{} (including at most double (qubit)
excitations) are

\begin{align}
  k\!
  \operatorname{-UCCS}(\vec{\uptheta})         & = \prod _{i=1}^{k}\prod _{p >\NumE \geqslant q}
  A_{pq}\left(\theta _{pq}^{( i)}\right), \label{eq:diqc}                                                                                           \\
  k\!
  \operatorname{-UCCGS}(\vec{\uptheta})        & = \prod _{i=1}^{k}\prod _{p >q}
  A_{pq}\left(\theta _{pq}^{( i)}\right),                                                                                                           \\
  k\!
  \operatorname{-UCCSD}(\vec{\uptheta})        & = \prod _{i=1}^{k}\prod _{p >\NumE \geqslant q}
  A_{pq}\left(\theta _{pq}^{( i)}\right)\prod _{p >q >\NumE \geqslant r >s} B_{pqrs}\left(\theta _{pqrs}^{( i)}\right),                             \\
  k\!
  \operatorname{-UCCGSD}(\vec{\uptheta})       & = \prod
  _{i=1}^{k}\prod _{p >q}
  A_{pq}\left(\theta _{pq}^{( i)}\right)\prod _{p >q >r >s} B_{pqrs}\left(\theta _{pqrs}^{( i)}\right),                                             \\
  k\!
  \operatorname{-qubit-UCCS}(\vec{\uptheta})   & = \prod _{i=1}^{k}\prod _{p
    >\NumE\geqslant q}
  A^{\text{qubit}}_{pq}\left(\theta _{pq}^{( i)}\right),                                                                                            \\
  k\!
  \operatorname{-qubit-UCCGS}(\vec{\uptheta})  & = \prod_{i=1}^{k}\prod _{p >q}
  A^{\text{qubit}}_{pq}\left(\theta _{pq}^{( i)}\right),                                                                                            \\
  k\!
  \operatorname{-qubit-UCCSD}(\vec{\uptheta})  & = \prod_{i=1}^{k}\prod _{p >\NumE \geqslant q}
  A^{\text{qubit}}_{pq}\left(\theta_{pq}^{( i)}\right)\prod _{p >q >\NumE \geqslant r >s}B^{\text{qubit}}_{pqrs}\left(\theta _{pqrs}^{( i)}\right), \\
  k\!
  \operatorname{-qubit-UCCGSD}(\vec{\uptheta}) & = \prod_{i=1}^{k}\prod _{p >q}
  A^{\text{qubit}}_{pq}\left(\theta _{pq}^{(i)}\right)\prod _{p >q >r >s} B^{\text{qubit}}_{pqrs}\left(\theta _{pqrs}^{(i)}\right),                 \\
  k\!
  \operatorname{-BRA}(\vec{\uptheta})          & = \prod_{i=1}^{k}\prod_{p=1}^{n-1}
  A_{p,p+1}\left(\theta_{p,p+1}^{(i)}\right),                                                                                                       \\
  k\!
  \operatorname{-UpCCGSD}(\vec{\uptheta})      & = \prod_{i=1}^{k}\prod _{p >q}
  A_{pq}\left(\theta _{pq}^{( i)}\right)\prod _{a >b} B_{2a,2a-1,2b,2b-1}\left(\theta _{ab}^{( i)}\right).
  \label{eq:vmj4}
\end{align}
Here the `G' in the name of \ansatzes stands for ``generalized''.
Notice that in \Crefrange{eq:diqc}{eq:vmj4}, the ordering of products of excitations within each block $k$ is not specified.
The reason is two-fold.
Firstly, there exists no established convention regarding the ordering of these products, and their optimality has only been investigated numerically \cite{grimsley2019trotterized}.
Secondly, our results, focusing on the $k \to \infty$ regime, remain unaffected by the specific ordering of excitations.

In VQE, the cost function to be optimized is naturally chosen to be the
energy of the molecule system:
\begin{equation}
  \label{eq:2apd}
  C(\vec{\uptheta};\rho,U,O) =\tr( O U(\vec{\uptheta})\rho U(\vec{\uptheta})^{\dagger }).
\end{equation}
Here,
\begin{itemize}
  \item
        $\rho$ is an easy-to-prepare reference state.
        In this work, $\rho$ is fixed to be $\op{\psi_0}$, where $\ket{\psi_0}$ is a
        Hartree-Fock state \cite{dallaire2019low},
        \begin{equation}
          \label{eq:gkth} \ket{\psi_{0}} :=|\underbrace{1\dots 1}_{\NumE}\underbrace{0\dots
            0}_{n-\NumE} \rangle.
        \end{equation}
  \item
        $U(\vec{\uptheta})$ is the \ansatze{}, such that $U(\vec{\uptheta})\ket{\psi_0}$ produces a trial state.
  \item
        $O$ is the observable, usually the electronic structure Hamiltonian
        $H_{\mathrm{el}}=\sum_{pq} h_{pq}\hat{a}_{p}^{\dagger }\hat{a}_{q} +\sum_{pqrs}
          g_{pqrs}\hat{a}_{p}^{\dagger }\hat{a}_{q}^{\dagger }\hat{a}_{r}\hat{a}_{s}$.
        In our study, we focus on the real case, where the one- and two-electron integrals $h_{pq}$ and $g_{pqrs}$ are assumed to be both real and symmetric ($h_{pq}=h_{qp}\in\mathbb{R}$ and $g_{pqrs}=g_{srqp}\in\mathbb{R}$).
        Moreover, we assume that $g_{pqrs}\neq 0$ only if $p>q>r>s$ or $p<q<r<s$.
        In other words, terms such as $\hat{a}_{1}^{\dagger }\hat{a}_{3}^{\dagger }\hat{a}_{2}\hat{a}_{4}$ are forbidden.
        It is important to note that this assumption is made for ease of analysis, rather than due to the drawback of our proof techniques.
        In fact, our results can be generalized to any observables, although the progress of extending the analysis could be tedious.
        Remark that under the aforementioned assumptions, $H_{\mathrm{el}}$ can be
        written as:
        \begin{equation}
          \label{eq:lava}
          H_{\mathrm{el}} =\sum _{p >q} h_{pq}\left(\hat{a}_{p}^{\dagger }\hat{a}_{q} +h.c.
          \right)
          +\sum _{p >q >r >s} g_{pqrs}\qty(\hat{a}_{p}^{\dagger }\hat{a}_{q}^
          {\dagger }\hat{a}_{r}\hat{a}_{s} +h.c.).
        \end{equation}

\end{itemize}

The cost function we study may take various forms, with $\rho=\op{\psi_0}$ fixed, and $U(\vec{\uptheta}),O$ potentially varying in different instances.
For simplicity, we will slightly abuse the notation and use $C(\vec{\uptheta};U,O),C(\vec{\uptheta};U),C(\vec{\uptheta};O)$, or even $C(\vec{\uptheta})$ instead of $C(\vec{\uptheta};\rho,U,O)$ whenever the context is clear.
\subsection{Cost concentration and barren plateau}

While VQA and VQE are promising in the NISQ era, there are several unresolved obstacles, among which the Barren Plateaus (BP) phenomenon poses a major concern \cite{cerezo2021variational}.
Analogous to the gradient vanishing problem in training classical neural networks, BP refers to the exponential concentration of gradients (of the cost function in VQA) around 0 over random parameters, thus ruling out any gradient-based optimization method with a random starting point.
Another related concept is the concentration of the cost function, which also characterizes the flatness of the cost landscape.
For example, the exponential concentration of the cost function around its mean would rule out any gradient-free optimization method with a random starting point.
The two types of concentration are strictly defined as follows.
Let $\left\{C_n\right\}_{n\in\mathbb{N}_+}$ be a family of cost functions indexed by qubit number $n$ such that $\E[\vec{\uptheta}]{\partial_{\theta_j}C(\vec{\uptheta})}=0$ for all $\theta_j\in\vec{\uptheta}$.

\begin{definition}[Gradient concentration and BP]
  Let $G(n)=\max_{\theta_j\in\vec{\uptheta}} \Var[\vec{\uptheta}]{\partial_{\theta_j}
      C_n(\vec{\uptheta})}$.
  We say the gradients of $\left\{C_n\right\}_{n\in\mathbb{N}_+}$ concentrate polynomially if $G(n)=1/\operatorname{poly}(n)$, and concentrate exponentially if $G(n)=1/\exp(\Omega(n))$.
  Specifically, $\left\{C_n\right\}_{n\in\mathbb{N}_+}$ exhibits BP if the gradients concentrate exponentially.
\end{definition}

\begin{definition}[Cost concentration]
  Let $F(n)=\Var[\vec{\uptheta}]{C_n(\vec{\uptheta})}$.
  We say $\left\{C_n\right\}_{n\in\mathbb{N}_+}$ concentrates polynomially if $F(n)=1/\operatorname{poly}(n)$, and concentrates exponentially if $F(n)=1/\exp(\Omega(n))$.
\end{definition}

Intuitively, both gradient concentration and cost concentration describe the flatness of the landscape.
The equivalence between exponential cost concentration and exponential gradient concentration (i.e., BP) was established for \ansatzes{} comprising a polynomial number of independent rotations whose generator is Hermitian with eigenvalue $\pm 1$, using integral arguments and the parameter shift rule \cite{arrasmith2022equivalence}.
Another recent work \cite{miao2024equivalence} also proves such equivalence using an elementary derivation by considering the Riemannian gradient instead of the Euclidean gradient.
They also assume each parameterized gate to explore a whole unitary group.

Since these results are not directly applicable in our setting, in the following lemma we derive a quantitative relationship between cost variance and gradient variances, using the fact that the cost function of alternated dUCC \ansatze{} is periodic and has no rapid oscillations.
The proof of this lemma is postponed to \SM{Note 2}, and is largely inspired by another work \cite{napp2022quantifying}.

\begin{lemma}[Relationship between variances of cost and gradients]
  \label{lem:mmg3}
  Let $U^{\vec{R}}_{k}(\vec{\uptheta})$ be an alternated (qubit) dUCC \ansatze{}
  defined in \Cref{eq:t2ru}, and $C(\vec{\uptheta};U^{\vec{R}}_{k})$ defined in
  \Cref{eq:lava}, then
  \begin{equation}
    \label{eq:jlre}
    (k\abs{\vec{R}})^{-1}\cdot\Var[\vec{\uptheta}]{C}
    \le\max_{\theta_j\in\vec{\uptheta}}\Var[\vec{\uptheta}]{\partial_{\theta_{j}}
      C} \le \Var[\vec{\uptheta}]{C}.
  \end{equation}
\end{lemma}

Consequently, for alternated dUCC \ansatzes{}, the exponential decay of cost variance implies exponential decay of gradient variance by second inequality of \Cref{eq:jlre}, while the reverse holds only when $k\abs{\vec{R}}$ is sub-exponential by the first inequality.
Our main focus will be determining the scaling of variance of cost, and bound the variance of gradients by \Cref{lem:mmg3}.
\subsection{Separation between single and double excitation rotations}

In this section, we lay out our main result, proving the scaling of cost concentration for alternated dUCC \ansatzes{} composed of (qubit) single and double excitation rotations.
The detailed proofs can be found in \SM{Information}.

\begin{theorem}[Main result]\label{thm:yqn9}
  Let $U^{\vec{R}}_{k}(\vec{\uptheta})$ be an alternated (qubit) dUCC \ansatze{} defined in \Cref{eq:t2ru}, and $C(\vec{\uptheta};U^{\vec{R}}_{k},H_{\mathrm{el}})$ defined in \Cref{eq:2apd}, where the qubit number is $n$, and $H_{\mathrm{el}}$ is defined in \Cref{eq:lava}, with one- and two-electron integrals $h_{pq} ,g_{pqrs}$.
  Denote the simple undirected graph formed by index pairs of (qubit) single
  excitations in $\vec{R}$ by $G$, i.e.,
  \begin{equation}
    G:=(V=[n],E=\left\{(u,v)\middle|A_{uv}\text{ or }A^{\mathrm{qubit}}_{uv}\in \vec{R}\right\}).
  \end{equation}
  Suppose $G$ is connected.
  The limit $\lim _{k\rightarrow \infty }\Var{C\qty(\vec{\uptheta}
    ;U_{k}^{\vec{R}}),H_{\mathrm{el}}} =:V( n)$ exists, and
  \begin{enumerate}
    \item
          \label{itm:c1o6}

          If $\vec{R}$ contains only single excitation rotations, then $V( n) \sim \operatorname{poly}(n^{-1},\NumE,\norm{\vec{h}},\norm{\vec{g}})$.

    \item
          \label{itm:9q5r}

          If $\vec{R}$ contains both single and double excitation rotations, then $V( n) \sim \operatorname{poly} \left(\norm{\vec{h}},\norm{\vec{g}}\right)/\binom{n}{\NumE}$.

    \item
          \label{itm:rk5h}

          If $\vec{R}$ contains only single qubit excitation rotations, then $V( n) \sim \operatorname{poly}(n^{-1},\NumE,\norm{\vec{h}},\norm{\vec{g}})$ if the maximum degree of $G$ is 2, and $V( n) \sim \operatorname{poly} \left(\norm{\vec{h}},\norm{\vec{g}}\right)/\binom{n}{\NumE}$ otherwise.

    \item
          \label{itm:ipnl}

          If $\vec{R}$ contains both single and double qubit excitation rotations (and satisfies the condition in \SM{Theorem 37}), then $V( n) \sim \operatorname{poly} \left(\norm{\vec{h}},\norm{\vec{g}}\right)/\binom{n}{\NumE}$.

  \end{enumerate}
\end{theorem}

Combining \Cref{lem:mmg3} and \Cref{thm:yqn9}, we have the following corollary.

\begin{corollary}
  \label{cor:74ej}
  There exists finite function $K(n)$, such that the following alternated (qubit) dUCC \ansatzes{} suffer from BP when the number of alternations $k > K(n)$: $k$-UCC(G)SD, $k$-UpCCGSD, $k$-qubit-UCC(G)S, $k$-qubit-UCC(G)SD.
\end{corollary}

We make three remarks regarding the main result.

Firstly, the separation between cost concentration can be understood as the result of a difference in expressiveness.
By Thouless theorem \cite{thouless1960stability}, an alternated dUCC \ansatze{}
composed of single excitation rotations can be identified with a single
$2^n\times 2^n$ unitary
\begin{equation}
  \label{eq:nshi} \exp(\sum_{i,j=1}^{n}
  h_{ij} (\hat{a}_{p}^{\dagger }\hat{a}_{q} -h.c.)
  ),\qq{with} h\in \mathrm{O}(n).
\end{equation}
Hence, if there are only single excitation rotations, the \ansatze{} will only explore an $O(n^2)$-dimensional subspace.
On the other hand, if there are both single and double excitation rotations, the \ansatze{} will explore the real invariant space of dimension $\binom{n}{\NumE}$, spanned by all computational basis states with Hamming weight $\NumE$ \cite{evangelista2019exact}.
The cost concentration of the alternated qubit dUCC \ansatzes{} composed of single qubit excitation rotations can be linked with the classical simulatability of matchgates on different topologies.
Observe that qubit excitation rotations, or equivalently the Grover rotations, are matchgates.
It is well-known \cite{jozsa2008matchgates} that circuits composed of matchgates acting on nearest neighbors can be classically efficiently simulated, in a way similar to \Cref{eq:nshi}, and circuits composed of matchgates that act on nearest neighbors and next-nearest neighbors can perform universal quantum computation.
This fact corresponds to \Cref{itm:rk5h} of \Cref{thm:yqn9}.
A recent work \cite{cerezo2023does} discussed the seemingly deep connection between BP and classical simulatability, by examining existing BP results.
Our results provide another example supporting their arguments.

Secondly, our results can be partly explained by the dynamic Lie algebra (DLA), which is used in recent works \cite{larocca2022diagnosing,ragone2023unified,fontana2023adjoint,diaz2023showcasing} to establish a unified theory of the barren plateau of deep periodic \ansatzes{}.
Notice that the alternated dUCC \ansatzes{} can also be seen as periodic \ansatzes{}.
When there are both single and double (qubit) excitations, as in \Cref{itm:9q5r} of \Cref{thm:yqn9}, the DLA is $\mathfrak{g}=\mathfrak{so}\left(\binom{n}{\NumE}\right)$ when restricted to the invariant subspace.
Hence, the cost variance can also be obtained by integrating over $\ConstE^{\mathfrak{g}}=\mathrm{SO}\left(\binom{n}{\NumE}\right)$ \cite{collins2009some}.
However, other cases are more involved.
It is important to note that these results \cite{larocca2022diagnosing,ragone2023unified,fontana2023adjoint,diaz2023showcasing} cannot be directly applied in our setting for different reasons, such as $\mathfrak{g}$ being real, $\ConstI\rho,\ConstI O\notin\mathfrak{g}$, or the \ansatze{} considered is a specific one that differs from ours.
Moreover, our method is different from theirs, relying more on symmetries and taking a combinatorial approach.
Another work \cite{monbroussou2023trainability}, although their techniques do not fall into the DLA category, also studies the \ansatzes{} composed of single excitations.
Their results differ from ours into they assumed that the input or output of the \ansatzes{} are Haar distributed.
In contrast, our results are predicated on a fixed input state, the Hartree-Fock state, with assumptions made regarding the structure of the \ansatze{} rather than focusing on the output.

Finally, to obtain the main results, we made three assumptions: (1) the number of alternations $k \to \infty$; (2) independence of parameters between different Trotter steps \cf{eq:nzaj}; and (3) random initialization of parameters.
These assumptions greatly simplify the analysis but at the same time deviate from practical applications.
In Trotterized UCCSD, parameters between different Trotter steps are coupled, while in $k$-UpCCGSD the parameters are independent.
In both cases, the number of Trotter steps or alternations $k$ is finite.
Additionally, the parameters in VQE are almost always initialized in all-zero.
Regrettably, our techniques currently face challenges in fully addressing the limitations associated with these three assumptions.
Nevertheless, we hope that our results are also instructive for practical applications.
In fact, the convergence of cost variance towards its asymptotic value is equivalent to the convergence of alternating projections \cite{halperin1962product}, which is exponential in $k$.
Such exponential convergence is also verified through numerical experiments in the ``Simulations'' subsection in the Results.
In that sense, the assumption that $k \to \infty$ can be greatly relaxed, except that we failed to find a way to determine the exact convergence rate, for example the $K(n)$ in \Cref{cor:74ej}, through our current techniques.
Recently, a step was made in calculating the mixing time to $t$-designs, which is analogous to $k$, determining that the time is $\mathrm{poly}(n)$ rather than some finite value \cite{fontana2023adjoint}.
Their techniques may prove helpful in improving our results.
In our work, the convergence rate is numerically studied in the ``Simulations'' subsection in the Results, where the results hint that the $K(n)$ in \Cref{cor:74ej} could even be a constant as small as 1 for the case of $k$-UCCSD.
This finding further mitigates the gap between our theoretical results and practical applications.
The implications and caveats of assumptions (2) and (3) will be discussed in ``Discussion'' subsection.
\subsection{Simulations}

We conduct numerical experiments to validate our theoretical results and explore the more practical scenarios where $k$ is finite.

First, we verify the prediction in \Cref{thm:yqn9} that the variance of the cost function converges as $k$ increases and investigate its convergence rate.
In \Cref{fig:t98b}, we present experimental data for 4 alternated dUCC \ansatzes{}: $k$-BRA, $k$-UCCS, $k$-UCCSD and $k$-UpCCGSD.
For each \ansatze{}, we plot the scaling of the following distances concerning $k \in \{1, 2, \dots, 100\}$ for different values of $n \in \{4, 6, 8, 10\}$:
\begin{enumerate*}[(1)]
  \item
        The $\ell^1,\ell^2,\ell^\infty$ distance between $\left(\prod _{j=1}^{m}\mso{R_j}{2}\right)^{k}\ket{\psi _{0}}^{\otimes 2}$ and $P_M\ket{\psi _{0}}^{\otimes 2}$.
  \item
        The distances (relative errors) between cost variance at $k$ and the asymptotic variance $\lim_{k\to\infty}\Var{C}$, with respect to monomial observables $\hat{a}_p^{\dagger} \hat{a}_q + h.c\period$ ($p > q$) and $\hat{a}_p^{\dagger} \hat{a}_q^{\dagger} \hat{a}_r \hat{a}_s + h.c\period$ ($p > q > r > s$).
\end{enumerate*}
In all cases, we set $\NumE = n/2$.
This choice is reasonable since in practice $\NumE = \Theta(n)$ \cite{cccbdb}.
In each alternation of these \ansatzes{}, the excitation rotations are arranged in lexicographical order according to their position, and when there are both single and double excitation rotations, the single excitation rotations are placed before the double ones.
It is noted that all data are accurate within machine precision, as all calculations are based on the formulas derived in the \SM{Note 6 and 8}.
Due to the complexity of computations, the maximum value for $n$ is limited to 10.
The results indicate that all distances decay to zero as $k$ increases, until exhausting machine precision, which confirms the correctness of our theoretical results.
Moreover, all distances converge exponentially concerning $k$, and for increasing $n$, it appears that the convergence rates decrease moderately for $k$-BRA and $k$-UCCS, and slightly increase for $k$-UCCSD and $k$-UpCCGSD.

Based on the results in \Cref{fig:t98b}, it is appealing to conjecture that $k$-UCCSD and $k$-UpCCGSD exhibit BP even for constant $k$.
To support this conjecture, we further investigate how the variance of the cost function for $k$-UCCSD changes with increasing $n$ when $k$ is small.
In \Cref{fig:27ba}, we plot the scaling of the cost variance concerning $n \in \{4, 6, \dots, 24\}$ when $k \in \{1, 2, 3, \infty\}$.
We only display the results for two fixed observables $\hat{a}^{\dagger}_2 \hat{a}_1 + h.c\period$ and $\hat{a}^{\dagger}_4 \hat{a}^{\dagger}_3 \hat{a}_2 \hat{a}_1 + h.c\period$ because, according to previous numerical and theoretical results, the cost variances corresponding to different one- (or two-) body operators are similar.
Additionally, the number of electrons is still set to $\NumE = n/2$.
Unlike Figure 1, the cost variances for $k = 1, 2, 3$ are estimated by randomly sampling 6000 times for each different $n$.
The asymptotic variance of the cost function is calculated using the formulas derived in the \SM{Note 8}.
We also fit the estimated data points with an exponential function $a \exp (b \cdot n)$, where $a$ and $b$ are fitting parameters.
The data points and the fitted curves match well, justifying our conjecture.
Moreover, the variances converge to their asymptotic values quickly.
A very small value of $k = 2$ or $3$ seems to be enough for the relative error between finite-$k$ variance and asymptotic variance to be bounded by a constant.
In \SM{Note 9}, we present numerical results for $\NumE = n/4$, which is similar to the case of $\NumE = n/2$.
This shows that our conclusion on $k$-UCCSD is not affected by the number of electrons.

In summary, these numerical results reduce the distance between our theoretical results and reality.
Specifically:
\begin{enumerate}
  \item
        The unrealistic assumption of $k \to \infty$ can be greatly relaxed.
        \Cref{fig:t98b} shows that the relative error between the variance of the cost function with finite $k$ and the asymptotic value decays exponentially with $k$, and the convergence rate does not degrade much with increasing qubit count.
  \item
        Our theoretical results regarding  $k$-UCCSD seem to extend to cases when $k$ is constant.
        \Cref{fig:27ba} shows that for $k$-UCCSD, our predictions are already accurate enough when $k = 2$, and when $k = 1$, the variance of the cost function also exhibits exponential decay with qubit count.
        This implies that our theory is also instructive for \ansatzes{} like 1-step Trotterized UCCSD that are used in practice.
\end{enumerate}

It should be noted that these conclusions are merely observations derived from numerical results and are limited by the scale of these numerical experiments due to computational resources.
Unfortunately, our current technology cannot provide more theoretical evidence to support our views.
We leave more theoretical analysis for future work.

\begin{figure*}[ht]
  \centering
  \includegraphics{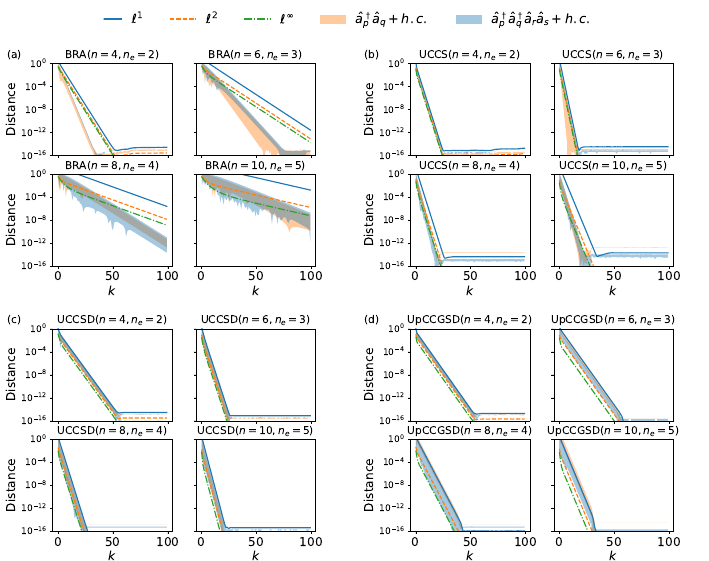}
  \caption{
    \textbf{Convergence of the projected state and cost variance for four alternated dUCC \ansatzes{}.}
    The four panels correspond to four \ansatzes{}: (a) basis rotation \ansatze (BRA), (b) unitary coupled cluster with singles (UCCS), (c) unitary coupled cluster with singles and doubles (UCCSD), and (d) unitary pair coupled cluster with generalized singles and doubles (UpCCGSD).
    For each \ansatze{}, we provide four plots corresponding to number of orbitals $n = 4,6,8,10$ and number of electrons $\NumE = n/2$.
    The number of alternations $k$ takes value in $\{1, 2, \dots, 100\}$.
    The y-axes are the dimensionless distances between the finite-alternations quantities and the infinite-alternations quantities ($k \to \infty$).
    The blue solid, orange dashed, and green dash-dotted curves represent the $\ell^1, \ell^2, \ell^{\infty}$ distance, respectively.
    The shaded areas in orange and blue encompass all single and double excitation monomials.
    The abbreviation ``h.c\period'' stands for ``Hermitian conjugate''.
  }
  \label{fig:t98b}
\end{figure*}

\begin{figure*}[ht]
  \centering
  \includegraphics{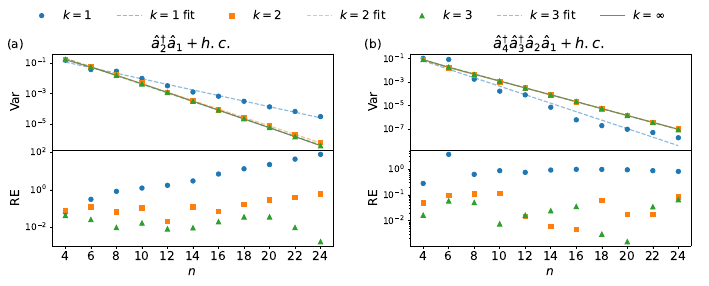}
  \caption{
    \textbf{Scaling of the cost variance and relative error to asymptotic variance for $k$-UCCSD when $\NumE = n/2$.}
    Here we present the cost variance (Var) and the relative error to asymptotic variance (RE) for the $k$ alternations of unitary coupled cluster with singles and doubles ($k$-UCCSD) \ansatze measured by (a) $\hat{a}^{\dagger}_2 \hat{a}_1 + h.c\period$ and (b) $\hat{a}^{\dagger}_4 \hat{a}^{\dagger}_3 \hat{a}_2 \hat{a}_1 + h.c\period$.
    The abbreviation ``h.c\period'' stands for ``Hermitian conjugate''.
    The number of qubit is $n \in \{4, 6, \dots, 24\}$, with the number of electrons given by $\NumE = n/2$.
    For each observable, the illustrated variances for $k = 1, 2, 3$ are estimated from 6000 random samples at each $n$ value, and are represented by blue circles, orange squares, and green triangles, respectively.
    The variances for $k = 1, 2, 3$ are fitted using the function $a \exp(b \cdot n)$, with the resulting fitting curves shown as blue, orange, and green dashed lines.
    The computed asymptotic variances ($k \to \infty$) are depicted as gray solid curves.
    In the lower half of each panel, we display the relative error of the variances at $k = 1, 2, 3$ compared to the asymptotic variances, also represented by blue circles, orange squares, and green triangles.
    The error bars, estimated as $\sqrt{2/(6000-1)} \cdot \mathrm{Var}$, are not shown as they are smaller than the data points.
  }
  \label{fig:27ba}
\end{figure*}
\section{Discussion}

In this work, we conducted a rigorous analysis of cost concentration for the alternated dUCC \ansatze{} at infinite depth.
Contrary to previous belief \cite{cerezo2021variational}, we showed that the chemically inspired \ansatze{} can also suffer from BP.
Specifically, our findings indicate that while double excitations are essential for accuracy \cite{szabo2012modern}, they can significantly increase the training cost.
Numerical results suggest that our results can be extended to the finite-depth cases, and potentially even a constant number of alternations for interesting \ansatzes{} such as $k$-UCCSD.

While we made specific assumptions regarding the domain of parameters and the form of the electronic structure Hamiltonian to derive our main results, our proof techniques can be adapted to a more general scenario without these assumptions, albeit with increased complexity.
In fact, we anticipate the applicability of our proof techniques to analyze the concentration and BP of other periodic cost functions.
Furthermore, our analysis of cost concentration involves characterizations of the \mth{$t$} moment of any excitation rotations for any $t$, which may be of independent interest.

We acknowledge that our work is at a distance from ruling out the practicality of chemically inspired \ansatzes{}, in the following aspects.

Firstly, our study focused on alternated dUCC \ansatzes{}, incorporating a large number of independent parameters.
The Trotterized UCCSD may be free of BP because its correlated parameters are limited in a smaller range with increasing Trotter steps \cf{eq:nzaj} \cite{park2024hamiltonian}.
That said, such an argument essentially builds upon the assumption that the circuit can be approximated by a local Hamiltonian evolution, which is unlikey reasonable for an \ansatze{} such as UCCSD.

Secondly, the BP phenomenon is considered under random parameter initialization, whereas VQE often benefits from reasonable heuristics for choosing an initial guess.
For example, the all-zero initialization is commonly used in the UCC \ansatzes{}.
If the initial guess happens to be inside the same local trap as the optimal solution, the optimizer may effortlessly find the optimal solution, even if the landscape is a plateau elsewhere.
Otherwise, BP prevents the optimizer from escaping the narrow gorge, resulting in a sub-optimal solution.
One possible solution to this is to use adaptive \ansatzes{}, which have been reported to be able to ``burrow'' out of local traps \cite{grimsley2023adaptive}.

Lastly, we exclusively considered the Jordan-Wigner transformation, as employed in the canonical implementation of UCCSD.
However, the Bravyi-Kitaev transformation might be a better choice due to its locality \cite{cerezo2021cost}.

In recent years, there has been significant progress in the study of BP for periodic \ansatzes{} of the form $U(\vec{\uptheta}) = \prod_{i=1}^{k} U_i(\vec{\uptheta}_i)$, where $U_i(\vec{\uptheta}_i) = \prod_{j=1}^{m} \ConstE^{-\ConstI H_j \theta_{ij}}$.
This class of \ansatze{} includes the quantum approximate optimization algorithm (QAOA) \cite{farhi2014quantum}, HVA, and the alternated dUCC \ansatze{} that we investigate.
Specifically, the pioneering work \cite{larocca2022diagnosing} points out that the dynamic Lie algebra (DLA), defined by the real span of the Lie closure of the generators $\mathfrak{g} := \operatorname{span}_{\mathbb{R}} \ev{\ConstI H_1, \dots, \ConstI H_m}_{\text{Lie}}$, could be a powerful tool for diagnosing BP of periodic \ansatzes{}.
This is because when $k$ is large enough, the set of unitaries that can be generated by the periodic \ansatzes{} is precisely $\ConstE^{\mathfrak{g}} := \{\ConstE^V: V \in \mathfrak{g}\}$.
When $\mathfrak{g} = \mathfrak{su}(2^n)$, $\mathfrak{so}(2^n)$, or is effectively full-rank within an invariant subspace containing the initial state, the distribution over $\ConstE^{\mathfrak{g}}$ converges to a 2-design as $k \to \infty$ \cite{larocca2022diagnosing}.
This allows the variance of gradients to be calculated via integration over the special unitary (or orthogonal) group.
Despite the elegance and intuition of their methods, several limitations exist.
Firstly, the characterization of the generated unitary set and its distribution relies on the assumption that $k$ is sufficiently large, with the specific scaling of $k$ remaining unknown.
Secondly, they assume the parameters to be independently and fully tunable.
But there are other parameter initialization strategies, such as constraining the parameters to be small or correlated, that are used in practice and may avoid BP \cite{park2024hamiltonian,volkoff2021large}.
Thirdly, their results for the full-rank subspace case rely on the initial state being in that subspace.
Lastly, their analysis was limited to full-rank DLA due to the intricacy of integrating over a proper subgroup of the unitary group.
Recent works \cite{ragone2023unified,fontana2023adjoint} have resolved this last limitation, providing a finer analysis of how large $k$ must be for the distribution to approximate a 2-design \cite{fontana2023adjoint}.
Nonetheless, their theories still depend on either the initial state $\rho$ or the observable $O \in \ConstI\mathfrak{g}$, which does not apply in our settings.
These constraints can be removed for specific \ansatze{}, but only for a certain \ansatze{} \cite{diaz2023showcasing}.
Our work differs from these works in two aspects.
Secondly, we address a problem where existing methods are not fully applicable.
However, our results share similar limitations concerning $k$ and parameter assumptions.
\section{Methods}

\subsection{Sketch of the proof of the main theorem}
We analyze cost concentration for the alternated dUCC in the infinite depth limit by providing an explicit formula for the cost variance.
In a unified manner, we will analyze the \mth{$t$} moment of the cost function, $\E{C^t}$, with a focus on $t=1,2$.
The calculation of $\E{C^t}$ involves two key steps:
\begin{enumerate}
  \item
        Through tensor network contraction \cite{pesah2021absence, zhao2021analyzing,
          liu2022presence, martin2023barren}, $\E{C^t}$ can be represented as product of
        enlarged tensors corresponding to the reference state, gates, and the
        observable:
        \begin{equation}
          \label{eq:9xhq} \E{C^{t}}
          =\bra{H_{\mathrm{el}}}^{\otimes t}\qty(\prod
          _{j=1}^{m}\mso{R_j}{t})^{k}\ket{\psi _{0}}^{\otimes 2t}.
        \end{equation}
        Here, $M^{\otimes t,t}:=M^{\otimes{t}}\otimes(M^{*})^{\otimes{t}}$.
        This representation allows for the separation of different parameters, resolving non-linearity by incorporating an enlarged Hilbert space with dimension $2^{2tn}$.
  \item
        The periodicity of excitation rotations implies that each $\mso{R_j}{t}$ is an orthogonal projection, denoted as $P_{M_j}$.
        Let $M:=\bigcap_j M_j$ represent the intersection space.
        Through the convergence of alternating projections \cite{halperin1962product},
        we obtain:
        \begin{equation}
          \label{eq:z86k}
          \lim_{k\to\infty}\E{C^t}=\bra{H_{\mathrm{el}}}^{\otimes t}
          P_M\ket{\psi _{0}}^{\otimes 2t}.
        \end{equation}
        This approach circumvents the challenge of tracking the evolution of $\ket{\psi _{0}}^{\otimes 2t}$ or $\ket{H_{\mathrm{el}}}^{\otimes t}$ at finite $k$.
        Instead, it suffices to determine the intersection space $M$, which turns out to be tractable.

\end{enumerate}

Let us elaborate a bit more on the two key steps.
Firstly, to facilitate a more straightforward description of the elevated tensors, the $2tn$ qubits are rearranged so that the enlarged Hilbert space can be conceptualized as a tensor product of $n$ subsystems.
Each subsystem, which we refer to as a site, comprises $2t$ qubits.
As an illustrative example, $\ket{\psi_0}^{\otimes 2t}$ is identified with
\begin{equation}
  \ket{\psi _{0}}^{\otimes 2t} =|\underbrace{1\dots 1}_{2t}
  \rangle ^{\otimes \NumE} |\underbrace{0\dots 0}_{2t} \rangle ^{\otimes (n-\NumE
    )}.
\end{equation}

Secondly, in BP studies, it is common that the \mth{$t$} moment superoperators of parameterized gates, such as $\mso{R_j}{t}$ in our case, are orthogonal projections.
For instance, the circuit or block of gates is often assumed to be Haar random up to the \nth{2} moment \cite{mcclean2018barren,cerezo2021cost,liu2022presence}.
Under such an assumption one can verify that the \nth{2} moment superoperator of the circuit or block of gates is an orthogonal projection of rank 2.
However, the projections encountered in our analysis are notably more intricate.
Specifically, the \mth{$t$} moment of qubit single excitation rotations forms a projection of rank 70 within the subsystem it acts on, let alone normal excitation rotations that are highly non-local.
Consequently, while the evolution of the reference state can be successfully tracked under the Haar random assumption \cite{cerezo2021cost}, achieving a similar goal may not be possible in the present setting.
Instead, we turn to studying the infinite case.
It is worth noting that the phenomenon where the infinite case is easier than the finite one also arises in the theoretical analysis of classical neural networks \cite{jacot2018neural}.

Thirdly, we ascertain the intersection space $M$ by leveraging the following
identity:
\begin{equation}
  \label{eq:h6uy}
  M^{\perp } =M_{1}^{\perp } +M_{2}^{\perp } +\dotsc +M_{m}^{\perp }.
\end{equation}
In other words, it suffices to determine the spanning set of $M_{j}^{\perp }$, and then take the union to obtain the spanning set of $M^{\perp }$.
To further simplify the analysis, we utilize the symmetries of $P_{M_j}=\mso{R_j}{t}$ and identify the spanning set of $M_j^{\perp}$ within the invariant spaces induced by these symmetries.
To be specific, the symmetries and their corresponding invariant spaces that contain $\ket{\psi_0}^{\otimes 2t}$ are listed below.

\begin{itemize}
  \item
        The $Z_{p}^{\otimes 2t}$-symmetry induces an invariant space
        $\mathcal{H}_{t}^{\mathrm{even}}=\operatorname{span}\mathcal{S}^{\mathrm{even}}_{t}$,
        where
        \begin{equation}
          \mathcal{S}^{\mathrm{even}}_{t}:=\qty{ \ket{b_1 b_2
              \dots b_{2t}}\middle|\sum_{i} b_i \equiv 0 }^{\otimes n}.
        \end{equation}
  \item
        The particle number symmetry further induces an invariant space $\mathcal{H}_{2}^{\mathrm{paired}}=\operatorname{span}\mathcal{S}^{\mathrm{paired}}_{2}$ at $t=2$.
        Here, we define $\mathcal{S}_{2}^{\mathrm{paired}} \subseteq
          \mathcal{S}_{2}^{\mathrm{even}}$ to be the set of paired states and a vector
        $\ket{\Phi } \in \mathcal{S}_{2}^{\mathrm{even}}$ is called a paired state, if
        \begin{align}
          \#\ket{0000} & =\#\ket{1111} +n-2\NumE , \\
          \#\ket{0011} &
          =\#\ket{1100} ,                          \\
          \#\ket{0101} & =\#\ket{1010} ,           \\
          \#\ket{0110} &
          =\#\ket{1001},
        \end{align}
        where $\#\ket{b_{1} b_{2} b_{3} b_{4}}$ counts the
        number of sites in $\ket{\Phi }$ that are in state $\ket{b_{1} b_{2} b_{3}
            b_{4}}$.
  \item
        The $\left(\PermB_{\tau }\right)^{\otimes n}$-symmetry induces an invariant space $\mathcal{H}_{t}^{\tau }$.
        Here, the operator $\PermB_{\tau }$ defines a permutation of qubits in one site by $\tau \in \mathfrak{S}_{2t}$, and $\mathcal{H}_{t}^{\tau }$ is the +1 eigenspace of $( \PermB_{\tau })^{\otimes n}$.

\end{itemize}

In \SM{Note 4}, the spanning set of $M_j^{\perp}\cap \mathcal{H}_{2}^{\mathrm{paired}}$ and $M_j^{\perp}\cap \mathcal{H}_{2}^{\mathrm{paired}}\cap\bigcap_{\tau\in\mathfrak{S}_4}\mathcal{H}_{2}^{\tau }$ is characterized at $t=2$.
With these spanning sets in hand, we can then characterize the intersection space $M$ inside these invariant spaces, according to \Cref{eq:h6uy}.
Surprisingly, as shown in \SM{Note 4}, while the dimension of each $M_j$ is exponentially large, the dimension of the intersection $M\cap\mathcal{H}_{2}^{\mathrm{paired}}$ is at most $\operatorname{poly}(n)$ at $t=2$, under a mild assumption.

\subsection{Simulation details}

The code used in this study is available \cite{mao_2024_13329063}.
The implementation is designed to run with Python version 3.10.12 and Rust version 1.72.1.
The computations were performed on a system equipped with an NVIDIA GeForce RTX 3090 (24 GB) GPU and an Intel(R) Xeon(R) Gold 5222 CPU running at 3.80 GHz.
\section*{Data availability}
The data generated in this study along with the accession codes are available at \url{https://doi.org/10.5281/zenodo.13359192}.

\section*{Code availability}
The code used in this study is available at \url{https://doi.org/10.5281/zenodo.13359192}.

\section*{Acknowledgements}
% \begin{noindent}
We would like to thank Reviewers for taking the time and effort necessary to review the manuscript.
This work was supported in part by the National Natural Science Foundation of China Grants No. 62325210 and the Strategic Priority Research Program of Chinese Academy of Sciences Grant No. XDB28000000.
% \end{noindent}

\section*{Author contributions}
% \begin{noindent}
The project was conceived by all authors.
R.M. developed the methodology, conducted the simulations, and analyzed the results with inputs from G.T. and X.S.
The manuscript was written by R.M. and G.T.
All authors contributed to the manuscript review.
% \end{noindent}

\section*{Competing interests}
The authors declare no competing interests.
 
% Supplementary information

\titleformat{\section}{\normalfont\Large\bfseries}{Supplementary Note \thesection: }{0pt}{}
\renewcommand{\figurename}{Supplementary Figure}

\Crefname{section}{Supplementary Note}{Supplementary Notes}
\Crefname{figure}{Supplementary Figure}{Supplementary Figures}
\Crefname{equation}{Supplementary Equation}{Supplementary Equations}

\renewcommand\refname{Supplementary References}

\pagebreak

\section*{Supplementary Methods: Notations}

Throughout the text, we will use $X,Y,Z$ for Pauli matrices, $I$ for $2\times 2$ identity matrix, and $\openone_{N}$ for $N\times N$ identity matrix.
Two additional $2\times 2$ matrices are used: qubit annihilation operator $Q=\op{0}{1}$, and occupation number operator $N=\op{1}$.
The symbol $n$ is reserved for the number of orbitals (qubits), and $\NumE$ for the number of occupied orbitals.
For any $2\times 2$ matrix $M$ (for example $X,Y,Z,I$), we define $M_i:=I^{\otimes{i-1}}\otimes M\otimes I^{\otimes n-i}$, where $i \in [n]$ and $[n]:=\{1,\dots,n\}$.

We use $z^*$ to denote the complex conjugate of $z$, and $M^{\dagger}$ to denote the conjugate transpose of matrix $M$.
$\ConstI$ represents the imaginary unit.
Denote $\mathfrak{S}_m$ as the symmetric group on a set of size $m$.
The set of bit strings of length $m$ is denoted by $\mathbb{F}_2^{m}$.
For a bit string $\mathbf{b}\in\mathbb{F}_2^{m}$, $\overline{\mathbf{b}}$ denotes the flip of $\mathbf{b}$, i.e., $\overline{\mathbf{b}}=b_1 b_2\dots b_m\in\mathbb{F}_2^{m}$ and $\overline{\mathbf{b}}_i=1-\mathbf{b}_i, \forall i$.
For two bit strings $\mathbf{b},\mathbf{b}'\in\mathbb{F}_2^{m}$, $\mathbf{b}\odot\mathbf{b}'$ denotes the bitwise dot, defined by $\sum_{i=1}^{m} \mathbf{b}_i \cdot \mathbf{b}'_i$.

We first recall some notations, definitions, and theorems from the main text.

Cost function
\begin{equation}
  \label{eq:wtfm}
  C(\vec{\uptheta};\rho,U,O) =\tr( O U(\vec{\uptheta})\rho U(\vec{\uptheta})^{\dagger }).
\end{equation}

Reference state
\begin{equation}
  \label{eq:jj66} \ket{\psi_{0}} :=|\underbrace{1\dots 1}_{\NumE}\underbrace{0\dots
    0}_{n-\NumE} \rangle.
\end{equation}

Electronic structure Hamiltonian
\begin{equation}
  \label{eq:dv4e}
  H_{\mathrm{el}} =\sum _{p >q} h_{pq}\left(\hat{a}_{p}^{\dagger }\hat{a}_{q} +h.c.
  \right)
  +\sum _{p >q >r >s} g_{pqrs}\qty(\hat{a}_{p}^{\dagger }\hat{a}_{q}^
  {\dagger }\hat{a}_{r}\hat{a}_{s} +h.c.).
\end{equation}

\begin{definition}[Alternated (qubit) dUCC \ansatze{}]\label{def:r86r}
  Call the unitary $U^{\vec{R}}_{k}(\vec{\uptheta})$ ($k\in\mathbb{N}_+$) an
  alternated (qubit) dUCC \ansatze{}, if it can be written as
  \begin{equation}
    \label{eq:o2yh}
    U^{\vec{R}}_{k}(\vec{\uptheta})=\prod_{i=1}^{k}\prod_{j=1}^{m}R_j(\theta^{(i)}_j),
  \end{equation}
  where $\vec{R}=(R_1,\dots,R_m)$ is a sequence of (qubit) excitation rotations, $R_j(\theta)=\exp(\theta(\hat{\tau}_j-\hat{\tau}_j^\dagger))$ and $\hat{\tau}_j\in\left\{\hat{a}_{p}^{\dagger}\hat{a}_{q},\hat{a}_{p}^{\dagger} \hat{a}_{q}^{\dagger}\hat{a}_{r}\hat{a}_{s},\dots\right\}$ (or $\hat{\tau}_j\in\left\{Q_p^\dagger Q_q,Q_p^\dagger Q_q^\dagger Q_r Q_s,\dots\right\}$ in qubit version).
\end{definition}

\begin{lemma}[Relationship between variances of cost and gradients]
  \label{lem:ymbr}
  Let $U^{\vec{R}}_{k}(\vec{\uptheta})$ be an alternated (qubit) dUCC \ansatze{}
  defined in \Cref{eq:o2yh}, and $C(\vec{\uptheta};U^{\vec{R}}_{k})$ defined in
  \Cref{eq:dv4e}, then
  \begin{equation}
    \label{eq:v0ff}
    (k\abs{\vec{R}})^{-1}\cdot\Var[\vec{\uptheta}]{C}
    \le\max_{\theta_j\in\vec{\uptheta}}\Var[\vec{\uptheta}]{\partial_{\theta_{j}}
      C} \le \Var[\vec{\uptheta}]{C}.
  \end{equation}
\end{lemma}

\begin{theorem}[Main result]\label{thm:8pgp}
  Let $U^{\vec{R}}_{k}(\vec{\uptheta})$ be an alternated (qubit) dUCC \ansatze{} defined in \Cref{eq:o2yh}, and $C(\vec{\uptheta};U^{\vec{R}}_{k},H_{\mathrm{el}})$ defined in \Cref{eq:wtfm}, where the qubit number is $n$, and $H_{\mathrm{el}}$ is defined in \Cref{eq:dv4e}, with one- and two-electron integrals $h_{pq} ,g_{pqrs}$.
  Denote the simple undirected graph formed by index pairs of (qubit) single
  excitations in $\vec{R}$ by $G$, i.e.,
  \begin{equation}
    G:=(V=[n],E=\left\{(u,v)\middle|A_{uv}\text{ or }A^{\mathrm{qubit}}_{uv}\in \vec{R}\right\}).
  \end{equation}
  Suppose $G$ is connected.
  The limit $\lim _{k\rightarrow \infty }\Var{C\qty(\vec{\uptheta}
    ;U_{k}^{\vec{R}}),H_{\mathrm{el}}} =:V( n)$ exists, and
  \begin{enumerate}
    \item
          \label{itm:x2wa}

          If $\vec{R}$ contains only single excitation rotations, then $V( n) \sim \operatorname{poly}(n^{-1},\NumE,\norm{\vec{h}},\norm{\vec{g}})$.

    \item
          \label{itm:k2mi}

          If $\vec{R}$ contains both single and double excitation rotations, then $V( n) \sim \operatorname{poly} \left(\norm{\vec{h}},\norm{\vec{g}}\right)/\binom{n}{\NumE}$.

    \item
          \label{itm:ka9j}

          If $\vec{R}$ contains only single qubit excitation rotations, then $V( n) \sim \operatorname{poly}(n^{-1},\NumE,\norm{\vec{h}},\norm{\vec{g}})$ if the maximum degree of $G$ is 2, and $V( n) \sim \operatorname{poly} \left(\norm{\vec{h}},\norm{\vec{g}}\right)/\binom{n}{\NumE}$ otherwise.

    \item
          \label{itm:4cl4}

          If $\vec{R}$ contains both single and double qubit excitation rotations (and satisfies the condition in \SM{Theorem 44}), then $V( n) \sim \operatorname{poly} \left(\norm{\vec{h}},\norm{\vec{g}}\right)/\binom{n}{\NumE}$.

  \end{enumerate}
\end{theorem}
\section{Summary of proof}

\Crefrange{app:oql3}{app:q52u} are devoted to the proof of \Cref{thm:8pgp},
which is, in fact, the calculation of the cost variance $\Var{C}$.
The proof is divided into the following parts (illustrated in \Cref{fig:1t9s}):

\begin{itemize}
  \item
        In \Cref{app:oql3}, we give a high-level description of how we calculate the \mth{$t$} moment of the cost function.
        \begin{itemize}
          \item
                We represent the \mth{$t$} moment of cost function as a circuit-like tensor network (\Cref{lem:csxv}), by contracting the initial state, gates, and observable into larger tensors: $\E{C^{t}} =\bra{H_{\mathrm{el}}}^{\otimes t}\left(\prod _{j=1}^{|\vec{R} |}\mso{R_{j}}{t}\right)^{k}\ket{\psi _{0}}^{\otimes 2t}$.
                Since the \mth{$t$} moment of cost is essentially captured by the vector $\ket{\Psi _{t,k}^{\vec{R}}} :=\left(\prod _{j=1}^{|\vec{R} |}\mso{R_{j}}{t}\right)^{k}\ket{\psi _{0}}^{\otimes 2t}$, we refer to such vector as $(\vec{R} ,t,k)$-moment vector (\Cref{def:wzl3}).
          \item
                We introduce site decomposition (\Cref{def:0dpb}), so that the enlarged space $\mathbb{C}^{2^{2tn}}$ involved in calculating $\E{C^{t}}$ can be still viewed as tensor product of $n$ subsystems.
                These subsystems are referred to as sites, each containing $2t$ qubits.
                Under site decomposition, tensor like $\mso{R_{j}}{t}$ can be viewed as an operator acting on sites.
          \item
                We show that each $\mso{R_{j}}{t}$ is an orthogonal projection (\Cref{lem:s14e}), denoted by $P_{M_{j}}$.
          \item
                The convergence of alternating projections hence assures that $\ket{\Psi _{t,\infty }^{\vec{R}}} =P_{M}\ket{\psi _{0}}^{\otimes 2t}$, where $M=\bigcap _{j=1}^{|\vec{R} |}
                  M_{j}$ (\Cref{cor:9rac}).
                While it is not obvious how to find the intersection space $M$, it is easy to find its orthogonal complement $M^{\perp }$, since $M^{\perp } =\sum _{j=1}^{|\vec{R} |} M_{j}^{\perp }$ (\Cref{lem:hpti}).

        \end{itemize}
  \item
        In \Cref{app:s6u6}, we characterize the spanning set of $M_{j}^{\perp }$ (\Cref{lem:vio4}).
        Symmetries of $\mso{R_{j}}{t}$ are used to reduce the space.
        \begin{itemize}
          \item
                The $Z_{p}^{\otimes 2t}$-symmetry helps to reduce to the invariant space $\mathcal{H}_{2}^{\mathrm{even}}$ (\Cref{cor:5iu8}).
          \item
                The particle number symmetry helps to further reduce to the invariant space $\mathcal{H}_{2}^{\mathrm{paired}}$ (\Cref{lem:x2vk}) when $t=2$.
          \item
                The $\left(\PermB_{\tau }\right)^{\otimes n}$-symmetry helps to reduce to the invariant space $\mathcal{H}_{t}^{\tau }$ (\Cref{cor:vciu}).
        \end{itemize}
  \item
        In \Cref{app:jbcg}, we prove that the cost function is unbiased, i.e., the first moment is zero.
        We also illustrate how to calculate the second moment by an example.
  \item
        In \Cref{app:hqmr,app:mbe3,app:q52u}, we prove our main result in \Cref{thm:8pgp}, by calculating $\ket{\Psi ^{\vec{R}}_{2,\infty}}$, and taking the inner product between $\ket{\Psi ^{\vec{R}}_{2,\infty}}$ and $\ket{H_{\mathrm{el}}}^{\otimes 2}$ which evaluates to $\E{C^2}$.
        \begin{itemize}
          \item
                The proof of case 1 and part of case 2 is constructive: we give an explicit vector $\ket{\Psi ^{*}}$, and prove $\ket{\Psi _{t,\infty }^{\vec{R}}} =\ket{\Psi ^{*}}$ by showing that $\ket{\Psi ^{*}} \in M$ and $\ket{\Psi ^{*}} -\ket{\psi _{0}}^{\otimes 4} \in M^{\perp }$.
                In this proof, we use the spanning set of the orthogonal complement $M$ inside $\mathcal{H}_{2}^{\mathrm{paired}}$ --- $(M\cap\mathcal{H}_{2}^{\mathrm{paired}})^{\perp}\cap\mathcal{H}_{2}^{\mathrm{paired}}$.
          \item
                We prove the rest of the theorem by showing that $\dim\left(M\cap \mathcal{H}_{2}^{\mathrm{paired}} \cap \bigcap _{\tau }\mathcal{H}_{2}^{\tau }\right) =1$.
                In such a case, $\ket{\Psi _{t,\infty }^{\vec{R}}}$ is obvious.
        \end{itemize}
\end{itemize}

\begin{figure}[ht]
  \centering
  \includegraphics{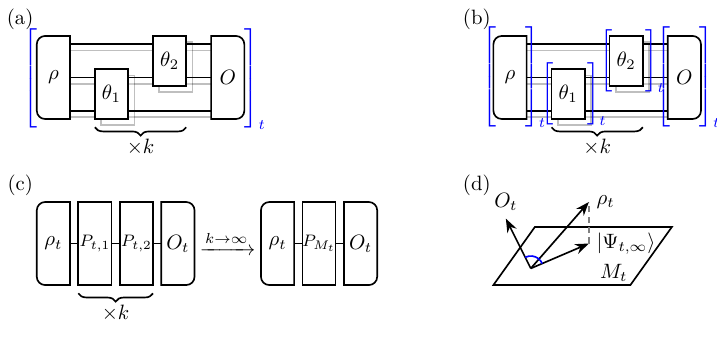}
  \caption{
    Illustration of how we calculate the \mth{$t$} moment of the cost function.
    (a) Tensor network representation of the \mth{$t$} moment.
    Here the square bracket represents taking the \mth{$t$} moment.
    (b) Use the independence of parameters to contract the tensor network into a
    circuit-like one.
    (c) The \mth{$t$} moment of each excitation rotation is an orthogonal projection.
    In the infinite depth limit, these projections contract into an orthogonal projection onto the intersection space $M_t$.
    (d) The \mth{$t$} moment in the infinite depth limit can be recovered as the inner
    product between vector $O_t$ and moment vector $\ket{\Psi_{t,\infty}}$
    (projected $\rho_t$).
  }
  \label{fig:1t9s}
\end{figure}
\section{Proof of \Cref{lem:ymbr}}\label{app:7s1s}

In this section, we prove the equivalence between the variances of the cost and gradients for alternated (qubit) dUCC \ansatze{} (\Cref{lem:ymbr}).
In fact, we prove a more generalized version in \Cref{lem:6cf2} for any bounded frequency periodic function (defined below).
Intuitively, such a function does not have rapid oscillation.

\begin{definition}[Bounded frequency periodic function]
  A function $f: \mathbb{R}^L \to \mathbb{R}$ is called \emph{bounded frequency
    periodic} if it is periodic with the following Fourier expansion
  \begin{equation}
    \label{eq:le7f} f(\vec{\uptheta})=\sum_{\vec{n}\in \mathbb{Z}^L}
    c_{\vec{n}}\exp(\ConstI\vec{n}\cdot\vec{\uptheta}), \quad \text{where }
    c_{\vec{n}}\in\mathbb{C}.
  \end{equation}
  Moreover, there exists a constant $B>0$ independent of $L$, s.t.
  \begin{equation}
    \label{eq:sk6w}
    c_{\vec{n}}=0 \qq{if} \norm{\vec{n}}_\infty > B.
  \end{equation}
\end{definition}

Notice that we implicitly assumed the bounded frequency periodic function to have a period of $2\ConstPi$, but it can be generalized to any periodic function by rescaling.
The following proof is similar to Lemma 1 of \cite{napp2022quantifying}.
We include the proof for completeness.

\begin{lemma}[Relationship between variances of cost and gradients, generalized]
  \label{lem:6cf2}
  For bounded frequency periodic function $f: \mathbb{R}^L \to \mathbb{R}$ where
  the frequencies are bounded by $B$ as in \Cref{eq:sk6w},
  \begin{equation}
    L^{-1}\cdot\Var[\vec{\uptheta}]{f(\vec{\uptheta})} \le \max_{1\le j\le L}\Var[\vec{\uptheta}]{\partial_{\theta_{j}} f(\vec{\uptheta})} \le B^2 \cdot \Var[\vec{\uptheta}]{f(\vec{\uptheta})}.
  \end{equation}
\end{lemma}

\begin{proof}
  Expand $f$ in Fourier basis as in \Cref{eq:le7f}.
  Since $f$ is real-valued, $c_{\vec{n}}=\conj{c}_{-\vec{n}}$ for any $\vec{n}\in \mathbb{Z}^L$.
  Since $\E[\theta\in\mathbb{R}]{\ConstE^{\ConstI m\theta}}=\delta(m)$ for any integer $m$, we
  have
  \begin{align}
    \Var{f}                       & =\E{f^{2}} -\E{f}^{2} =\sum_{\vec{n}}
    \abs{c_{\vec{n}}}^{2} -\abs{c_{\vec{0}}}^{2} =\sum_{\vec{n} \neq \vec{0}}
    \abs{c_{\vec{n}}}^{2},                                                \\
    \Var{\partial_{\theta_{j}} f} & = E\qty[(
      \partial_{\theta_{j}} f)^{2}] -\E{\partial_{\theta_{j}} f}^{2}
    =\sum_{\vec{n}}\abs{c_{\vec{n}}}^{2} n_{j}^{2}-0
    =\sum_{\vec{n}}\abs{c_{\vec{n}}}^{2} n_{j}^{2}.
  \end{align}
  By the fact that $\abs{n_j}\le B$,
  \begin{equation}
    \Var{f}\le \sum _{\vec{n}}
    \abs{ c_{\vec{n}} }^{2} \norm{\vec{n}} _{2}^{2} =\sum _{j}\Var{\partial
      _{\theta _{j}} f}\le L\max_{j}( \partial _{\theta _{j}} f) \le L B^2 \cdot
    \Var{f}.
  \end{equation}
\end{proof}

The following two lemmas show that the cost function of alternated (qubit) dUCC \ansatze{} is bounded frequency periodic, completing the proof of \Cref{lem:ymbr}.

\begin{lemma}[Periodicity of (qubit) excitation rotations]
  \label{lem:9lnr}
  Let $R(\theta)$ be a (qubit) excitation rotation (\Cref{def:r86r}).
  $R(\theta)$ is periodic with a period of $2\ConstPi$.
  Moreover, there exists constant matrices $M_R^+,M_R^-,M_R^0$, such that
  \begin{equation}
    R( \theta ) =\ConstE^{\ConstI\theta } M_{R}^{+} +\ConstE^{-\ConstI\theta } M_{R}^{-} +M_{R}^{0}.
  \end{equation}
\end{lemma}

\begin{proof}
  Recall $R(\theta)=\exp(\theta(\hat{\tau}-\hat{\tau}^\dagger))$, where $\hat{\tau}\in\left\{\hat{a}_{p}^{\dagger}\hat{a}_{q},\hat{a}_{p}^{\dagger} \hat{a}_{q}^{\dagger}\hat{a}_{r}\hat{a}_{s},\dots\right\}\cup\left\{Q_p^\dagger Q_q,Q_p^\dagger Q_q^\dagger Q_r Q_s,\dots\right\}$.
  Notice that $\hat{\tau}-\hat{\tau}^{\dagger}$ is anti-Hermitian.
  It suffices to show that the eigenvalues of $\hat{\tau}-\hat{\tau}^\dagger$ are $0,\pm \ConstI$.
  In fact,
  \begin{equation}
    \begin{split}
      \left(\hat{\tau } -\hat{\tau }^{\dagger }\right)^{2}
       & =\qty(Q_{p}^{\dagger }
      Q_{q}^{\dagger } \dots Q_{r} Q_{s} \dots -Q_{r}^{\dagger } Q_{s}^{\dagger } \dots Q_{p} Q_{q} \dots)^{2} \\
       & =-( N_{p} N_{q} \dots +N_{r} N_{s} \dots -N_{p} N_{q} \dots N_{r} N_{s} \dots ).
    \end{split}
  \end{equation}
  Hence, $\left(\hat{\tau } -\hat{\tau }^{\dagger }\right)^{2}$ is a diagonal matrix, and each element on the diagonal is either 0 or -1.
\end{proof}

\begin{lemma}
  \label{lem:fx5w}
  Let $U^{\vec{R}}_{k}(\vec{\uptheta})$ be an alternated (qubit) dUCC \ansatze{} defined in \Cref{eq:o2yh}.
  The cost function $C(\vec{\uptheta};U^{\vec{R}}_{k})$ defined in \Cref{eq:wtfm} is bounded frequency periodic.
\end{lemma}

\begin{proof}
  By \Cref{lem:9lnr},
  \begin{align}
    C(\vec{\uptheta};U^{\vec{R}}_{k}) & =\tr( OU^{\vec{R}}_{k}(\vec{\uptheta})\op{\psi _{0}} U^{\vec{R}}_{k}(\vec{\uptheta})^{\dagger })                                                                                                                                  \\
                                      & =\tr( O\left(\prod _{ij} R_{j}( \theta ^{(j)}_{i})\right)\op{\psi _{0}}\left(\prod _{ij} R_{j}( \theta^{(j)} _{i})\right)^{\dagger })                                                                                             \\
                                      & =\sum _{\vec{c} ,\vec{c}'\in \{0,\pm 1\}^{m}} \ConstE^{\ConstI(\vec{c} -\vec{c}') \cdot \vec{\uptheta}} \tr( O\left(\prod _{ij} M_{R_{j}}^{c_{ij}}\right) \op{\psi _{0}}\left(\prod _{ij} M_{R_{j}}^{c'_{ij}}\right)^{\dagger }).
  \end{align}

  Since $\norm{\vec{c} -\vec{c}'}_{\infty} \le 2$, $C(\vec{\uptheta};U^{\vec{R}}_{k})$ is bounded frequency periodic.
\end{proof}

To conclude this section, we make two remarks.
First, the lower bound in \Cref{lem:6cf2} can be saturated, for example by the function $f(\vec{\uptheta}) =\sum _{j}\cos^{2}\frac{\theta _{j}}{2}$.
Notice that $f(\vec{\uptheta})$ emerges as the global cost function $C(\vec{\uptheta}) =\tr( OU(\vec{\uptheta})\op{\vec{0}} U^{\dagger }(\vec{\uptheta}))$, where $U(\vec{\uptheta}) =\prod _{j=1}^{n}\exp(\ConstI\theta _{j} X_{j})$ and $O=\sum _{j=1}^{n}\op{0}_{j}$.
Second, we only utilize the periodicity of alternated (qubit) dUCC \ansatzes{} (\Cref{lem:9lnr}) when proving the equivalence between cost variance and gradient variance.
Such an argument could possibly be strengthened using other properties such as non-locality.
\section{Moments of cost function}\label{app:oql3}

In the last section, we showed that for alternated dUCC \ansatzes{}, the variance of the gradient can be bounded by the variance of the cost itself in both directions (\Cref{lem:6cf2,lem:fx5w}).
From now on, we turn to calculating the variance of the cost function.
To start with, we employ the common trick in the study of BP \cite{pesah2021absence, zhao2021analyzing, liu2022presence, martin2023barren} to express the \mth{$t$} moment of the cost function as a circuit-like tensor network.
The motivation is to separate the initial state, gates, and observables apart, and to resolve the non-linearity in high-order moment.

\subsection{Circuit-like tensor network representation of
  \texorpdfstring{$\E{C^t}$}{ECt} and moment vector}

In this section, we express the \mth{$t$} moment of the cost function as a circuit-like tensor network.
All quantum gates related to the same parameter are contracted into one ``elevated'' tensor which has a larger dimension, namely, the \mth{$t$} moment superoperator.
The (matrix form of) \mth{$t$} moment superoperator of an operator $T(
  \theta )$ with a real parameter $\theta$ is defined as
\begin{equation}
  \mso[\theta\in\mathbb{R}]{T(\theta)}{t} = \int _{\mathbb{R}}
  T( \theta )^{\otimes t} \otimes \conj{T}( \theta )^{\otimes t} \dd\theta.
\end{equation}
To simplify the notations, we will use the shorthand $\mso{T}{t}:=\mso[\theta\in\mathbb{R}]{T(\theta)}{t}$ whenever the context is clear.
Moreover, we introduce dedicated notation $\bar{T} ,\tilde{T}$ for the \nth{1}
and \nth{2} moments of $T( \theta )$, as they will be used frequently:
\begin{equation}
  \bar{T}:=\mso{T}{1} ,\quad \tilde{T} :=\mso{T}{2}.
\end{equation}
The vectorization of an operator $U\in \mathbb{C}^{2^m\times 2^m}$ is defined as $\ket{U} :=( U\otimes \openone_{2^m})\sum _{i\in \mathbb{F}^m_2}\ket{i,i}$.

Using the fact that $\E{XY}=\E{X}\E{Y}$ for independent random variables $X,Y$, we can express the \mth{$t$} moment of the cost function of alternated (qubit) dUCC \ansatzes{} as the product of the \mth{$t$} moment of excitation operators.

\begin{lemma}[\mth{$t$} moment of cost function]
  \label{lem:csxv}
  Let $U^{\vec{R}}_{k}(\vec{\uptheta})$ be an alternated (qubit) dUCC \ansatze{} defined in \Cref{def:r86r}, $C(\vec{\uptheta};U^{\vec{R}}_{k},O)$ be the cost function defined in \Cref{eq:wtfm} where $O$ is any observable.
  For any $t\in \mathbb{N}_{+}$ and observables $O_{1} ,\dots ,O_{t}\in
    \mathbb{C}^{2^m\times 2^m}$,
  \begin{equation}
    \label{eq:gi3k} \E{\prod
      _{l=1}^{t}
      C\left(\vec{\uptheta} ;U_{k}^{\vec{R}} ,O_{l}\right)} =\left(\bigotimes_{l=1}^{t}\bra{O_{l}}\right)\left(\prod _{j=1}^{\abs{\vec{R}} }\mso{R_j}{t}\right)^{k}\ket{\psi _{0}}^{\otimes 2t} .
  \end{equation}
  In particular, the \mth{$t$} moment of the cost function $C\qty(\vec{\uptheta}
    ;U_{k}^{\vec{R}} ,O)$ for some observable $O$ is
  \begin{equation}
    \label{eq:xal6} \E{C^{t}\left(\vec{\uptheta} ;U_{k}^{\vec{R}} ,O\right)}
    =\bra{O}^{\otimes t}\left(\prod _{j=1}^{|\vec{R} |}\mso{R_j}{t}\right)^{k}\ket{\psi
      _{0}}^{\otimes 2t}.
  \end{equation}
\end{lemma}

\begin{proof}
  Observe that $C(\vec{\uptheta};U,O) =\tr( OU(\vec{\uptheta}) \op{\psi _{0}} U(\vec{\uptheta})^{\dagger }) =\bra{O} U(\vec{\uptheta})^{\otimes 1,1}\ket{\psi _{0}}^{\otimes 2}$, and use the independence of parameters.
\end{proof}

We make two remarks regarding \Cref{lem:csxv}.
\begin{enumerate}
  \item
        \Cref{eq:gi3k} is useful in calculating covariance, or other quantities alike.
  \item
        The expectations in \Cref{lem:csxv} are taken over random
        parameters $\vec{\uptheta}$ sampled from $\mathbb{R}^{k|\vec{R}|}$, or
        equivalently from $[0,2\ConstPi)^{k|\vec{R}|}$ by the periodicity of (qubit)
        excitation rotations (\Cref{lem:9lnr}).
\end{enumerate}

\Cref{lem:csxv} indicates that the \mth{$t$} moment of cost function for
different observables are essentially captured by a vector, which is the
product of \mth{$t$} moment superoperators of excitation rotations and initial
state.
In other words, one can in principle calculate $\E{C(\vec{\uptheta};O)^t}$ for any observable $O$ if such vector is known --- just take an inner product of the vectorization of $O$ and the vector.
Hence, we refer to such a vector as a moment vector, as defined below.

\begin{definition}[$(\vec{R},t,k)$-moment vector] \label{def:wzl3}
  Let $\vec{R}$ be a sequence of excitation rotations defined in \Cref{def:r86r}, and $\ket{\psi_0}$ be the Hartree-Fock state defined in \Cref{eq:jj66}.
  For $t,k\in\mathbb{N}_+$, the $(\vec{R},t,k)$-moment vector is defined
  to be
  \begin{equation}
    \ket{\Psi^{\vec{R}}_{t,k}}=\qty(\prod
    _{j=1}^{\abs{\vec{R}}}\mso{T_{j}}{t})^{k}\ket{\psi _{0}}^{\otimes 2t}.
  \end{equation}
  In particular, the moment vector of $k$-UCCSD, $k$-BRA etc. is denoted by $\ket{\Psi^{\mathrm{UCCSD}}_{t,k}},\ket{\Psi^{\mathrm{BRA}}_{t,k}}$ etc. And the moment vector of $k$-qubit-UCCSD etc. is denoted by $\ket{\Psi^{\mathrm{qUCCSD}}_{t,k}}$ etc.
\end{definition}

\Cref{eq:xal6} can be rewritten as $\E{C^t}=\bra{O}^{\otimes t}\ket{\Psi^{\vec{R}}_{t,k}}$.
As an example, $\ip{\openone^{\otimes t}}{\Psi^{\vec{R}}_{t,k}} =\bra{\openone}^{\otimes t}\ket{\Psi^{\vec{R}}_{t,k}} =\E{C^t(\vec{\uptheta};U_{k}^{\vec{R}},\openone)}=1$.

\subsection{Site decomposition: a straightforward way to describe the elevated
  tensors}

So far, we have addressed the non-linearity in calculating the \mth{$t$} moment of cost by considering $2t$ replicas of the original system.
Moreover, $\E{C^t}$ turns out to be the inner product of the vectorization of $O$ and the $(\vec{R},t,k)$-moment vector $\ket{\Psi^{\vec{R}}_{t,k}}$.
Before delving into the calculation of $\ket{\Psi^{\vec{R}}_{t,k}}$, it is worthwhile to reorder the qubits in the enlarged Hilbert space so that the elevated tensor can be described naturally.
Such reordering, which we refer to as site decomposition, is formally defined below.

\begin{definition}[Site decomposition] \label{def:0dpb}
  The isomorphism between Hilbert spaces $\mathcal{H}^{\otimes 2t} \cong
    \bigotimes_{i=1}^{n}\mathcal{H}_{i}$ with $\mathcal{H} =\mathbb{C}^{2^{n}} ,
    \mathcal{H}_{i} =\mathbb{C}^{2^{2t}}$, defined by
  \begin{equation}
    \bigotimes_{j=1}^{2t} \ket{b_{1}^{( j)} \dots b_{n}^{( j)}} \rightarrow
    \bigotimes_{i=1}^{n} \ket{b_{i}^{( 1)} \dots b_{i}^{( 2t)}},
  \end{equation}
  is
  called a site decomposition.
  Each $\mathcal{H}_{i} =\mathbb{C}^{2^{2t}}$ is called a site of length $2t$.
  Moreover, we will use $\ket{\Psi}$ to denote any state in the enlarged space $\mathbb{C}^{2^{2tn}}$, while $\ket{\Phi}$ is reserved for computational basis states.
  $\Phi$ is interpreted as a bit string in $\mathbb{F}_2^{2tn}$.
  For $i\in [n], j\in [2t]$, $\Phi_i$ denotes the bit string of the \mth{$i$} site, and $\Phi_{ij}$ denotes the \mth{$j$} bit of \mth{$i$} site.
\end{definition}

This procedure described in \Cref{def:0dpb} can be understood as reordering and splitting the $2tn$ qubits into $n$ equally-sized subsystems, with each subsystem forming a site.
As an example, the \mth{$t$} moment of a qubit single excitation rotation acting on qubits 1 and 2 can be viewed as a tensor acting on sites 1 and 2, as illustrated in \Cref{fig:bb7p}.
Without site decomposition, it is less straightforward to describe such a tensor.
\begin{figure}[ht]
  \centering
  \includegraphics{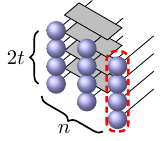}
  \caption{Illustration of the \mth{$t$} moment of a double qubit gate acting on qubits 1 and 2, which can be viewed as a larger tensor acting on sites 1 and 2.
    Here $t=2$ and $n=3$.
    Each ball represents a qubit, the gray boxes together represent the \mth{$t$} moment of the double qubit gate, and the balls in the red dashed cycle form the \nth{3} site of length $2t$.
  }
  \label{fig:bb7p}
\end{figure}
The reader should be aware that we will implicitly assume site decomposition in the subsequent text.

The following proposition reexpresses the initial state and vectorization of the observable $\hat{a}_{p}^{\dagger }\hat{a}_{q}+h.c.
$ and $\hat{a}_{p}^{\dagger }\hat{a}_{q}^{\dagger}\hat{a}_{r}\hat{a}_{s}+h.c.$
under site decomposition.

\begin{proposition}[Initial state and observable under site decomposition]
  \label{prp:7iwk}
  Let $t\in \mathbb{N}_{+}$, and $\ket{\psi _{0}}$ as defined in \Cref{eq:jj66}.
  Under site decomposition,
  \begin{equation}
    \ket{\psi _{0}}^{\otimes 2t}
    =|\underbrace{1\dots 1}_{2t} \rangle ^{\otimes \NumE} |\underbrace{0\dots
      0}_{2t} \rangle ^{\otimes (n-\NumE )}.
  \end{equation}
  After Jordan-Wigner transformation, when $O=\hat{a}_{p}^{\dagger }\hat{a}_{q} +h.c.
  $ ($p >q$),
  \begin{multline}
    \ket{O}^{\otimes t} =\left(\ket{01}_{q}\ket{10}_{p} +
    \ket{10}_{q}\ket{01}_{p}\right)^{\otimes t} \\
    \otimes \underset{a\in [ 1,q \cup ( p,n]}{\bigotimes }
    \left(\ket{00}_a +\ket{11}_a\right)^{\otimes t} \otimes
    \underset{b\in ( p,q)}{\bigotimes }\qty(\ket{00}_b -
    \ket{11}_b)^{\otimes t}.
  \end{multline}
  And when $O=\hat{a}_{p}^{\dagger }\hat{a}_{q}^{\dagger }\hat{a}_{r}\hat{a}_{s} +h.c.
  $ ($p >q >r >s$),
  \begin{multline}
    \ket{O}^{\otimes t} =\qty(\ket{01}_{s}\ket{01}_{r}\ket{10}_{q}
    \ket{10}_{p} +\ket{10}_{s}\ket{10}_{r}\ket{01}_{q}\ket{01}_{p} )^{\otimes t} \\
    \otimes \underset{a\in [ 1,s) \cup ( r,q) \cup ( p,n]}{\bigotimes }
    \left(\ket{00}_a +\ket{11}_a\right)^{\otimes t} \otimes
    \underset{b\in ( s,r) \cup ( q.p)}{\bigotimes }\qty(\ket{00}_b -
    \ket{11}_b)^{\otimes t}.
  \end{multline}
\end{proposition}

\begin{proof}
  Notice that after Jordan-Wigner transformation $\hat{a}_{p} =Q_{p}\prod _{a< p} Z_{a}$, and the vectorization of $Q,I,Z$ is $\ket{Q}=\ket{01},\ket{I} =\ket{00} +\ket{11},\ket{Z} =\ket{00} -\ket{11}$, respectively.
\end{proof}

The following operators related to sites will be useful.

\begin{definition}[$S_{\pi},S_{pq},\PermB_{\tau},F_{V}^{W}$]
  Let $\pi\in\mathfrak{S}_n,\tau\in\mathfrak{S}_{2t},V\subseteq [n]$, and
  $W\subseteq [2t]$.
  \begin{enumerate}
    \item
          Define $S_{\pi}\in \mathbb{C}^{2tn\times 2tn}$ as the permutation
          of sites by $\pi$:
          \begin{equation}
            S_{\pi }\ket{\Phi_1\Phi_2\dots\Phi_n} =\ket{\Phi_{\pi ^{-1} (1)} \Phi_{\pi ^{-1} (2)} \dots \Phi_{\pi ^{-1} (n)}} ,\quad\forall \Phi_i \in \mathbb{F}_{2}^{2t}.
          \end{equation}
          In particular, denote $S_{pq}:=S_{(p\, q)}$ the swap of site $p$ and $q$.
    \item
          Define $\PermB_{\tau}\in \mathbb{C}^{2t\times 2t}$ as the permutation
          of qubits in one site by $\tau$:
          \begin{equation}
            \PermB_{\tau }\ket{b_{1} b_{2} \dots b_{2t}} =\ket{b_{\tau ^{-1} (1)} b_{\tau ^{-1} (2)} \dots b_{\tau ^{-1} (2t)}} ,\quad\forall b_i \in \mathbb{F}_{2}.
          \end{equation}
    \item
          Define $F_{V}^{W}\in \mathbb{C}^{2tn\times 2tn}$ as the flip of
          the \mth{$j$} bit in the \mth{$i$} site for all $i\in V,j\in W$:
          \begin{equation}
            F_{V}^{W}=\prod_{i=1}^{n}\bigotimes_{j=1}^{2t}X_{i}^{[i\in V\land j\in W]}.
          \end{equation}
          In particular, denote $F_{pq\dots }^{ab\dots } =F_{\{p,q,\dots \}}^{\{a,b,\dots \}}$.
  \end{enumerate}
\end{definition}

\subsection{Moments of excitation rotations}

After introducing site decomposition, we now return to the calculation of moment vector $\ket{\Psi^{\vec{R}}_{t,k}}$.
Recall the definition of $\ket{\Psi^{\vec{R}}_{t,k}}$ in \Cref{def:wzl3}.
Since the initial state is fixed to be $\ket{\psi_0}$, it remains to determine each $\mso{R_j}{t}$, where $R_j$ is some (qubit) excitation rotation.
The following lemma gives some basic properties of $\mso{R_j}{t}$.
More properties will be covered in later sections.

\begin{lemma}[Basic properties of $\mso{R}{t}$]\label{lem:s14e}
  Let $R(\theta)$ be a (qubit) excitation rotation (\Cref{def:r86r}), and $t\in\mathbb{N}_+$.
  \begin{enumerate}
    \item
          \label{itm:zclt}
          $\mso{R}{t}=\E{R^{\otimes 2t}}$.
    \item
          \label{itm:egmj}
          $\mso{R}{t}$ is an orthogonal projection.
    \item
          \label{itm:at88}
          Suppose $R( \theta )$ is a qubit excitation rotation, i.e., $R( \theta ) =\exp\left( \theta \left(\hat{\tau } -\hat{\tau }^{\dagger }\right)\right)$ for some $\hat{\tau } =Q_{p_{1}}^{\dagger } \dots Q_{p_{r}}^{\dagger } Q_{p_{r+1}} \dots Q_{p_{2r}}$.
          Define $\Phi _{p_{1} :p_{2r} ,j} :=\Phi _{p_{1} ,j} \dots \Phi _{p_{2r} ,j}
            ,\vec{b}_{0} :=0^{r} 1^{r} ,\vec{b}_{1} :=1^{r} 0^{r}$, and
          \begin{equation}
            n_{ab} :=\#\{j\in [ 2t] | \Phi _{p_{1} :p_{2r} ,j} =\vec{b}_{a} ,\Phi '_{p_{1}
            :p_{2r} ,j} =\vec{b}_{b}\} ,\quad \forall a,b\in \mathbb{F}_{2} .
          \end{equation}
          Then
          \begin{equation}
            \bra{\Phi } \mso{R}{t} \ket{\Phi '} =
            \begin{cases}
              0, &
              ( *) ,                  \\
              ( -1)^{n_{01}}\frac{( n_{00} +n_{11} -1) !
              !( n_{01} +n_{10} -1) !!}{( n_{00} +n_{11} +n_{01} +n_{10}) !!} ,
                 & \text{otherwise} .
            \end{cases}
          \end{equation}
          Here we use the convention that $( -1) !
            !:=1$,
          and $(*)$ is the union of the following cases:
          \begin{itemize}
            \item
                  $\Phi _{i} \neq \Phi '_{i}$ for some $i\in [ n]\backslash \qty{p_{s} |
                      s\in [ 2r]}$.
            \item
                  $\Phi _{p_{1} :p_{2r} ,j} \neq \Phi '_{p_{1} :p_{2r} ,j}$ and $\qty{\Phi
                    _{p_{1} :p_{2r} ,j} ,\Phi '_{p_{1} :p_{2r} ,j}} \nsubseteq \qty{\vec{b}_{0}
                      ,\vec{b}_{1}}$ for some $j\in [ 2t]$.
            \item
                  One of $n_{00} +n_{11}$ and $n_{01}+n_{10}$ is odd.
          \end{itemize}

    \item
          \label{itm:nub3}
          Suppose $R( \theta )$ is an excitation rotation, i.e., $R( \theta ) =\exp\left( \theta \left(\hat{\tau } -\hat{\tau }^{\dagger }\right)\right)$ for some $\hat{\tau } =\hat{a}_{p_{1}}^{\dagger } \dots \hat{a}_{p_{r}}^{\dagger } \hat{a}_{p_{r+1}} \dots \hat{a}_{p_{2r}}$.
          Assume $\hat{\tau}=\pm Q_{p_{1}}^{\dagger } \dots Q_{p_{r}}^{\dagger } Q_{p_{r+1}} \dots Q_{p_{2r}} \prod_{i\in V} Z_i$ for index set $V\subset [n]\backslash \left\{p_{s} | s\in [ 2r]\right\}$.
          Define $\hat{\tau}':=Q_{p_{1}}^{\dagger } \allowbreak\dots Q_{p_{r}}^{\dagger } Q_{p_{r+1}} \dots Q_{p_{2r}}, R'(\theta):=\exp(\theta\left(\hat{\tau}'-(\hat{\tau}')^{\dagger}\right)), \vec{z} :=\bigotimes_{a\in V} \Phi_{a}$ and $X_{i}^{\vec{z}} := \bigotimes_{j=1}^{2t} X_{i}^{z_{j}}$.
          The following conversion rule holds:
          \begin{equation}
            \mso{R}{t}\ket{\Phi }
            =\qty(\prod _{i=1}^{n}
            X_{i}^{\vec{z}})\mso{(R')}{t}\left(\prod _{i=1}^{n} X_{i}^{\vec{z}}\right)\ket{\Phi }.
          \end{equation}
  \end{enumerate}
\end{lemma}

\begin{proof}
  \Cref{itm:zclt}.
  Notice that $R(\theta)$ is real.

  \Cref{itm:egmj}.
  Since $R(\theta)=\exp(\theta(\hat{\tau}-\hat{\tau}^{\dagger}))$ for some
  constant operator $\hat{\tau}$,
  \begin{equation}
    \label{eq:nasu} \mso[\theta
      _{1}]{R( \theta _{1})}{t}\mso[\theta _{2}]{R( \theta _{2})}{t} =\E[\theta
      _{1}]{\mso[\theta _{2}]{R( \theta _{1} +\theta _{2})}{t}}.
  \end{equation}

  By \Cref{lem:9lnr}, $R( \theta )$ is periodic.
  Thus,
  \begin{equation}
    \label{eq:fns9} \E[\theta _{1}]{\mso[\theta _{2}]{R(
        \theta _{1} +\theta _{2})}{t}} =\E[\theta _{1}]{\mso[\theta _{2}]{R( \theta
        _{2})}{t}} =\mso[\theta _{2}]{R( \theta _{2})}{t}.
  \end{equation}

  Combining \Cref{eq:nasu} and \Cref{eq:fns9}, we have $\mso{R}{t}^{2} =\mso{R}{t}$.

  Moreover,
  \begin{equation}
    \mso{R( \theta )}{t}^{\dagger } =\mso{R( -\theta
      )}{t} =\mso{R( \theta )}{t}.
  \end{equation}

  Hence, $\mso{R}{t}$ is an orthogonal projection.

  \Cref{itm:at88}.
  By taking the Taylor expansion of matrix exponential $\ConstE^{M} =\sum _{m\geqslant 0}\frac{M^{m}}{m!
    }$, it is easy to verify that
  \begin{align}
    R( \theta ) & =\exp\left(\theta \left(\op{\vec{b}_{1}} -\op{\vec{b}_{0}}\right)_{p_{1} ,\dots ,p_{2r}}\right)                                                                                                              \\
                & =\openone_{2^{n}} +\big(\sin \theta \left(\op{\vec{b}_{1}}{\vec{b}_{0}} -\op{\vec{b}_{0}}{\vec{b}_{1}}\right) +(\cos \theta -1)\left(\op{\vec{b}_{0}} +\op{\vec{b}_{1}}\right)\big)_{p_{1} ,\dots ,p_{2r}} .
  \end{align}
  Consequently, $\bra{\Phi } \mso{R}{t} \ket{\Phi '}$ can only be non-zero if none of the first 2 cases of $( *)$ happens.
  If so,
  \begin{equation}
    \begin{split}
       & \bra{\Phi }\mso{R(\theta )}{t}\ket{\Phi
      '}                                                                         \\
       & =\E{( -1)^{n_{01}}\cos( \theta )^{n_{00} +n_{11}}\sin( \theta )^{n_{01}
      +n_{10}}}                                                                  \\
       & =
      \begin{cases}
        0,
         & \text{if } n_{00}+n_{11} \text{ or } n_{01}+n_{10} \text{ is
        odd} ,                                                          \\
        ( -1)^{n_{01}}\frac{( n_{00}+n_{11}-1) !
        !( n_{01}+n_{10} -1) !!}{( n_{00}+n_{11} +n_{01}+n_{10}) !!} ,
         & \text{otherwise} .
      \end{cases}
    \end{split}
  \end{equation}

  \Cref{itm:nub3}.
  Since $\hat{\tau}'-(\hat{\tau}')^{\dagger}$ anti-commutes with $\prod
    _{i=1}^{n} X_{i}$,
  \begin{equation}
    \begin{split}
      \E{R(\theta)^{\otimes 2t}}\ket{\Phi }
       & =\E{\bigotimes_{j=1}^{2t}
      R'\left(( -1)^{z_{j}} \theta\right)}\ket{\Phi }                                                                                        \\
       & =\left(\prod _{i=1}^{n} X_{i}^{\vec{z}}\right)\E{R'(\theta)^{\otimes 2t}} \left(\prod _{i=1}^{n} X_{i}^{\vec{z}}\right)\ket{\Phi }.
    \end{split}
  \end{equation}
\end{proof}

\subsection{Convergence of alternating projections}

One of the most important findings in \Cref{lem:s14e} is that the \mth{$t$} moment superoperator of (qubit) excitation rotations $\mso{R}{t}$ are orthogonal projections.
This is not unusual in BP studies.
For example, in \cite{mcclean2018barren,cerezo2021cost,liu2022presence} the circuit (or block of gates) is assumed to be Haar random up to \nth{2} moment.
Under such an assumption one verifies that the \nth{2} moment superoperator of the circuit (or block of gates) is an orthogonal projector of rank 2.
However, the projectors we encountered are significantly more complex compared to the Haar random case.
In fact, the \nth{2} moment of qubit single excitation rotations is a projector of rank 70 in the subsystem it acts on, let alone normal excitation rotations which are highly non-local.
Hence, we do not expect it to be easy to figure out or even bound $\ket{\Psi^{\vec{R}}_{t,k}}$ for any finite $k$.
Rather, we turn to study the infinite-$k$ case, which turns out to be tractable.
The phenomenon that the infinite case is easier than the finite one is ubiquitous, for example, in the theoretical analysis of classical neural networks \cite{jacot2018neural}.
The following lemma will play a central role.

\begin{lemma}[Convergence of alternating projections \cite{halperin1962product}]
  \label{lem:2btd}
  Let $\mathcal{H}$ be a Hilbert space and denote $P_M$ to be the orthogonal projection onto a subspace $M\subseteq \mathcal{H}$.
  Given $N$ subspace $M_1,\dots,M_N$ with intersection $M=M_1\cap\dots\cap M_N$,
  \begin{equation}
    \lim _{k\to \infty } \norm{ ( P_{M_{N}} \cdots P_{M_{1}})^{k}(
    x) -P_{M}( x) } =0,\quad \forall x\in\mathcal{H}.
  \end{equation}
\end{lemma}
Remark that we are working in a finite Hilbert space, and in such a case uniform convergence $\lim _{k\to \infty } ( P_{M_{N}} \cdots P_{M_{1}})^{k} = P_M$ can be shown.

\begin{corollary}
  \label{cor:9rac}
  Let $\vec{R}$ be a sequence of excitation rotations defined in \Cref{def:r86r}, and $t\in\mathbb{N}_+$.
  Denote the projection $\mso{R_j}{t}$ by $P_{M_j}$, where $M_j$ is the subspace that $P_{M_j}$ projects onto.
  Define $M:=\bigcap_j M_j$.
  We have
  \begin{equation}
    \lim_{k\to\infty}\left(\prod_{j=1}^{\abs{\vec{R}}}\mso{R_j}{t}\right)^{k} = P_M
    \qq{and} \ket{\Psi^{\vec{R}}_{t,\infty}}=P_M \ket{\psi_0}^{\otimes 2t}.
  \end{equation}
\end{corollary}

We make two remarks regarding \Cref{cor:9rac}.

\begin{enumerate}
  \item

        The reason why the infinite case is easier is that, by \Cref{cor:9rac}, it suffices to figure out the intersection $M=\bigcap_j M_j$, rather than tracking how $\ket{\psi_0}^{\otimes 2t}$ evolves.

  \item

        In the subsequent text, regardless of the form of $\vec{R}$ and the order of moment $t$, we will denote the subspace that $\mso{R_{j}}{t}$ projects onto by $M_{j}$, and the intersection space $\bigcap _{j} M_{j}$ by $M$, as in \Cref{cor:9rac}.
        The reader should be cautious about which $\vec{R}$ and $t$ are used in context to define $M$ and $M_{j}$.

\end{enumerate}

While we may be able to characterize (albeit a bit complex) each $M_j$, since the matrix form of $\mso{R_j}{t}$ has been explicitly written out in \Cref{lem:s14e}, it is not obvious how to calculate their intersection at first sight.
On the other hand, it is straightforward to determine the spanning set of the orthogonal complement $M^{\perp}$ if one has determined the spanning set of each $M_j^{\perp}$ --- just take the union of these spanning sets.
The reason is explained in \Cref{lem:hpti} (\Cref{itm:52de}).
\Cref{lem:hpti} also includes other properties of the orthogonal complement
which will be used in later sections.
The proof of \Cref{lem:hpti} is elementary and is omitted.

\begin{lemma}
  \label{lem:hpti}
  Let $V_1,V_2,\dots,V_m$, and $V'$ be subspaces of finite dimensional vector space $V$.
  \begin{enumerate}
    \item
          $\orthc{(\orthc{V_1})}=V_1$.
    \item
          \label{itm:52de}
          $\orthc{(V_1\cap V_2 \cap\dots\cap V_m)}=\orthc{V_1}+\orthc{V_2}+\dots+\orthc{V_m}$.
    \item
          \label{itm:cbjx}
          $\orthc{(\orthc{(V_1\cap V')}\cap V')}\cap V'=V_1\cap V'$.
    \item
          \label{itm:58vf}
          $\orthc{((V_1\cap V_2 \cap\dots\cap V_m)\cap V')}\cap V'=\orthc{(V_1\cap
              V'))}\cap V'+\orthc{(V_2\cap V'))}\cap V'+\dots+\orthc{(V_m\cap V'))}\cap V'$.
  \end{enumerate}
  If, in addition, $[P_{V'},P_{V_{i}}] =0$ for all $i=1,2,\dots,m$, then
  \begin{enumerate}
    \setcounter{enumi}{4}
    \item
          \label{itm:qzeo} $(V_{1} \cap
            V')^{\perp} \cap V'=V_{1}^{\perp} \cap V'$.
    \item
          \label{itm:z8mn}
          $((V_{1} \cap V_{2} \cap \dots \cap V_{m}) \cap V')^{\perp} \cap V'=(V_{1} \cap V_{2} \cap \dots \cap V_{m})^{\perp} \cap V'$
  \end{enumerate}
\end{lemma}
\section{Reduce the space by symmetries}\label{app:s6u6}

In the last section, we have hinted at how we will calculate the \mth{$t$} moment
of cost function at $k=\infty$ for alternated dUCC \ansatzes{}:
\begin{enumerate}
  \item
        \label{itm:zvp9} we find out the spanning set of each $\orthc{M_j}$
        (recall that $P_{M_j}:=\mso{R_j}{t}$),
  \item
        \label{itm:upcm} take the union to
        get the spanning set of $\orthc{M}$ (recall that $M:=\bigcap_{j}
          M_j$),
  \item
        \label{itm:hk56} somehow calculate $P_M \ket{\psi_0}^{\otimes 2t}$, using the spanning set of $\orthc{M}$,
  \item
        \label{itm:ywdi} finally, take the inner product between $\ket{H_{\mathrm{el}}}^{\otimes t}$ and $\ket{\Psi^{\vec{R}}_{t,\infty}}=P_M \ket{\psi_0}^{\otimes 2t}$, which evaluates to $\E{C^t}$.
\end{enumerate}

Steps \ref{itm:upcm} and \ref{itm:ywdi} have been explained in the last section.
This section will be devoted to step \ref{itm:zvp9} and will sketch the idea behind step \ref{itm:hk56}.
We do so by restricting ourselves into invariant subspaces using symmetries of $\mso{R_j}{t}$.
The reduction of space is in sequence.
It is worth noting that while these symmetries should apply for any (qubit) excitation rotations and any $t\in\mathbb{N}_+$, we primarily focus on the cases of (qubit) single/double excitation rotations and $t=1,2$ since these are enough for proving \Cref{thm:8pgp}.

\subsection{\texorpdfstring{$Z_p^{\otimes 2t}$}{Zp2t}-symmetry}

The $Z_p^{\otimes 2t}$-symmetry of $\mso{R_j}{t}$ helps to reduce from the whole space $\mathbb{C}^{2^{2tn}}$ to $\mathcal{H}^{\mathrm{even}}_{t}$, the space spanned by states where each site has an even Hamming weight.

\begin{definition}[$\mathcal{S}^{\mathrm{even}}_{t},\mathcal{H}^{\mathrm{even}}_{t}$]
  Define $\mathcal{S}^{\mathrm{even}}_{t},\mathcal{H}^{\mathrm{even}}_{t}
    \subset\mathbb{C}^{2^{2tn}}$ as follows:
  \begin{equation}
    \mathcal{S}^{\mathrm{even}}_{t}:=\qty{ \ket{b_1 b_2
        \dots b_{2t}}\middle|\sum_{i=1}^{2t} b_i \equiv 0
      \pmod{2}, b_i \in
      \mathbb{F}_2 }^{\otimes n},\quad
    \mathcal{H}^{\mathrm{even}}_{t}:=
    \operatorname{span}\mathcal{S}^{\mathrm{even}}_{t}.
  \end{equation}
\end{definition}

Let $R(\theta)$ be a (qubit) excitation rotation, and $\vec{R}$ be a sequence of (qubit) excitation rotations (\Cref{def:r86r}).

\begin{lemma}[$Z_p^{\otimes 2t}$-symmetry]
  For any $p\in [n]$, $\qty[\mso{R}{t} ,Z_{p}^{\otimes 2t}] =0$.
\end{lemma}

\begin{proof}
  Notice that $R(\theta)=\exp(\theta(\hat{\tau}-\hat{\tau}^{\dagger}))$ for some (qubit) excitation $\hat{\tau}$, and $Z_p$ either commutes or anti-commutes with $\hat{\tau}-\hat{\tau}^{\dagger}$ (since $Z$ commutes with $I,Z$ and anti-commutes with $Q,Q^\dagger$).
  \begin{itemize}
    \item
          If $Z_p$ commutes with $\hat{\tau}-\hat{\tau}^{\dagger}$, then $Z_p$ commutes
          with $R(\theta)$, and thus $Z_p^{\otimes 2t}$ commutes with $\E{R( \theta
              )^{\otimes 2t}}$.
    \item
          If $Z_p$ anti-commutes with $\hat{\tau}-\hat{\tau}^{\dagger}$, then $
            Z_p^{\otimes 2t} \E{R( \theta )^{\otimes 2t}}
            Z_p^{\otimes 2t} = \E{R( -\theta )^{\otimes 2t}} = \E{R( \theta )^{\otimes 2t}} $.
  \end{itemize}
\end{proof}

\begin{corollary}[Invariance of $\mathcal{H}^{\mathrm{even}}_{t}$]\label{cor:5iu8}
  $\mathcal{H}^{\mathrm{even}}_{t}$ is an invariant subspace of $\mso{R}{t}$.
  Moreover, $\ket{\Psi^{\vec{R}}_{t,k}},\ket{\Psi^{\vec{R}}_{t,\infty}}\in\mathcal{H}^{\mathrm{even}}_{t}$.
\end{corollary}

\begin{proof}
  Notice that $\mathcal{H}^{\mathrm{even}}_{t}$ is the common +1 eigenspace of $Z_{1}^{\otimes 2t},Z_{2}^{\otimes 2t},\dots,Z_{n}^{\otimes 2t}$.
  Since $Z_{1}^{\otimes 2t},Z_{2}^{\otimes 2t},\dots,Z_{n}^{\otimes 2t}$ and $\mso{R}{t}$ commute mutually, $\mathcal{H}^{\mathrm{even}}_{t}$ is an invariant subspace of $\mso{R}{t}$.
  By \Cref{prp:7iwk}, $\ket{\psi_0}^{\otimes 2t}\in\mathcal{H}^{\mathrm{even}}_{t}$.
  Hence, $\ket{\Psi^{\vec{R}}_{t,k}}=\left(\prod_{j}\mso{R_j}{t}\right)^k\ket{\psi_0}^{\otimes 2t}\in\mathcal{H}^{\mathrm{even}}_{t}$.
  Since $\mathcal{H}^{\mathrm{even}}_{t}$ is closed, $\ket{\Psi^{\vec{R}}_{t,\infty}}=\lim_{k\to\infty}\ket{\Psi^{\vec{R}}_{t,k}}\in\mathcal{H}^{\mathrm{even}}_{t}$.
\end{proof}

The invariance of $\mathcal{H}^{\mathrm{even}}_{t}$ is enough to calculate the first moment of the cost function.
The reader can refer to \Cref{app:jbcg} for more details.
To calculate the second moments, however, we still need to find the spanning set of $\orthc{M_j}$, or equivalently, diagonalize $\tilde{R}_j=\mso{R_j}{2}$.
Since $\mathcal{H}^{\mathrm{even}}_{t}$ is an invariant subspace that contains $\ket{\Psi^{\vec{R}}_{t,\infty}}$, we can diagonalize within $\mathcal{H}^{\mathrm{even}}_{t}$ to save some work.
But before that, we first introduce special notations for sites with an even Hamming weight at $t=2$, so that the notation for states in $\mathcal{H}^{\mathrm{even}}_{2}$ can be simpler.
Recall that when $t=2$, each site has a length of $4$, indicating that the dimension of the Hilbert space for each site is $2^4$.
Out of the 16 computational basis states of each site, there are 8 with an even Hamming weight, as follows.

\begin{definition}[8 special basis states of site at $t=2$]
  Define 8 product state in $\mathbb{C}^{2^4}$ as follows:
  \begin{equation}
    \ket{I_{ab}} =\ket{a,I(a),b,I(b)} ,\quad \ket{X_{ab}} =\ket{a,X(a),b,X(b)}
    ,\quad a,b\in \mathbb{F}_{2}.
  \end{equation}
  Namely, they are
  \begin{align}
    \ket{I_{00}} & =\ket{0000}, & \ket{I_{11}} &
    =\ket{1111}, & \ket{I_{01}} & =\ket{0011}, & \ket{I_{10}} & =\ket{1100},   \\
    \ket{X_{00}} & =\ket{0101}, & \ket{X_{11}} & =\ket{1010}, & \ket{X_{01}} &
    =\ket{0110}, & \ket{X_{10}} & =\ket{1001}.
  \end{align}
\end{definition}

Remark that $\mathcal{S}^{\mathrm{even}}_{2}=\left\{\ket{I_{ab}},\ket{X_{ab}}\middle|a,b\in\mathbb{F}_2\right\}^{\otimes n}$.
For example, the following state is a paired state at $n=8$.
Different ``pair''s are marked in different colors.
\begin{equation}
  \label{eq:tnmr}
  \textcolor{I1}{\ket{I_{00}}}
  \textcolor{X1}{\ket{X_{11}}}
  \textcolor{I2}{\ket{I_{01}}}
  \textcolor{X2}{\ket{X_{10}}}
  \textcolor{X1}{\ket{X_{00}}}
  \textcolor{I2}{\ket{I_{10}}}
  \textcolor{I1}{\ket{I_{11}}}
  \textcolor{X2}{\ket{X_{01}}}.
\end{equation}

We can give a more succinct expression to the square of the vectorization of $\hat{a}_{p}^{\dagger }\hat{a}_{q}+h.c.
$ and $\hat{a}_{p}^{\dagger }\hat{a}_{q}^{\dagger}\hat{a}_{r}\hat{a}_{s}+h.c.$
compared to that in \Cref{prp:7iwk},
using the notations $\ket{I_{ab}},\ket{X_{ab}}$ as follows.

\begin{proposition}[Succinct expression of $\ket{O}^{\otimes 2}$]\label{prp:q0g8}
  After Jordan-Wigner transformation, when $O=\hat{a}_{p}^{\dagger }\hat{a}_{q} +h.c.
  $ ($p >q$),
  \begin{multline}
    \ket{O}^{\otimes 2} =\left(\sum _{c,d\in \mathbb{F}_{2}}\ket{X_{cd}}_{q}
    \ket{X_{\overline{c}\overline{d}}}_{p}\right) \otimes \\
    \underset{a\in [ 1,q) \cup ( p,n]}{\bigotimes }\qty(\sum _{c,d\in \mathbb{F}_{2}}
    \ket{I_{cd}}_{a}) \otimes \underset{b\in ( p,q)}{\bigotimes }
    \left(\sum _{c,d\in \mathbb{F}_{2}}( -1)^{c+d}\ket{I_{cd}}_{b}\right).
  \end{multline}
  And when $O=\hat{a}_{p}^{\dagger }\hat{a}_{q}^{\dagger }\hat{a}_{r}\hat{a}_{s} +h.c.
  $ ($p >q >r >s$),
  \begin{multline}
    \ket{O}^{\otimes 2} =\qty(\sum _{c,d\in \mathbb{F}_{2}}\ket{X_{cd}}_{s}
    \ket{X_{cd}}_{r}\ket{X_{\overline{c}\overline{d}}}_{q}\ket{X_{
        \overline{c}\overline{d}}}_{p}) \otimes \\
    \underset{a\in [ 1,s) \cup ( r,q) \cup ( p,n]}{\bigotimes }\qty(
    \sum _{c,d\in \mathbb{F}_{2}}\ket{I_{cd}}_{a}) \otimes
    \underset{b\in ( s,r) \cup ( q.p)}{\bigotimes }\qty(\sum _{c,d\in
      \mathbb{F}_{2}}( -1)^{c+d}\ket{I_{cd}}_{b}).
  \end{multline}
\end{proposition}

Now that we have defined the notation $\ket{I_{ab}},\ket{X_{ab}}$, we return to the diagonalization of $\tilde{R}_j$ inside $\mathcal{H}^{\mathrm{even}}_{t}$.
The following lemma gives the diagonalization of $\tilde{A}^{\mathrm{qubit}}$ and the partial diagonalization of $\tilde{B}^{\mathrm{qubit}}$ inside $\mathcal{H}^{\mathrm{even}}_{t}$.
These will be used in the next section to derive the spanning set of $\orthc{M_j}$ for $R_{j} \in \left\{A,A^{\mathrm{qubit}} ,B,B^{\mathrm{qubit}}\right\}$.

\begin{lemma}[Diagonalization of $\tilde{A}^{\mathrm{qubit}},\tilde{B}^{\text{qubit}}$ within $\mathcal{H}^{\mathrm{even}}_{t}$]
  \label{lem:duqi} ~
  \begin{enumerate}
    \item
          \label{itm:nag5}
          The subspace that
          $\tilde{A}_{pq}^{\mathrm{qubit}}|_{\mathcal{H}^{\mathrm{even}}_2}$ projects
          onto is spanned by $S_{1} \cup S_{2} \cup S_{3}$, where
          \begin{align}
            S_{1} & :=\left\{\ket{\Phi }\middle| \ket{\Phi } \in \mathcal{S}^{\mathrm{even}}_2 , \Phi _{p} =\Phi _{q}\right\}, \label{eq:wosj}                                                                                      \\
            S_{2} & :=\left\{\ket{\Phi } +( -1)^{\Phi _{p} \odot \Phi _{q}} S_{pq}\ket{\Phi}\middle| \ket{\Phi } \in \mathcal{S}^{\mathrm{even}}_2 , \Phi _{p} \neq \Phi _{q} , \Phi_p \neq\bar{\Phi }_{q}\right\}, \label{eq:9917} \\
            S_{3} & :=\left\{\ket{\Phi } +S_{pq}\ket{\Phi } +F_{pq}^{st}\ket{\Phi } +S_{pq} F_{pq}^{st}\ket{\Phi } \middle|
            \begin{matrix}
              \ket{\Phi } \in \mathcal{S}^{\mathrm{even}}_2 , \ket{\Phi _{p}} =\ket{I_{00}} , \\
              \ket{\Phi _{q}} =\ket{I_{11}}, 1\le s< t\le 4
            \end{matrix}
            \right\}.
            \label{eq:u9lo}
          \end{align}

    \item
          \label{itm:0tgt}
          The space spanned by $S:=\left\{\ket{\Phi } \in \mathcal{S}_{2}^{\mathrm{even}}\middle| \Phi _{p} \neq \Phi _{q} ,\Phi _{p} \neq \overline{\Phi }_{q} ,\Phi _{p} =\overline{\Phi }_{r} ,\Phi _{q} =\overline{\Phi} _{s}\right\}$ is invariant under $\tilde{B}_{pqrs}^{\mathrm{qubit}}$, and when restricted to such subspace, $\tilde{B}_{pqrs}^{\mathrm{qubit}}$ is an orthogonal projection onto the space spanned by $S':=\left\{\ket{\Phi } +( -1)^{\Phi _{p} \odot \Phi _{q}} S_{ps} S_{qr}\ket{\Phi }\middle| \ket{\Phi } \in S\right\}$.
  \end{enumerate}
\end{lemma}

\begin{proof}
  \Cref{itm:nag5}. ---
  By \Cref{lem:s14e} (\Cref{itm:at88}), $\hat{A}_{pq}^{\text{qubit}}$ stabilizes
  every vector in $S_{1} \cup S_{2} \cup S_{3}$, and
  \begin{equation}
    \begin{split}
      \tr(\hat{A}_{pq}^{\text{qubit}} |_{\mathcal{H}_{2}^{\mathrm{even}}}) & =\sum
      _{\ket{\Phi } \in
      \mathcal{S}_{2}^{\mathrm{even}}}\ev{\hat{A}_{pq}^{\text{qubit}}}{\Phi}                                                                                        \\
                                                                           & =\left(8\times 1+48\times \frac{1}{2} +8\times \frac{3}{8}\right) \times 8^{n-2} =\sum
      _{i=1}^{3}\dim\qty(\operatorname{span}
      S_{i}) .
    \end{split}
  \end{equation}
  Finally, since $S_{1} ,S_{2} ,S_{3}$ are mutually orthogonal, $\hat{A}_{pq}^{\text{qubit}} |_{\mathcal{H}_{2}^{\mathrm{even}}}$ is a projection onto $\operatorname{span}( S_{1} \cup S_{2} \cup S_{3})$.

  \Cref{itm:0tgt}. ---
  By \Cref{lem:s14e} (\Cref{itm:at88}), for any $\ket{\Phi}\in S$,
  \begin{equation}
    \tilde{B}^{\text{qubit}}_{pqrs}\ket{\Phi }
    =\frac{1}{2}\ket{\Phi } + \frac{( -1)^{\Phi _{p} \odot \Phi _{q}}}{2}
    S_{ps} S_{qr}\ket{\Phi }.
  \end{equation}
  Thus, $S$ is invariant under $\tilde{B}^{\text{qubit}}_{pqrs}$, and $\tilde{B}_{pqrs}^{\text{qubit}}|_{\operatorname{span} S}$ is indeed an orthogonal projection onto the space spanned by $S':=\left\{\ket{\Phi } +( -1)^{\Phi _{p} \odot \Phi _{q}} S_{ps} S_{qr}\ket{\Phi }\middle| \ket{\Phi } \in S\right\}$.
\end{proof}
\subsection{Particle number symmetry}

The particle number symmetry of $\tilde{R}_j=\mso{R_j}{2}$ helps to reduce the space from $\mathcal{H}^{\mathrm{even}}_{2}$ to $\mathcal{H}^{\mathrm{paired}}_{2}$, where certain constraints regarding the number of $\ket{I_{ab}},\ket{X_{ab}}$ must be satisfied.
Notice that we exclusively focus on the $t=2$ case for particle number symmetry, but such symmetry should hold for general $t$.

\begin{definition}[Paired state and related notions]
  Let $\ket{\Phi } \in \mathcal{S}^{\mathrm{even}}_{2}$ and $V\subseteq [n]$.
  \begin{enumerate}
    \item
          For any $a,b\in \mathbb{F}_{2}$, define
          \begin{equation}
            n_{ab}^{I}\left(\ket{\Phi };V\right) := \#\qty{i\in
              V\middle| \ket{\Phi _{i}} =\ket{I_{ab}}}, \quad n_{ab}^{X}\left(\ket{\Phi }; V\right)
            = \#\left\{i\in V\middle| \ket{\Phi _{i}} =\ket{X_{ab}}\right\}.
          \end{equation}
          If $V$ is omitted, it is assumed that $V=[n]$, i.e., $n^I_{ab}(\ket{\Phi}):=n^I_{ab}(\ket{\Phi};[n]),n^X_{ab}(\ket{\Phi}):=n^X_{ab}(\ket{\Phi};[n])$.
    \item
          Call $\ket{\Phi }$ a paired state, if
          \begin{gather}
            n_{01}^{I}\left(\ket{\Phi }\right) -n_{10}^{I}\left(\ket{\Phi }\right) =
            n_{00}^{X}\left(\ket{\Phi }\right) -n_{11}^{X}\left(\ket{\Phi }\right) =
            n_{01}^{X}\left(\ket{\Phi }\right) -n_{10}^{X}\left(\ket{\Phi }\right) = 0,\\
            n_{00}^{I}\left(\ket{\Phi }\right) -n_{11}^{I}\left(\ket{\Phi }\right) = n-2\NumE.
          \end{gather}
    \item
          Denote the set of all paired states by $\mathcal{S}^{\mathrm{paired}} _2$, and
          the Hilbert space spanned by these states by
          $\mathcal{H}^{\mathrm{paired}}_2$.
    \item
          Define the configuration of a paired state $\ket{\Phi }$ by
          \begin{equation}
            \operatorname{conf}(\ket{\Phi})=(n^I_{01}(\ket{\Phi}),
            n^X_{00}(\ket{\Phi}),n^X_{01}(\ket{\Phi})).
          \end{equation}
    \item
          Denote the set of all paired states with configuration $(a,b,c)$ by
          $\mathcal{S}^{\mathrm{paired}} _{2,(a,b,c)}$, and the Hilbert space spanned by
          these states by $\mathcal{H}^{\mathrm{paired}}_{2,(a,b,c)}$.
  \end{enumerate}
\end{definition}

Remark that the 8 numbers $n^I_{ab}(\ket{\Phi}),n^X_{ab}(\ket{\Phi})$ are uniquely determined by the configuration of a paired state $\ket{\Phi}$, and $n^I_{01}(\ket{\Phi})+ n^X_{00}(\ket{\Phi})+n^X_{01}(\ket{\Phi}) \le \min\left\{\NumE,n-\NumE\right\}$.

\begin{lemma}[Invariance of $\mathcal{H}^{\mathrm{paired}}_{2}$]\label{lem:x2vk}
  Let $\vec{R}$ be a sequence of (qubit) excitation rotations, with $R_j\in \left\{A,A^{\mathrm{qubit}} ,B,B^{\mathrm{qubit}}\right\}$.
  $\mathcal{H}^{\mathrm{paired}}_{2}$ is an invariant subspace of each
  $\tilde{R}_j$.
  Moreover, $\ket{\Psi^{\vec{R}}_{2,k}},\ket{\Psi^{\vec{R}}_{2,\infty}}\in\mathcal{H}^{\mathrm{paired}}_{2}$.
\end{lemma}

\begin{proof}
  In order to prove invariance of $\mathcal{H}^{\mathrm{paired}}_{2}$, it suffices to show that $\ket{\Psi}:=\tilde{R}\ket{\Phi}\in\mathcal{H}^{\mathrm{paired}}_2$, for all $\ket{\Phi}\in\mathcal{S}^{\mathrm{paired}}_2$ and $R\in\left\{A,A^{\mathrm{qubit}} ,B,B^{\mathrm{qubit}}\right\}$.
  \begin{itemize}
    \item
          We first prove the case when $R\in\qty{A^{\mathrm{qubit}},
              B^{\mathrm{qubit}}}$.
          Since $\tilde{A}_{pq}^{\text{qubit}}$ (and $\tilde{B}_{pqrs}^{\text{qubit}}$)
          acts non-trivially only on 2 (and 4) sites, one can enumerate $8^2$ (and $8^4$)
          states of these sites to verify that if $\ket{\Phi'}\in
            \mathcal{S}^{\mathrm{even}}_2$ has non-zero overlap with $\ket{\Psi}$, then for
          all $a,b\in\mathbb{F}_2$ and $S\in \left\{I,X\right\}$,
          \begin{equation}
            \label{eq:xiyt}
            n^S_{ab}(\ket{\Phi'}) - n^S_{ab}(\ket{\Phi}) =
            n^S_{\bar{a}\bar{b}}(\ket{\Phi'}) - n^S_{\bar{a}\bar{b}}(\ket{\Phi}).
          \end{equation}
          Thus, $\ket{\Psi}\in\mathcal{H}^{\mathrm{paired}}_2$.
          In fact, \Cref{eq:xiyt} is obvious when $R=\tilde{A}^{\mathrm{qubit}}$ by the diagonalization given in \Cref{lem:duqi} (\Cref{itm:nag5}).
    \item
          In order to prove the case when $R\in\left\{A,B\right\}$, we utilize the conversion rule
          in \Cref{lem:s14e} (\Cref{itm:nub3}).
          Notice that for any $a,b\in \mathbb{F}_{2}$ and $S\in \{I,X\}$, there exists
          $a',b'\in \mathbb{F}_{2}$ and $S'\in \{I,X\}$, such that for any $\ket{\Phi '}
            \in \mathcal{S}_{2}^{\mathrm{even}}$,
          \begin{equation}
            n_{ab}^{S}\qty(\ket{\Phi
              '}) =n_{a'b'}^{S'}\qty(\qty(\prod _{i=1}^{n}
              X_{i}^{\vec{z}})\ket{\Phi '}) ,\quad n_{\overline{a}\overline{b}}^{S}\left(\ket{\Phi '}\right) =n_{\overline{a} '\overline{b} '}^{S'}\left(\left(\prod _{i=1}^{n} X_{i}^{\vec{z}}\right)\ket{\Phi '}\right) .
          \end{equation}
          Here $\prod _{i=1}^{n} X_{i}^{\vec{z}}$ is defined as in \Cref{lem:s14e} (\Cref{itm:nub3}) with respect to $\ket{\Phi}$.
          Thus, if $\ket{\Phi '} \in \mathcal{S}_{2}^{\mathrm{paired}}$ has a non-zero
          overlap with $\ket{\Psi }$, then
          \begin{align}
            n_{ab}^{S}\left(\ket{\Phi '}\right)
            -n_{ab}^{S}\left(\ket{\Phi }\right)
             & =n_{a'b'}^{S'}\qty(\qty(\prod
              _{i=1}^{n}
            X_{i}^{\vec{z}})\ket{\Phi '}) -n_{a'b'}^{S'}\left(\left(\prod _{i=1}^{n} X_{i}^{\vec{z}}\right)\ket{\Phi }\right)                                                                                                                           \\
             & =n_{\overline{a} '\overline{b} '}^{S'}\left(\left(\prod _{i=1}^{n} X_{i}^{\vec{z}}\right)\ket{\Phi '}\right) -n_{\overline{a} '\overline{b} '}^{S'}\left(\left(\prod _{i=1}^{n} X_{i}^{\vec{z}}\right)\ket{\Phi }\right) \label{eq:hwqa} \\
             & =n_{\overline{a}\overline{b}}^{S}\left(\ket{\Phi '}\right) -n_{\overline{a}\overline{b}}^{S}\left(\ket{\Phi }\right).
          \end{align}
          \Cref{eq:hwqa} follows from \Cref{lem:s14e} (\Cref{itm:nub3}) and \Cref{eq:xiyt}.
  \end{itemize}
  Finally, by \Cref{prp:7iwk}, $\ket{\psi_0}^{\otimes 4}\in\mathcal{H}^{\mathrm{paired}}_{2}$.
  Hence, $\ket{\Psi^{\vec{R}}_{2,k}}=\left(\prod_{j}\tilde{R}_j\right)^k\ket{\psi_0}^{\otimes 4}\in\mathcal{H}^{\mathrm{paired}}_{2}$.
  Since $\mathcal{H}^{\mathrm{paired}}_{2}$ is closed, $\ket{\Psi^{\vec{R}}_{2,\infty}}=\lim_{k\to\infty}\ket{\Psi^{\vec{R}}_{2,k}} \in\mathcal{H}^{\mathrm{paired}}_{2}$.
\end{proof}

\Cref{lem:x2vk} indicates that we can restrict ourselves to
$\mathcal{H}_{2}^{\mathrm{paired}}$ --- instead of the spanning set of each
$M_{j}^{\perp }$, it suffices to find the spanning set of the orthogonal
complement of each $M_{j}$ inside $\mathcal{H}_{2}^{\mathrm{paired}}$:
$\left(M_{j}\cap \mathcal{H}_{2}^{\mathrm{paired}}\right)^{\perp }\cap
  \mathcal{H}_{2}^{\mathrm{paired}}$.
As one may expect, the union of these spanning sets spans $\left(M\cap \mathcal{H}_{2}^{\mathrm{paired}}\right)^{\perp }\cap \mathcal{H}_{2}^{\mathrm{paired}}$, according to \Cref{lem:hpti} (\Cref{itm:58vf}).
The following lemma characterizes the spanning set of $\left(M_{j}\cap \mathcal{H}_{2}^{\mathrm{paired}}\right)^{\perp }\cap \mathcal{H}_{2}^{\mathrm{paired}}$ for $R_{j} \in \left\{A,A^{\mathrm{qubit}} ,B,B^{\mathrm{qubit}}\right\}$.

\begin{lemma}[Spanning set of $\qty(M_{j}\cap \mathcal{H}_{2}^{
        \mathrm{paired} })^{\perp }\cap \mathcal{H}_{2}^{\mathrm{paired}}$ for
    $A,A^{\mathrm{qubit}} ,B,B^{\mathrm{qubit}}$] \label{lem:vio4}
  Denote by $M_{1} ,M_{2} ,M_{3} ,M_{4}$ the space that $\tilde{A}_{pq} ,\tilde{A}_{pq}^{\mathrm{qubit}} ,\tilde{B}_{pqrs} ,\tilde{B}_{pqrs}^{\mathrm{qubit}}$ projects onto.
  Suppose $p >q$ for (qubit) single excitations and $p >q >r >s$ for (qubit) double excitations.
  \begin{enumerate}
    \item
          \label{itm:xojb}
          $\left(M_{1}\cap \mathcal{H}_{2}^{\mathrm{paired}}\right)^{\perp }\cap
            \mathcal{H}_{2}^{\mathrm{paired}}$ is spanned by all the following
          vectors: for any $\ket{\Phi } \in \mathcal{S}_{2}^{\mathrm{paired}}$,
          let $\vec{z} :=\bigoplus _{a\in ( q,p)} \Phi _{a}$,
          \begin{itemize}
            \item
                  $\ket{\Phi } -( -1)^{( \Phi _{p} \oplus \vec{z}) \odot (
                        \Phi _{q} \oplus \vec{z})}
                    S_{pq}\ket{\Phi }$.
            \item
                  $\ket{\Phi } -( -1)^{z_{1} +z_{2}}
                    F_{pq}^{12}\ket{\Phi } -( -1)^{z_{1} +z_{3}} F_{pq}^{13}\ket{\Phi } -( -1)^{z_{2} +z_{3}} F_{pq}^{23}\ket{\Phi }$ if $\ket{\Phi _{p}} =\ket{I_{00}} ,\ket{\Phi _{q}} =\ket{I_{11}}$.
          \end{itemize}
    \item
          \label{itm:rtjl}
          $\left(M_{2}\cap \mathcal{H}_{2}^{\mathrm{paired}}\right)^{\perp }\cap
            \mathcal{H}_{2}^{\mathrm{paired}}$ is spanned by all the following
          vectors: for any $\ket{\Phi } \in \mathcal{S}_{2}^{\mathrm{paired}}$,
          \begin{itemize}
            \item
                  $\ket{\Phi } -( -1)^{\Phi _{p} \odot \Phi _{q}}
                    S_{pq}\ket{\Phi }$.
            \item
                  $\ket{\Phi } -F_{pq}^{12}\ket{\Phi } -F_{pq}^{13}\ket{\Phi }
                    -F_{pq}^{23}\ket{\Phi }$ if $\ket{\Phi _{p}} =\ket{I_{00}} ,
                    \ket{\Phi _{q}} =\ket{I_{11}}$.
          \end{itemize}
    \item
          \label{itm:646w}
          $\left(M_{3}\cap \mathcal{H}_{2}^{\mathrm{paired}}\right)^{\perp }\cap
            \mathcal{H}_{2}^{\mathrm{paired}}$ contains $\ket{\Phi } -( -1)^{( \Phi _{p}
                \oplus \vec{z}) \odot ( \Phi _{q} \oplus \vec{z})}
            S_{ps} S_{qr}\ket{\Phi }$, where $\ket{\Phi } \in \mathcal{S}_{2}^{\mathrm{paired}}$, $\Phi _{p} \neq \Phi _{q} ,\Phi _{p} \neq \overline{\Phi }_{q} ,\Phi _{p} =\overline{\Phi }_{r} ,\Phi _{q} =\overline{\Phi }_{s}$, and $\vec{z} :=\bigoplus _{a\in ( s,r) \cup ( q,p)} \Phi _{a}$.
    \item
          \label{itm:qfwn}
          $\left(M_{4}\cap \mathcal{H}_{2}^{\mathrm{paired}}\right)^{\perp }\cap
            \mathcal{H}_{2}^{\mathrm{paired}}$ contains $\ket{\Phi } -( -1)^{\Phi _{p}
                \odot \Phi _{q}}
            S_{ps} S_{qr}\ket{\Phi }$, where $\ket{\Phi } \in \mathcal{S}_{2}^{\mathrm{paired}}$, and $\Phi _{p} \neq \Phi _{q} ,\Phi _{p} \neq \overline{\Phi }_{q} ,\Phi _{p} =\overline{\Phi }_{r} ,\Phi _{q} =\overline{\Phi }_{s}$.
  \end{enumerate}
\end{lemma}

\begin{proof}
  We first prove \Cref{itm:rtjl,itm:qfwn} using the diagonalization of $\tilde{A}_{pq}^{\mathrm{qubit}} ,\tilde{B}_{pqrs}^{\mathrm{qubit}}$ in $\mathcal{H}_{2}^{\mathrm{even}}$ (\Cref{lem:duqi}), then prove \Cref{itm:xojb,itm:646w} using the qubit to non-qubit conversion rule in \Cref{lem:s14e} (\Cref{itm:nub3}).

  \Cref{itm:rtjl}. ---
  Denote the set of specified vectors by $S$.
  Obviously, $S\subseteq \mathcal{H}_{2}^{\mathrm{paired}}$.
  We need to prove (1) $S\subseteq \left(M_{2}\cap \mathcal{H}_{2}^{\mathrm{paired}}\right)^{\perp }$, (2) $S^{\perp }\cap \mathcal{H}_{2}^{\mathrm{paired}} \subseteq M_{2}\cap \mathcal{H}_{2}^{\mathrm{paired}}$ (since that would imply $\operatorname{span} S\supseteq \left(M_{2}\cap \mathcal{H}_{2}^{\mathrm{paired}}\right)^{\perp }\cap \mathcal{H}_{2}^{\mathrm{paired}}$).
  \begin{enumerate}[(1)]
    \item
          $S\subseteq \left(M_{2}\cap \mathcal{H}_{2}^{\mathrm{paired}}\right)^{\perp }$: We
          show that every vector in $S$ is orthogonal to $S_{1} \cup S_{2} \cup S_{3}$
          defined in \Cref{lem:duqi} (\Cref{itm:nag5}), hence $S\subseteq \qty(M_{2}\cap
            \mathcal{H}_{2}^{\mathrm{even}})^{\perp } \subseteq \qty(M_{2}\cap
            \mathcal{H}_{2}^{\mathrm{paired}})^{\perp }$.
          First, consider the vector $v_{1} :=\ket{\Phi } -( -1)^{\Phi _{p} \odot \Phi _{q}} S_{pq}\ket{\Phi } \in S$.
          \begin{itemize}
            \item
                  If $\Phi _{p} =\Phi _{q}$, then $v_{1} =0$ and $v_{1}$ is
                  orthogonal to $S_{1} \cup S_{2} \cup S_{3}$.
            \item
                  If $\Phi _{p} =\overline{\Phi }_{q}$, then $v_{1}$ is orthogonal
                  to $S_{1} \cup S_{2}$. $v_{1}$ is also orthogonal to $S_{3}$
                  since $v_{1} =\ket{\Phi } -S_{pq}\ket{\Phi }$ while for vectors
                  in $S_{3}$ the overlaps with $\ket{\Phi }$ and $S_{pq}\ket{\Phi }$
                  are the same.
            \item
                  Otherwise, $v_{1}$ is orthogonal to $S_{1} \cup S_{3}$. $v_{1}$
                  is also orthogonal to $S_{2}$ since for vectors in $S_{2}$ the
                  overlaps with $\ket{\Phi }$ and $S_{pq}\ket{\Phi }$ differ by
                  $( -1)^{\Phi _{p} \odot \Phi _{q}}$.
          \end{itemize}

          Next, consider the vector $v_{2} :=\ket{\Phi } -F_{pq}^{12}\ket{\Phi } -F_{pq}^{13}\ket{\Phi } -F_{pq}^{23}\ket{\Phi }$ with $\ket{\Phi _{p}} =\ket{I_{00}} ,\ket{\Phi _{q}} =\ket{I_{11}}$.
          $v_{2}$ is orthogonal to $S_{1} \cup S_{2}$.
          $v_{2}$ is also orthogonal to $S_{3}$ since for vectors in $S_{3}$ the overlap
          with $\ket{\Phi }$ equals one of the overlaps with $F_{pq}^{12}\ket{\Phi }
            ,F_{pq}^{13}\ket{\Phi } ,F_{pq}^{23}\ket{\Phi }$, while the rest two are both
          zero.
    \item
          $S^{\perp }\cap \mathcal{H}_{2}^{\mathrm{paired}} \subseteq M_{2}\cap
            \mathcal{H}_{2}^{\mathrm{paired}}$: Suppose $v\in S^{\perp }\cap
            \mathcal{H}_{2}^{\mathrm{paired}}$, we prove $v\in M_{2}\cap
            \mathcal{H}_{2}^{\mathrm{paired}}$.
          Write $v=\sum _{\ket{\Phi } \in \mathcal{S}_{2}^{\mathrm{paired}}} c_{\Phi }\ket{\Phi }$, with $c_{\Phi } \in \mathbb{C}$.
          Since $v$ is orthogonal to $S$, we have
          \begin{itemize}
            \item
                  $c_{\Phi } =(
                    -1)^{\Phi _{p} \odot \Phi _{q}} c_{\Phi '}$ if $\ket{\Phi '} =S_{pq}\ket{\Phi
                    }$.
            \item
                  $c_{\Phi } =c_{\Phi ^{12}} +c_{\Phi ^{13}} +c_{\Phi ^{23}}$ if $\ket{\Phi _{p}} =\ket{I_{00}} ,\ket{\Phi _{q}} =\ket{I_{11}}$ and $\ket{\Phi ^{ab}} =F_{pq}^{ab}\ket{\Phi }$.
          \end{itemize}

          Hence,
          \begin{multline}
            \ v=\sum _{\Phi _{p} =\Phi _{q}} c_{\Phi }\ket{\Phi }
            +\sum _{\Phi _{p} \neq \Phi _{q} ,\overline{\Phi }_{q}} c_{\Phi }\qty(\ket{\Phi
            } +( -1)^{\Phi _{p} \odot \Phi _{q}}\ket{\Phi })\\
            +\sum _{ \substack{
            \ket{\Phi _{p}} =F_{pq}^{ab}\ket{I_{00}} ,\ket{\Phi _{q}}
            =F_{pq}^{ab}\ket{I_{11}} , \\
            1\leqslant a< b\leqslant 4 } } c_{\Phi
                ^{ab}}\qty(F_{pq}^{ab}\ket{\Phi } +F_{pq}^{ab}
            S_{pq}\ket{\Phi } +\ket{\Phi } +S_{pq}\ket{\Phi }) .
          \end{multline}
          It is straightforward to verify that $v\in \mathcal{H}_{2}^{\mathrm{paired}}$ and by \Cref{lem:duqi} (\Cref{itm:nag5}) $v\in M_{2}\cap \mathcal{H}_{2}^{\mathrm{even}}$.
          Thus, $v\in M_{2}\cap \mathcal{H}_{2}^{\mathrm{paired}}$.
  \end{enumerate}

  \Cref{itm:qfwn}. ---
  Obviously, $v:=\ket{\Phi } -( -1)^{\Phi _{p} \odot \Phi _{q}}
    S_{ps} S_{qr}\ket{\Phi } \in \mathcal{H}_{2}^{\mathrm{paired}}$.
  Use the notation $S,S'$ from \Cref{lem:duqi} (\Cref{itm:0tgt}).
  Since $\operatorname{span} S$ is an invariant space of $P_{M_{4}}$ and $\operatorname{span} S\cap M_{4} =\operatorname{span} S'$, we have $M_{4} =\operatorname{span} S'\oplus \left(\left(\operatorname{span} S\right)^{\perp }\cap M_{4}\right)$.
  $v$ is orthogonal to $\qty(\operatorname{span}
    S)^{\perp }\cap M_{4}$ since $v\in \operatorname{span} S$.
  $v$ is also orthogonal to $\operatorname{span}
    S'$, since for vectors in $S'$ the overlaps with $\ket{\Phi }$ and $S_{ps} S_{qr}\ket{\Phi }$ differ by $( -1)^{\Phi _{p} \odot \Phi _{q}}$.
  Hence, $v\in M_4^{\perp}\cap\mathcal{H}_{2}^{\mathrm{paired}} =(M_4\cap\mathcal{H}_{2}^{\mathrm{paired}})^{\perp}\cap \mathcal{H}_{2}^{\mathrm{paired}}$ by \Cref{lem:hpti} (\Cref{itm:qzeo}).

  \Cref{itm:xojb,itm:646w}. ---
  Same as \Cref{itm:rtjl,itm:qfwn} but use the conversion rule.
\end{proof}

\Cref{lem:vio4} has a simple yet interesting corollary ---
if $\vec{R}$ contains enough (qubit) single excitations, the dimension of
$M\cap \mathcal{H}^{\mathrm{paired}}_{2}$ is at most $\operatorname{poly}(n)$.
Moreover, vectors in $M\cap\mathcal{H}^{\mathrm{paired}}_{2}$ have a nice decomposition which we refer to as decomposition in configuration basis.

\begin{corollary}[Configuration basis decomposition] \label{cor:h168}
  Let $\vec{R}$ be a sequence of (qubit) excitation rotations defined in \Cref{def:r86r}, $M$ be the intersection space defined in \Cref{lem:2btd}.
  Denote the simple undirected graph formed by index pairs of (qubit) single excitations in $\vec{R}$ by $G$ as in \Cref{thm:8pgp}.
  If $G$ is connected, then there exists a function
  $\operatorname{sign}(a,b,c)\in{\pm 1}$ which defines a set of
  configuration basis $\qty{\ket{\Psi_{(a,b,c)}}=\sum_{\ket{\Phi}\in
    \mathcal{S}^{\mathrm{paired}} _{2,(a,b,c)}}
    \operatorname{sign}(\ket{\Phi})\ket{\Phi}}_{a,b,c}$, such that for any
  $\ket{\Psi}\in M\cap\mathcal{H}^{\mathrm{paired}}_{2}$,
  \begin{equation}
    \label{eq:qydw} \ket{\Psi}=\sum_{a+b+c=0}^{\min\left\{\NumE,n-\NumE\right\}} c(a,b,c)
    \ket{\Psi_{(a,b,c)}}, \quad c(a,b,c)\in\mathbb{C}.
  \end{equation}
\end{corollary}

\begin{proof}
  By \Cref{lem:vio4}, $\ket{\Phi}\pm S_{uv}\ket{\Phi}\in \left(M\cap \mathcal{H}_{2}^{ \mathrm{paired} }\right)^{\perp }\cap \mathcal{H}_{2}^{\mathrm{paired}}$, for any $\ket{\Phi}\in\mathcal{H}_{2}^{\mathrm{paired}}$ and $(u,v)\in E$.
  By \Cref{lem:hpti}, $\left(M\cap \mathcal{H}_{2}^{ \mathrm{paired} }\right)^{\perp }\cap \mathcal{H}_{2}^{\mathrm{paired}}\subseteq M^{\perp}$.
  Hence, for any $\ket{\Psi}\in M\cap\mathcal{H}_{2}^{\mathrm{paired}}$, we have $\ip{\Phi}{\Psi}=\pm \mel{\Phi}{S_{uv}^{\dagger}}{\Psi}$.
  In other words, if two paired states differ only by a site swap on an edge, their overlaps with $\ket{\Psi}$ differs by either $+1$ or $-1$, and such a relative sign is independent of $\ket{\Psi}$.
  Since $G$ is connected, one can argue that if two paired states $\ket{\Phi},\ket{\Phi'}$ have the same configuration $(a,b,c)$, their overlaps with $\ket{\Psi}$ differs by at most $\pm 1$ --- two paired states with the same configuration differ by some site permutation, which can be decomposed into a product of swaps, and each swap can, in turn, be decomposed into a product of swaps on edges.
  The decomposition of a site permutation into swaps on edges may not be unique, and if the signs induced by two different decomposition conflict, it must be $c(a,b,c)=0$ for all $\ket{\Psi}\in M\cap\mathcal{H}_{2}^{\mathrm{paired}}$ (one can argue that $\ket{\Phi}\in M^{\perp}$, and thus $\mathcal{S}^{\mathrm{paired}}_{2,(a,b,c)}\subseteq M^{\perp}$).
  Otherwise, the relative signs for paired states in $\mathcal{S}^{\mathrm{paired}}_{2,(a,b,c)}$ must be unique, and are independent of $\ket{\Psi}$.
  Hence, one can pick a sign function that satisfies the requirements --- on $\mathcal{S}^{\mathrm{paired}}_{2,(a,b,c)}\subseteq M^{\perp}$, define $\operatorname{sign}$ arbitrarily, and otherwise define $\operatorname{sign}$ arbitrarily on some $\ket{\Phi}\in\mathcal{S}^{\mathrm{paired}}_{2,(a,b,c)}$, and extend to other $\ket{\Phi'}\in\mathcal{S}^{\mathrm{paired}}_{2,(a,b,c)}$ according the unique relative sign.
\end{proof}

Now that we have characterized the spanning set of $\left(M_{j}\cap \mathcal{H}_{2}^{ \mathrm{paired} }\right)^{\perp }\cap \mathcal{H}_{2}^{\mathrm{paired}}$, we are prepared to prove case 1 and part of case 2 of \Cref{thm:8pgp}.
Our proof for these parts is constructive --- we will give an explicit vector $\ket{\Psi^*}$, and prove that $\ket{\Psi^*}=\ket{\Psi^{\vec{R}}_{2,\infty}}=P_M \ket{\psi_0}^{\otimes 4}$.
To be precise, we show that the following two conditions hold:
\begin{itemize}
  \item
        $\ket{\Psi^*}$ is orthogonal to $\qty(M\cap \mathcal{H}_{2}^{
            \mathrm{paired} })^{\perp }\cap \mathcal{H}_{2}^{\mathrm{paired}}$, and
        $\ket{\Psi^*}\in\mathcal{H}_{2}^{\mathrm{paired}}$, hence $\ket{\Psi^*}\in M$
        (\Cref{lem:hpti} (\Cref{itm:cbjx})).
  \item
        $\ket{\Psi^*}-\ket{\psi_0}^{\otimes 4}\in\qty(M\cap \mathcal{H}_{2}
          ^{ \mathrm{paired} })^{\perp }\cap \mathcal{H}_{2}^{\mathrm{paired}}$,
        hence $\ket{\Psi^*}-\ket{\psi_0}^{\otimes 4}\in M^{\perp}$
        (\Cref{lem:hpti} (\Cref{itm:z8mn})).
\end{itemize}
The reader can jump to \Cref{app:hqmr,app:mbe3} for more details.
Remark that such proof relies on the complete spanning set of $\left(M\cap \mathcal{H}_{2}^{ \mathrm{paired} }\right)^{\perp }\cap \mathcal{H}_{2}^{\mathrm{paired}}$.
We have only given an incomplete spanning set for $\left(M_j\cap \mathcal{H}_{2}^{ \mathrm{paired} }\right)^{\perp }\cap \mathcal{H}_{2}^{\mathrm{paired}}$ of $R_j\in\left\{\tilde{B},\tilde{B}^{\mathrm{qubit}}\right\}$, thus simply taking the union of these spanning set is not enough.
In fact, to prove the rest of \Cref{thm:8pgp} we utilize another symmetry discussed in the next section to reduce the space down to one dimension.
In such case $\ket{\Psi^{\vec{R}}_{2,\infty}}$ is obvious.
\subsection{\texorpdfstring{$(\PermB_{\tau})^{\otimes n}$}{Sbt}-symmetry}

The $(\PermB_{\tau})^{\otimes n}$-symmetry of $\mso{R_j}{t}$ helps to further reduce the space $\mathcal{H}_{2}^{\mathrm{paired}}$.
We use this symmetry to find a one-dimensional subspace of $M\cap \mathcal{H}_{2}^{\mathrm{paired}}$, which contains $\ket{\Psi^{\vec{R}}_{2,\infty}}$.

\begin{definition}[$\mathcal{H}^{\tau}_t$]\label{def:9zvb}
  For any $\tau\in \mathfrak{S}_{2t}$, $\mathcal{H}^{\tau}_t$ is defined to be the +1 eigenspace of $(\PermB_{\tau})^{\otimes n}$.
\end{definition}

Let $R(\theta)$ be a (qubit) excitation rotation, and $\vec{R}$ be a sequence of (qubit) excitation rotations (\Cref{def:r86r}).

\begin{lemma}[$(\PermB_{\tau})^{\otimes n}$-symmetry]\label{lem:8mba}
  For any $\tau\in \mathfrak{S}_{2t}$, $\qty[\mso{R}{t} ,(\PermB_{\tau})^{\otimes n}] =0$.
\end{lemma}

\begin{proof}
  $\left(\PermB_{\tau }\right)^{\otimes n}\E{R( \theta )^{\otimes 2t}}\left(\left(\PermB_{\tau }\right)^{\otimes n}\right)^{-1} =\E{R( \theta )^{\otimes 2t}}$, since the action of
  $\left(\PermB_{\tau }\right)^{\otimes n}$ induces a permutation of $2t$ replicas of
  $R(\theta)$ by $\tau$, which does not change the result.
\end{proof}

\begin{corollary}[Invariance of $\mathcal{H}^{\tau}_t$]\label{cor:vciu}
  For any $\tau\in \mathfrak{S}_{2t}$, $\mathcal{H}^{\tau}_{t}$ is an invariant subspace of $\mso{R}{t}$.
  Moreover, $\ket{\Psi^{\vec{R}}_{t,k}},\ket{\Psi^{\vec{R}}_{t,\infty}}\in\mathcal{H}^{\tau}_{t}$.
\end{corollary}

\begin{proof}
  By \Cref{def:9zvb,lem:8mba}, $\mathcal{H}^{\tau}_{t}$ is an invariant subspace of $\mso{R}{t}$.
  By \Cref{prp:7iwk}, $\ket{\psi_0}^{\otimes 2t}\in\mathcal{H}^{\tau}_{t}$.
  Hence, $\ket{\Psi^{\vec{R}}_{t,k}}=\left(\prod_{j}\mso{R_j}{t}\right)^k\ket{\psi_0}^{\otimes 2t}\in\mathcal{H}^{\tau}_{t}$.
  Since $\mathcal{H}^{\tau}_{t}$ is closed, $\ket{\Psi^{\vec{R}}_{t,\infty}}=\lim_{k\to\infty}\ket{\Psi^{\vec{R}}_{t,k}}\in\mathcal{H}^{\tau}_{t}$.
\end{proof}

We use the following lemma to reduce the space to dimension one in later proofs.

\begin{lemma}
  \label{lem:vxe0}
  Let $\vec{R}$ be a sequence of (qubit) excitation rotations.
  Denote the simple undirected graph formed by index pairs of (qubit) single excitations in $\vec{R}$ by $G$ as in \Cref{thm:8pgp}.
  Suppose
  \begin{enumerate}
    \item
          $G$ is connected.
    \item
          For all $\ket{\Phi}\in \mathcal{S}^{\mathrm{paired}}_{2,(a,b,c)}$ where at
          least two of $a,b,c$ is non-zero, $\ket{\Phi}\in \qty(M\cap
            \mathcal{H}_{2}^{\mathrm{paired}} \cap \bigcap _{\tau \in
              \mathfrak{S}_{4}}\mathcal{H}_{2}^{\tau })^{\perp}$.
  \end{enumerate}
  Then,
  \begin{equation}
    \label{eq:47vm} \ket{\Psi^{\vec{R}}_{2,\infty}} =
    \frac{1}{\binom{n}{\NumE}^2+2\binom{n}{\NumE}} \sum_{l=0}^{m}
    \left(\ket{\Psi_{(l,0,0)}}+\ket{\Psi_{(0,l,0)}}+\ket{\Psi_{(0,0,l)}}\right),
  \end{equation}
  where $\ket{\Psi_{(a,b,c)}} :=
    \sum_{\ket{\Phi}\in\mathcal{S}^{\mathrm{paired}}_{2,(a,b,c)}} \ket{\Phi}, m :=
    \min\left\{\NumE,n-\NumE\right\}$.
\end{lemma}

\begin{proof}
  Since $G$ is connected (\nth{1} condition), vectors in $M\cap \mathcal{H}_{2}^{\mathrm{paired}}$ admits decomposition in configuration basis (\Cref{cor:h168}).
  Say, $\sum_{a,b,c} c(a,b,c)\ket{\Phi'_{(a,b,c)}}$ with $\ket{\Phi '_{(a,b,c)}} =\sum _{\ket{\Phi } \in \mathcal{S}_{2,(a,b,c)}^{\mathrm{paired}}} \pm \ket{\Phi }$.
  Moreover, by the \nth{2} condition, $c(a,b,c)=0$ if two of $a,b,c$ is non-zero.
  We prove (1) $\ket{\Phi '_{( l,0,0)}} =\ket{\Phi _{( l,0,0)}} =\sum _{\ket{\Phi } \in \mathcal{S}_{2,( l,0,0)}^{\mathrm{paired}}}\ket{\Phi }$, and likewise for configurations $( 0,l,0)$ and $( 0,0,l)$, (2) $\dim\left(M\cap \mathcal{H}_{2}^{\mathrm{paired}} \cap \bigcap _{\tau \in \mathfrak{S}_{2t}}\mathcal{H}_{2}^{\tau }\right) \leqslant 1$, (3) the coefficients of $\left\{\ket{\Phi _{( l,0,0)}} ,\ket{\Phi _{( 0,l,0)}} ,\ket{\Phi _{( 0,0,l)}}\right\}_{l}$ in $\ket{\Psi _{2,\infty }^{\vec{R}}}$ are exactly those provided in \Cref{eq:47vm}.

  \begin{enumerate}[(1)]
    \item

          We prove that $\ket{\Phi } -S_{uv}\ket{\Phi } \in \left(M\cap \mathcal{H}_{2}^{\mathrm{paired}}\right)^{\perp } \cap \mathcal{H}_{2}^{\mathrm{paired}}$, for any $\ket{\Phi } \in \mathcal{S}_{2,( l,0,0)}^{\mathrm{paired}}$ (and similarly for $\mathcal{S}_{2,( 0,l,0)}^{\mathrm{paired}}$ or $\mathcal{S}_{2,( 0,0,l)}^{\mathrm{paired}}$) and $( u,v) \in E$.
          After that, one can argue in the same way as in \Cref{cor:h168} to complete the proof.
          Suppose $u< v$.
          By the definition of $G$, either $A_{uv}$ or $A_{uv}^{\mathrm{qubit}}$ is contained in $\vec{R}$.
          \begin{itemize}
            \item
                  If $R_{j} =A_{uv}^{\mathrm{qubit}}$, by \Cref{lem:vio4}
                  (\Cref{itm:rtjl}) $\ket{\Phi } -( -1)^{\Phi _{u} \odot \Phi _{v}}
                    S_{uv}\ket{\Phi } \in \left(M_{j} \cap \mathcal{H}_{2}^{\mathrm{paired}}\right)^{\perp } \cap \mathcal{H}_{2}^{\mathrm{paired}}$.
                  Since $\operatorname{conf}\left(\ket{\Phi }\right) =( l,0,0)$, we have $\Phi _{u} ,\Phi _{v} \in \left\{\ket{I_{ab}}\middle| a,b\in \mathbb{F}_{2}\right\}$.
                  Hence, $( -1)^{\Phi _{u} \odot \Phi _{v}} =1$, and $\ket{\Phi } -S_{uv}\ket{\Phi } =\ket{\Phi } -( -1)^{\Phi _{u} \odot \Phi _{v}} S_{uv}\ket{\Phi } \in \left(M_{j} \cap \mathcal{H}_{2}^{\mathrm{paired}}\right)^{\perp } \cap \mathcal{H}_{2}^{\mathrm{paired}} \subseteq \left(M\cap \mathcal{H}_{2}^{\mathrm{paired}}\right)^{\perp }\mathcal{\cap H}_{2}^{\mathrm{paired}}$.
            \item
                  If $R_{j} =A_{uv}$, by \Cref{lem:vio4} (\Cref{itm:xojb}) $\ket{\Phi }
                    -( -1)^{( \Phi _{u} \oplus \vec{z}) \odot ( \Phi _{v} \oplus \vec{z})}
                    S_{uv}\ket{\Phi } \in \left(M_{j}\mathcal{\cap H}_{2}^{\mathrm{paired}}\right)^{\perp }\mathcal{\cap H}_{2}^{\mathrm{paired}}$, where $\vec{z} =\bigoplus _{a\in ( u,v)} \Phi _{a}$.
                  Since $\operatorname{conf}\left(\ket{\Phi }\right) =( l,0,0)$, we have $\Phi _{u} ,\Phi _{v} ,\vec{z} \in \left\{\ket{I_{ab}}\middle| a,b\in \mathbb{F}_{2}\right\}$.
                  Hence, $( -1)^{( \Phi _{u} \oplus \vec{z}) \odot ( \Phi _{v} \oplus \vec{z})} =1$.
                  The remainder is similar to the previous case.
          \end{itemize}

    \item

          Fix a vector $\ket{\Psi } \in M\cap \mathcal{H}_{2}^{\mathrm{paired}} \cap \bigcap _{\tau \in \mathfrak{S}_{2t}}\mathcal{H}_{2}^{\tau }$.
          We have proved that $\ket{\Psi}\in \operatorname{span} \left\{\ket{\Phi _{( l,0,0)}} ,\ket{\Phi _{( 0,l,0)}} ,\ket{\Phi _{( 0,0,l)}}\right\}_{l}$.
          Denote by $c( a,b,c)$ the coefficient of $\ket{\Phi _{( a,b,c)}}$ in $\ket{\Psi }$, where at most one of $a,b,c$ is non-zero.
          Define $c( a,b,c) =0$ if at least two of $a,b,c$ are non-zero.
          We prove that (2.1) $c( a,b,c) =c( a+1,b,c) +c( a,b+1,c) +c( a,b,c+1)$, (2.2) $c( 1,0,0) =c( 0,1,0) =c( 0,0,1)$.
          If so, it must be $\ket{\Psi } \varpropto \sum _{l=0}^{m}\left(\ket{\Phi _{( l,0,0)}} +\ket{\Phi _{( 0,l,0)}} +\ket{\Phi _{( 0,0,l)}}\right)$.
          Hence, $\dim\left(M\cap \mathcal{H}_{2}^{\mathrm{paired}} \cap \bigcap _{\tau \in \mathfrak{S}_{2t}}\mathcal{H}_{2}^{\tau }\right) \leqslant 1$.

    \item[(2.1)]

          Fix $( u,v) \in E$.
          Suppose $\ket{\Phi } \in \mathcal{S}_{2,( a,b,c)}^{\mathrm{paired}} ,\ket{\Phi _{u}} =\ket{I_{00}} ,\ket{\Phi _{v}} =\ket{I_{11}}$.
          We need to prove that $\ket{\Phi } -F_{uv}^{12}\ket{\Phi } -F_{uv}^{13}\ket{\Phi } -F_{uv}^{23}\ket{\Phi }$ is orthogonal to $\ket{\Psi}$, since the configuration of $F_{uv}^{12}\ket{\Phi } ,F_{uv}^{13}\ket{\Phi } ,F_{uv}^{23}\ket{\Phi }$ is $( a+1,b,c) ,( a,b+1,c) ,( a,b,c+1)$ respectively, and hence $\left(\ket{\Phi } -F_{uv}^{12}\ket{\Phi } -F_{uv}^{13}\ket{\Phi } -F_{uv}^{23}\ket{\Phi }\right)^{\dagger }\ket{\Psi } =c( a,b,c) -c( a+1,b,c) -c( a,b+1,c) -c( a,b,c+1)$.
          If two of $a,b,c$ are non-zero, then $c( a,b,c)=c( a+1,b,c) =c( a,b+1,c) =c( a,b,c+1)=0$ and the equality holds trivially.
          Assume without loss of generality that $a\ge 0,b=c=0$.
          Once more by definition of $G$, either $A_{uv}$ or $A_{uv}^{\mathrm{qubit}}$ is contained in $\vec{R}$.

          \begin{itemize}
            \item

                  If $R_{j} =A_{uv}^{\mathrm{qubit}}$, by \Cref{lem:vio4} (\Cref{itm:rtjl}) $\ket{\Phi } -F_{uv}^{12}\ket{\Phi } -F_{uv}^{13}\ket{\Phi } -F_{uv}^{23}\ket{\Phi } \in \left(M_{j} \cap \mathcal{H}_{2}^{\mathrm{paired}}\right)^{\perp } \cap \mathcal{H}_{2}^{\mathrm{paired}} \subseteq M^{\perp}$.

            \item

                  If $R_{j} =A_{uv}$, by \Cref{lem:vio4} (\Cref{itm:xojb}) $\ket{\Phi } -( -1)^{z_{1} +z_{2}} F_{uv}^{12}\ket{\Phi } -( -1)^{z_{1} +z_{3}} F_{uv}^{13}\ket{\Phi } -( -1)^{z_{2} +z_{3}} F_{uv}^{23}\ket{\Phi } \in \left(M_{j}\mathcal{\cap H}_{2}^{\mathrm{paired}}\right)^{\perp }\cap\mathcal{ H}_{2}^{\mathrm{paired}}\subseteq M^{\perp}$, where $\vec{z} =\bigoplus _{a\in ( u,v)} \Phi _{a}$.
                  Since $\operatorname{conf}\left(\ket{\Phi }\right) =( l,0,0)$, we have $\vec{z} \in \left\{\ket{I_{ab}}\middle| a,b\in \mathbb{F}_{2}\right\}$.
                  Thus, $(-1)^{z_1+z_2}=1$, while $(-1)^{z_1+z_3},(-1)^{z_2+z_3}=\pm 1$.
                  By the \nth{2} condition, $F_{uv}^{13}\ket{\Phi},F_{uv}^{23}\ket{\Phi}\in M^{\perp}$, since the configuration of $F_{uv}^{13}\ket{\Phi}$ is $(a,b+1,c)$ and $a,b+1>0$, and likewise for $F_{uv}^{23}\ket{\Phi}$.
                  Hence, $\ket{\Phi } -F_{uv}^{12}\ket{\Phi } -F_{uv}^{13}\ket{\Phi } -F_{uv}^{23}\ket{\Phi } \in M^{\perp}$.

          \end{itemize}

    \item[(2.2)]

          Fix a vector $\ket{\Phi } \in \mathcal{S}_{2,( 1,0,0)}^{\mathrm{paired}}$ and two swaps $\tau _{1} =( 1\ 3) ,\tau _{2} =( 1\ 4)$.
          Let $\ket{\Phi ^{1}} :=\left(\PermB_{\tau _{1}}\right)^{\otimes n}\ket{\Phi } ,\ket{\Phi ^{2}} :=\left(\PermB_{\tau _{2}}\right)^{\otimes n}\ket{\Phi }$.
          It is easy to check that the configuration of $\ket{\Phi ^{1}} ,\ket{\Phi ^{2}}$ is $( 0,1,0) ,( 0,0,1)$ respectively.
          By \Cref{def:9zvb}, $\ket{\Phi } -\ket{\Phi ^{1}} \in \left(\mathcal{H}_{2}^{\tau _{1}}\right)^{\perp }$ and $\ket{\Phi } -\ket{\Phi ^{1}} \in \left(\mathcal{H}_{2}^{\tau _{2}}\right)^{\perp }$.
          Finally, since $\ket{\Psi } \in \mathcal{H}_{2}^{\tau _{1}} \cap \mathcal{H}_{2}^{\tau _{2}}$, we have $0=\left(\ket{\Phi } -\ket{\Phi ^{1}}\right)^{\dagger }\ket{\Psi } =c( 1,0,0) -c( 0,1,0)$ and $0=\left(\ket{\Phi } -\ket{\Phi ^{2}}\right)^{\dagger }\ket{\Psi } =c( 1,0,0) -c( 0,0,1)$.

    \item

          Since $\ket{\Psi _{2,\infty }^{\vec{R}}} \in M\cap \mathcal{H}_{2}^{\mathrm{paired}} \cap \bigcap _{\tau \in \mathfrak{S}_{2t}}\mathcal{H}_{2}^{\tau }$, there exists constant $c$ such that $\ket{\Psi _{2,\infty }^{\vec{R}}} =c\sum _{l=0}^{m}(\ket{\Phi _{( l,0,0)}} +\ket{\Phi _{( 0,l,0)}} +\ket{\Phi _{( 0,0,l)}})$.
          Recall that $\ip{\openone^{\otimes 2t}}{\Psi _{2,\infty }^{\vec{R}}}=1$, thus
          \begin{equation}
            \begin{split}
              1 & =\ip{\openone^{\otimes 2t}}{\Psi _{2,\infty }^{\vec{R}}}
              =\left(\sum _{l=0}^{m}\ket{\Phi _{( l,0,0)}}\right)^{\dagger }\ket{\Psi _{2,\infty
              }^{\vec{R}}}                                                 \\
                & =c\sum _{l=0}^{m}\binom{n}{l,l,\NumE -l,n-\NumE -l}
              +2c\binom{n}{\NumE}
              =\left(\binom{n}{\NumE}^{2} +2\binom{n}{\NumE}\right) c.
            \end{split}
          \end{equation}
          Solving the equation yields the desired coefficients.

  \end{enumerate}
\end{proof}
\section{Average of cost is zero}\label{app:jbcg}

In this section, we will prove that the cost function of alternated dUCC \ansatzes{} is unbiased, i.e., $\E{C(\vec{\uptheta};U^{\vec{R}}_k,H_{\mathrm{el}})} =0$, where $H_{\mathrm{el}}$ is an electronic structure Hamiltonian.
We also sketch a failed attack towards the variance of the cost function, which serves as an illustration of techniques used in the proof of main results.

\begin{theorem}[Unbiased cost function]
  Let $C(\vec{\uptheta};U^{\vec{R}}_{k},H_{\mathrm{el}})$ be a cost function defined in
  \Cref{eq:wtfm}, where $U^{\vec{R}}_{k}(\vec{\uptheta})$ is an alternated (qubit)
  dUCC \ansatze{} defined in \Cref{eq:o2yh}, and $H_{\mathrm{el}}$ is an electronic structure
  Hamiltonian defined in \Cref{eq:dv4e}.
  We have $\E[\vec{\uptheta}]{C(\vec{\uptheta};U^{\vec{R}}_k,H_{\mathrm{el}})} =0$.
\end{theorem}

\begin{proof}
  By \Cref{prp:7iwk}, $\ket{\hat{a}_{p}^{\dagger }\hat{a}_{q} +h.c.
    }$ and $\ket{\hat{a}_{p}^{\dagger }\hat{a}_{q}^{\dagger }\hat{a}_{r}\hat{a}_{s}
      +h.c.}$ lie in $(\mathcal{H}^{\mathrm{even}}_1)^{\perp}$,
  and by \Cref{cor:5iu8}, $\ket{\Psi^{\vec{R}}_{1,k}} \in \mathcal{H}^{\mathrm{even}}_1$.
  Hence,
  \begin{equation}
    \E{C} = \ip{H_{\mathrm{el}}}{\Psi^{\vec{R}}_{1,k}} =
    \sum_{pq} h_{pq}\ip{\hat{a}_{p}^{\dagger }\hat{a}_{q} +h.c.
    }{\Psi^{\vec{R}}_{1,k}} + \sum_{pqrs} g_{pqrs}\ip{\hat{a}_{p}^{\dagger }\hat{a}_{q}^{\dagger }\hat{a}_{r}\hat{a}_{s} +h.c.}{\Psi^{\vec{R}}_{1,k}} = 0.
  \end{equation}
\end{proof}

As a corollary, to calculate the variance of the cost, it suffices to calculate the
\nth{2} moment, since
\begin{equation}
  \Var{C}=\E{C^2}-\E{C}^2=\E{C^2}.
\end{equation}

Now we sketch a failed attack towards the variance of the cost.
The reader can safely skip this part.
It would be great if we can give a nontrivial bound of the \nth{2} moment of
the cost function using only \nth{1} moments, for example utilizing the following
inequality:
\begin{equation}
  \label{eq:xzcr} \E{C\qty(\vec{\uptheta}
    ;U_{k}^{\vec{R}} ,O)^{2}} \leqslant \E{C\qty(\vec{\uptheta} ;U_{k}^{\vec{R}}
    ,O^{2})} ,\quad k\in \mathbb{N}_{+} \cup \{\infty \} .
\end{equation}
Unfortunately, this is not the case, at least not for \Cref{eq:xzcr}.
The derivation is sketched below and may serve as an example of the techniques used in the proof of main results.

Suppose we wish to bound $\lim_{k\to\infty}\E{C\left(\vec{\uptheta} ;U_{k}^{\vec{R}} ,O\right)^{2}}$ by calculating $\lim_{k\to\infty}\E{C\left(\vec{\uptheta} ;U_{k}^{\vec{R}} ,O^{2}\right)}$.
For simplicity, we consider the case $O=\hat{a}_{p}^{\dagger }\hat{a}_{q} +h.c.
$, since it is straightforward to generalize to other cases.
We further assume that the index pairs of (qubit) single excitations in $\vec{R}$ form a connected graph, i.e., there exists a connected simple graph $G=( V,E)$, such that $V=[ n]$ and $\vec{R}$ contains $A_{uv}$ (or $A_{uv}^{\text{qubit}}$) for all $( u,v) \in E$.
The importance of this assumption will be made clear below.

Recall that $\lim_{k\to\infty}\E{C\left(\vec{\uptheta} ;U_{k}^{\vec{R}} ,O^{2}\right)} =\lim _{k\rightarrow \infty }\bra{O^{2}}\left(\prod _{j=1}^{|\vec{R} |}\mso{R_{j}}{1}\right)^{k}\ket{\psi _{0}}^{\otimes 2} =\ip{O^{2}}{\Psi _{1,\infty }^{\vec{R}}}$ (\Cref{lem:csxv} and \Cref{def:wzl3}).
Since $\ket{\psi _{0}}^{\otimes 2} \in \mathcal{H}_{1}^{\mathrm{even}}$ (\Cref{prp:7iwk}) and $\mathcal{H}_{1}^{\mathrm{even}}$ is invariant under $\mso{R_{j}}{1}$ for all $j$ (\Cref{cor:5iu8}), we can restrict ourselves in the subspace $\mathcal{H}_{1}^{\mathrm{even}}$.
Fix an $R_{j} \in \vec{R}$, and suppose $R_{j}( \theta ) =\exp\left(\theta \left(\hat{\tau } -\hat{\tau }^{\dagger }\right)\right)$, where $\hat{\tau } =\hat{a}_{p_{1}}^{\dagger } \dots \hat{a}_{p_{r}}^{\dagger }\hat{a}_{p_{r+1}} \dots \hat{a}_{p_{2r}}$ (or the corresponding qubit version).
By \Cref{lem:s14e} (\Cref{itm:at88,itm:nub3}), for any $\ket{\Phi } \in
  \mathcal{S}_{1}^{\mathrm{even}}$,
\begin{equation}
  \begin{split}
     & \mso{R_{j}}{1}\ket{\Phi } \\
     & =
    \begin{cases}
      \frac{1}{2}\ket{\Phi } +\frac{1}{2}
      S_{( p_{1} \ p_{r+1}) \dots ( p_{r} \ p_{2r})}\ket{\Phi } , & \text{if } \Phi _{p_{1}} =\dots =\Phi _{p_{r}} =\overline{\Phi }_{p_{r+1}} =\dots =\overline{\Phi }_{p_{2r}} , \\
      \ket{\Phi } ,                                               & \text{otherwise} .
    \end{cases}
  \end{split}
\end{equation}
Notice that $\mso{R_{j}}{1}$ preserves the number of $\ket{00}$ and $\ket{11}$ in $\ket{\Phi }$, due to particle number symmetry.
In particular, if $R_{j} =A_{uv}$ (or $A_{uv}^{\text{qubit}}$) for any $( u,v) \in E$, and $\ket{\Phi } ,\ket{\Phi '} \in \mathcal{S}_{1}^{\mathrm{even}}$ differ by $S_{uv}$, the equality $\mso{R_{j}}{1}\ket{\Psi _{1,\infty }^{\vec{R}}} =\ket{\Psi _{1,\infty }^{\vec{R}}}$ implies that $\ip{\Phi}{\Psi _{1,\infty }^{\vec{R}}} =\ip{\Phi '}{\Psi _{1,\infty }^{\vec{R}}}$.
The connectivity of $G$ and the particle number symmetry thus indicates
\begin{equation}
  \label{eq:gqpm} \ket{\Psi _{1,\infty }^{\vec{R}}} \varpropto
  \sum _{\substack{ \ket{\Phi } \in \mathcal{S}_{1}^{\mathrm{even}} , \\
  \#\ket{11} \text{ in } \ket{\Phi } \text{ is } \NumE } }\ket{\Phi } .
\end{equation}
The coefficient in (\Cref{eq:gqpm}) must be $\binom{n}{\NumE}^{-1}$ by the equality $\ip{\openone_{2^n}}{\Psi _{1,\infty }^{\vec{R}}} =1$.
Lastly, since $O=\hat{a}_{p}^{\dagger }\hat{a}_{q} +h.c.
$,
\begin{equation}
  \ket{O^{2}} =\ket{N_{p} +N_{q} -N_{p}
  N_{q}} =\left(\ket{00}_{p}\ket{11}_{q} +\ket{11}_{p}\ket{00}_{q}\right) \otimes \bigotimes _{a\neq p,q}\left(\ket{00} +\ket{11}\right)_{a} .
\end{equation}
And therefore,
\begin{equation}
  \E{C\qty(\vec{\uptheta} ;U_{\infty }^{\vec{R}}
    ,O^{2})} =\ip{O^{2}}{\Psi _{1,\infty }^{\vec{R}}} =2\binom{n-2}{\NumE
    -1}\binom{n}{\NumE}^{-1} =\frac{2\NumE( n-\NumE)}{n( n-1)} .
\end{equation}
Such an upper bound is at least $1/\operatorname{poly}( n)$, and is roughly $1/2$ when $\NumE =n/2$, definitely not enough to argue for an exponential decay, which is a sufficient condition for BP.
\section{Proof of main result: Case 1}\label{app:hqmr}

In this section, we prove the polynomial concentration of cost function for alternated dUCC \ansatze{} containing only single excitation rotations, with mild connectivity assumption (\Cref{thm:8pgp} (\Cref{itm:x2wa})).
Examples of such \ansatzes{} include $k$-BRA, $k$-UCCS and $k$-UCCGS.

\begin{definition}[Crossing number of paired state]\label{def:loro}
  Let $\ket{\Phi } \in \mathcal{S}_{2}^{\mathrm{paired}}$.
  Draw a graph of $n$ vertices according to $\ket{\Phi }$ as follows:
  \begin{enumerate}
    \item
          Place the $n$ sites of $\ket{\Phi }$ in order on a circle.
          Color $\ket{I_{01}} ,\ket{I_{10}}$ by red, $\ket{X_{00}} ,\ket{X_{11}}$ by green, and $\ket{X_{01}} ,\ket{X_{10}}$ by blue.
          Do not color $\ket{I_{00}} ,\ket{I_{11}}$.
    \item
          Draw an edge inside the circle between any two sites of the same color, and color it with this same color.
          Make sure no three edges cross at the same point.
  \end{enumerate}

  The crossing number of $\ket{\Phi }$, denoted by $\operatorname{cr}\left(\ket{\Phi }\right)$, is defined to be the minimal number of crossing points of edges in color red and green, or green and blue, or red and blue.
  Alternatively, suppose site $i$ is colored by $c( i) \in \{\perp,R,G,B\}$ as
  above,
  \begin{equation}
    \operatorname{cr}\left(\ket{\Phi }\right) :=\#\qty{( i_{1}
    ,i_{2} ,i_{3} ,i_{4})\middle|
    \begin{matrix}
      0\le i_{1} < i_{2} < i_{3} < i_{4}
      \le n,c( i_{1}) \neq c( i_{2}) , \\
      c( i_{1}) =c( i_{3}) \neq \perp,c( i_{2})
      =c( i_{4}) \neq \perp
    \end{matrix}
    } .
  \end{equation}
\end{definition}

\begin{figure}[ht]
  \centering
  \includegraphics{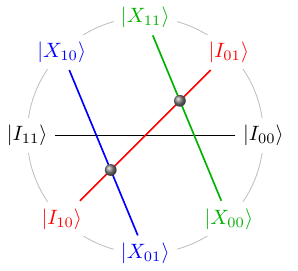}
  \caption{
    Illustration of \Cref{def:loro}.
    The crossing number of this paired state is 2, with the corresponding crossing points marked by gray balls.
    Notice that there are 3 crossing points not counted into the crossing number, since the horizontal edge is not colored.
  }
  \label{fig:8jt0}
\end{figure}

For example, the crossing number of the paired state in \Cref{eq:tnmr} is 2 (\Cref{fig:8jt0}).
We need the following properties of crossing numbers in the proof.

\begin{lemma}[Properties of crossing numbers] \label{lem:038o}
  Let $\ket{\Phi } \in \mathcal{S}_{2}^{\mathrm{paired}} ,1\leqslant u< v\leqslant n,\vec{z} :=\bigoplus _{a\in (u,v)} \Phi _{a}$.
  The following properties of crossing numbers hold.
  \begin{enumerate}
    \item
          \label{itm:ggd9}
          $\operatorname{cr}\left(S_{uv}\ket{\Phi }\right) -\operatorname{cr}\left(\ket{\Phi }\right)
            \equiv ( \Phi _{u} \oplus \vec{z}) \odot ( \Phi _{v} \oplus \vec{z})\pmod{2}$.
    \item
          \label{itm:9iq7}
          If $\ket{\Phi _{u}} =\ket{I_{00}} ,\ket{\Phi _{v}} =\ket{I_{11}}$, then
          \begin{equation}
            \begin{split}
              \label{eq:7qhi} \operatorname{cr}\left(\ket{\Phi }\right) \equiv
              z_{1} +z_{2} +\operatorname{cr}\left(F_{uv}^{12}\ket{\Phi }\right) & \equiv z_{1} +z_{3}
              +\operatorname{cr}\left(F_{uv}^{13}\ket{\Phi }\right)                                    \\
                                                                                 & \equiv z_{2} +z_{3}
              +\operatorname{cr}\left(F_{uv}^{23}\ket{\Phi }\right)\pmod{2}.
            \end{split}
          \end{equation}
  \end{enumerate}
\end{lemma}

\begin{proof}
  \Cref{itm:ggd9}. ---
  The key observation here is that by swapping two neighboring sites $j,j+1$ of
  $\ket{\Phi}$, the parity of crossing number will change by $\Phi _{j} \odot
    \Phi _{j+1}$.
  Indeed, if one of $\Phi _{j}$ or $\Phi _{j+1}$ is not colored or are both in the same color, then $\Phi _{j} \odot \Phi _{j+1} \equiv 0$, and swapping $\Phi _{j}$ with $\Phi _{j+1}$ will not change the crossing number.
  And if $\Phi _{j}$ and $\Phi _{j+1}$ are in different colors, then $\Phi _{j} \odot \Phi _{j+1} \equiv 1$, and swapping $\Phi _{j}$ with $\Phi _{j+1}$ will change the parity of crossing number since there are even-number of red, green or blue sites and thus the degrees of $\Phi _{j}$ and $\Phi _{j+1}$ are both odd.
  Consider the paired state sequence $\left(\ket{\Phi^{(t)}}\right)_{1\le t\le T}$ where
  $T=2v-2u-1$ and
  \begin{equation}
    \label{eq:ylwl} \ket{\Phi ^{(t)}} =
    \begin{cases}
      \ket{\Phi},                               & t=0,         \\
      S_{u+t-1,u+t}\ket{\Phi ^{( t-1)}} ,       & 1\le
      t\le v-u,                                                \\
      S_{2v-u-t-1,2v-u-t}\ket{\Phi ^{( t-1)}} , & v-u< t\le T.
    \end{cases}
  \end{equation}
  Notice that $\ket{\Phi^{(T)}}=S_{uv}\ket{\Phi}$.
  By the observation above,
  \begin{equation}
    \operatorname{cr}\qty(\ket{\Phi ^{(
        t)}}) -\operatorname{cr}\left(\ket{\Phi ^{( t-1)}}\right) \equiv
    \begin{cases}
      \Phi
      _{u} \odot \Phi _{u+t} ,         & 0< t\le v-u, \\
      \Phi _{v} \odot \Phi _{2v-u-t} , &
      v-u< t\le T.
    \end{cases}
  \end{equation}
  Thus,
  \begin{equation}
    \label{eq:kuq5} \operatorname{cr}\left(S_{uv}\ket{\Phi}\right)
    -\operatorname{cr}\left(\ket{\Phi }\right) =\sum
    _{t=1}^{T}\qty(\operatorname{cr}\left(\ket{\Phi ^{( t)}}\right)
    -\operatorname{cr}\left(\ket{\Phi ^{( t-1)}}\right)) \equiv (\Phi _{u}\oplus \vec{z})
    \odot (\Phi _{v}\oplus \vec{z}).
  \end{equation}
  The last equation uses the fact that $\vec{z}\odot\vec{z}\equiv 0$.

  \Cref{itm:9iq7}. ---
  Let $\ket{\Phi^z} =\ket{\Phi } ,\ket{\Phi^a} =F_{uv}^{12}\ket{\Phi }
    ,\ket{\Phi^b} =F_{uv}^{13}\ket{\Phi } ,\ket{\Phi^c} =F_{uv}^{23}\ket{\Phi
    }$.
  Notice that \Cref{eq:7qhi} is equivalent to
  \begin{equation}
    \label{eq:8iab}
    \Phi^z_{u} \odot \vec{z}+\operatorname{cr}\left(\ket{\Phi^z }\right) \equiv \Phi^a_{u} \odot \vec{z}+\operatorname{cr}\left(\ket{\Phi^a }\right) \equiv \Phi^b_{u} \odot \vec{z}+\operatorname{cr}\left(\ket{\Phi^b }\right) \equiv \Phi^c_{u} \odot \vec{z}+\operatorname{cr}\left(\ket{\Phi^c }\right) \pmod{2}.
  \end{equation}
  \begin{itemize}
    \item
          If $u+1=v$, then $\vec{z}=0$.
          By changing $\ket{\Phi^z}$ to $\ket{\Phi^a}$ (or $\ket{\Phi^b} ,\ket{\Phi^c}$), site $u$ and $v$ will both be colored by red (or green, blue), and the number of crossing points increased will be even.
          Hence,
          \begin{equation}
            \label{eq:2pex} \operatorname{cr}\left(\ket{\Phi^z}\right)
            \equiv \operatorname{cr}\left(\ket{\Phi^a}\right) \equiv
            \operatorname{cr}\left(\ket{\Phi^b}\right) \equiv \operatorname{cr}\left(\ket{\Phi^c}\right)
            \pmod{2},
          \end{equation}
          and \Cref{eq:8iab} follows.
    \item
          If $u+1<v$, then by \Cref{eq:kuq5},
          \begin{equation}
            \label{eq:kdup}
            \operatorname{cr}(S_{u,v-1}\ket{\Phi^x}) \equiv \operatorname{cr}(\ket{\Phi^x})
            + \Phi^x_u \odot \vec{z} + \vec{z} \odot \Phi^x_{v-1}, \quad \forall
            x\in\left\{z,a,b,c\right\}.
          \end{equation}
          Similar to \Cref{eq:2pex},
          \begin{equation}
            \label{eq:05n7}
            \operatorname{cr}\left(S_{u,v-1}\ket{\Phi^z}\right) \equiv
            \operatorname{cr}\left(S_{u,v-1}\ket{\Phi^a}\right) \equiv
            \operatorname{cr}\left(S_{u,v-1}\ket{\Phi^b}\right) \equiv
            \operatorname{cr}\left(S_{u,v-1}\ket{\Phi^c}\right) \pmod{2}.
          \end{equation}
          Finally, notice that $\Phi^z_{v-1}=\Phi^a_{v-1}=\Phi^b_{v-1}=\Phi^c_{v-1}$.
          Combined with \Cref{eq:kdup,eq:05n7} we have proved \Cref{eq:8iab}.
  \end{itemize}
\end{proof}

\begin{lemma}
  \label{lem:dbc5}
  Let $G=(V,E)$ be a connected graph with $|V|=n$ vertices.
  Suppose $\vec{R}$ contains a subsequence of single excitation rotations $(A_{uv})_{(u,v)\in E}$.
  Let $M$ be the intersection space of $\vec{R}$ defined in \Cref{cor:9rac} at $t=2$.
  For any $\ket{\Phi},\ket{\Phi'} \in \mathcal{S}_{2,( a,b,c)}^{\mathrm{paired}}$, $( -1)^{\operatorname{cr}\left(\ket{\Phi }\right)}\ket{\Phi } -( -1)^{\operatorname{cr}\left(\ket{\Phi '}\right)}\ket{\Phi '} \in M^{\perp }$.
\end{lemma}

\begin{proof}
  The proof is similar to that of \Cref{cor:h168}, but more specific.
  \begin{itemize}
    \item
          Suppose $\ket{\Phi '} =S_{uv}\ket{\Phi}$ for some $(u,v) \in E$ and $u< v$.
          By assumption, there exist some $R_{j} \in \vec{R}$ such that $R_{j} =A_{uv}$.
          Thus, by \Cref{lem:vio4} (\Cref{itm:xojb}), $\ket{\Phi} -(-1)^{(\Phi _{u} \oplus \vec{z}) \odot (\Phi _{v} \oplus \vec{z})}\ket{\Phi '} \in \left(M_{j} \cap \mathcal{H}_{2}^{\mathrm{paired}}\right)^{\perp} \cap \mathcal{H}_{2}^{\mathrm{paired}}$, where $\vec{z} :=\bigoplus _{a\in (u,v)} \Phi _{a}$.
          By \Cref{lem:hpti}, $\left(M_{j} \cap \mathcal{H}_{2}^{\mathrm{paired}}\right)^{\perp} \cap \mathcal{H}_{2}^{\mathrm{paired}} \subseteq M^{\perp}$.
          And by \Cref{lem:038o}, $(\Phi _{u} \oplus \vec{z}) \odot (\Phi _{v} \oplus \vec{z}) \equiv \operatorname{cr}\left(\ket{\Phi '}\right) -\operatorname{cr}\left(\ket{\Phi}\right) \pmod{2}$.
          Hence, $(-1)^{\operatorname{cr}\left(\ket{\Phi}\right)}\ket{\Phi} -(-1)^{\operatorname{cr}\left(\ket{\Phi '}\right)}\ket{\Phi '} \in M^{\perp}$.
    \item
          Otherwise, $\ket{\Phi '} =S_{\pi}\ket{\Phi}$ for some $\pi \in \mathfrak{S}_{n}$.
          The permutation $\pi $ can be decomposed into a product of swaps, and by the fact that $G$ is connected, each swap can, in turn, be decomposed into a product of swaps on edges (pick an arbitrary path connecting $(u,v)$ in $G$, the swap $(u\ v)$ can be decomposed into a product of swaps along the path).
          Hence, there exists $T\in\mathbb{N}$ and a sequence of edges $(u_{t},v_{t})_{1\le t\le T} \subseteq E$, such that $\pi =\prod _{t=1}^{T}(u_{t} \ v_{t})$.
          Let $\ket{\Phi ^{(t)}} =\prod _{i=1}^{t} S_{u_{i} v_{i}}\ket{\Phi}$, we have $\ket{\Phi ^{(0)}} =\ket{\Phi},\ket{\Phi ^{(T)}} =\ket{\Phi '}$ and $\ket{\Phi ^{(t)}} =S_{u_{t} v_{t}}\ket{\Phi ^{(t-1)}}$ for $1< t\le T$.
          By the previous case, $(-1)^{\operatorname{cr}\left(\ket{\Phi ^{(t-1)}}\right)}\ket{\Phi ^{(t-1)}} -(-1)^{\operatorname{cr}\left(\ket{\Phi ^{(t)}}\right)}\ket{\Phi ^{(t)}} \in M^{\perp}$ for $1< t\le T$.
          The sum of these $T-1$ vectors yields $(-1)^{\operatorname{cr}\left(\ket{\Phi}\right)}\ket{\Phi} -(-1)^{\operatorname{cr}\left(\ket{\Phi '}\right)}\ket{\Phi '}$, which is also a member of $M^{\perp}$.
  \end{itemize}
\end{proof}

\begin{theorem}
  \label{thm:if80}
  Let $G=(V,E)$ be a connected graph with $|V|=n$ vertices, $\vec{R}$ be a sequence of single excitation rotations $(A_{uv})_{(u,v)\in E}$.
  Define $m := \min\left\{\NumE, n-\NumE\right\}$.
  The $(\vec{R},2,\infty)$-moment vector is
  \begin{equation}
    \label{eq:8yx3}
    \ket{\Psi ^{\vec{R}}_{2,\infty}} =
    \frac{\binom{n+4}{2}}{\binom{n}{\NumE}\binom{n+4}{\NumE+2}} \sum_{\substack{
        a,b,c\geqslant 0, \\ a+b+c\le m }}
    \frac{\ket{\Psi_{(a,b,c)}}}{\binom{a+b+c+2}{2}\binom{a+b+c}{a,b,c}},
  \end{equation}
  where $\ket{\Psi_{(a,b,c)}} := \sum
    _{\ket{\Phi}\in\mathcal{S}^{\mathrm{paired}}_{2,(a,b,c)}}(
    -1)^{\operatorname{cr}\left(\ket{\Phi }\right)}\ket{\Phi }$.
\end{theorem}

\begin{proof}
  Denote the right-hand side of \Cref{eq:8yx3} by $\ket{\Psi^*}$.
  Use the notations $M,M_j$ from \Cref{cor:9rac}.
  Since $\ket{\Psi^{\vec{R}}_{2,\infty}}=P_M \ket{\psi _{0}}^{\otimes 4}$, in order to prove $\ket{\Psi^{\vec{R}}_{2,\infty}} =\ket{\Psi ^{*}}$, it suffices to prove that (1) $\ket{\Psi ^{*}} \in M$ and (2) $\ket{\Psi ^{*}} -\ket{\psi _{0}}^{\otimes 4} \in M^{\perp }$.
  \begin{enumerate}[(1)]
    \item
          Obviously, $\ket{\Psi ^{*}} \in \mathcal{H}_{2}^{\mathrm{paired}}$.
          We prove that $\ket{\Psi ^{*}}$ is orthogonal to $\left(M_{j} \cap \mathcal{H}_{2}^{\mathrm{paired}}\right)^{\perp} \cap \mathcal{H}_{2}^{\mathrm{paired}}$, for each $R_{j} \in \vec{R}$.
          If so, we have $\ket{\Psi ^{*}} \in M$ by \Cref{lem:hpti}.
          Recall the spanning set of $\left(M_{j} \cap \mathcal{H}_{2}^{\mathrm{paired}}\right)^{\perp} \cap \mathcal{H}_{2}^{\mathrm{paired}}$ characterized in \Cref{lem:vio4} (\Cref{itm:xojb}).
          Suppose $R_{j} =A_{uv}$.
          For any $\ket{\Phi} \in \mathcal{S}_{2,(a,b,c)}^{\mathrm{paired}}$, let $\vec{z} =\bigoplus _{a\in (u,v)} \Phi _{a}$.
          We check the following two cases.
          \begin{itemize}
            \item

                  $\ket{\Psi ^{*}}$ is orthogonal to $\ket{\Phi} -(-1)^{(\Phi _{u} \oplus
                        \vec{z}) \odot (\Phi _{v} \oplus \vec{z})}
                    S_{uv}\ket{\Phi}$.
                  In fact, the overlap between these two vectors is proportional to $(-1)^{\operatorname{cr}\left(\ket{\Phi}\right)} -(-1)^{(\Phi _{u} \oplus \vec{z}) \odot (\Phi _{v} \oplus \vec{z}) +\operatorname{cr}\left(S_{uv}\ket{\Phi}\right)}$, which is 0 by \Cref{lem:038o} (\Cref{itm:ggd9}).

            \item

                  $\ket{\Psi ^{*}}$ is orthogonal to $\ket{\Phi} -(-1)^{z_{1} +z_{2}}
                    F_{uv}^{12}\ket{\Phi} -(-1)^{z_{1} +z_{3}} F_{uv}^{13}\ket{\Phi} -(-1)^{z_{2} +z_{3}} F_{uv}^{23}\ket{\Phi}$.
                  In fact, the overlap between these two vectors is proportional to
                  \begin{multline}
                    \label{eq:7eft} \frac{( -1)^{\operatorname{cr}\qty(\ket{\Phi
                          })}}{\binom{k+2}{2}\binom{k}{a,b,c}} -\frac{(
                      -1)^{z_1+z_2+\operatorname{cr}\qty(\ket{F_{uv}^{12}\Phi
                          })}}{\binom{k+3}{2}\binom{k+1}{a+1,b,c}} \\
                    -\frac{(
                      -1)^{z_1+z_3+\operatorname{cr}\qty(\ket{F_{uv}^{13}\Phi
                          })}}{\binom{k+3}{2}\binom{k+1}{a,b+1,c}}
                    -\frac{(
                      -1)^{z_2+z_3+\operatorname{cr}\qty(\ket{F_{uv}^{23}\Phi
                          })}}{\binom{k+3}{2}\binom{k+1}{a,b,c+1}}.
                  \end{multline}
                  Here we define $k:=a+b+c$.
                  By \Cref{lem:038o} (\Cref{itm:9iq7}), the numerators in the four fractions in \Cref{eq:7eft} are the same.
                  Factoring out the four identical numerators, it can be proved that \Cref{eq:7eft} equals 0 for any integer $a,b,c$ using combinatorial arguments.
          \end{itemize}

    \item
          Next, we prove $\ket{\Psi ^{*}} -\ket{\psi _{0}}^{\otimes 4} \in M^{\perp }$,
          by expressing $\ket{\Psi ^{*}} -\ket{\psi _{0}}^{\otimes 4}$ as a linear
          combination of vectors in $\left( M\cap \mathcal{H}_{2}^{\mathrm{paired}}\right)^{\perp } \cap
            \mathcal{H}_{2}^{\mathrm{paired}} \subseteq M^{\perp }$.
          Recalls that the union of the spanning sets of each $\left( M_j\cap \mathcal{H}_{2}^{\mathrm{paired}}\right)^{\perp } \cap \mathcal{H}_{2}^{\mathrm{paired}}$ (characterized in \Cref{lem:vio4} (\Cref{itm:xojb})) spans $\left( M\cap \mathcal{H}_{2}^{\mathrm{paired}}\right)^{\perp } \cap \mathcal{H}_{2}^{\mathrm{paired}}$.

          \begin{itemize}
            \item
                  For any $\ket{\Phi } \in
                    \mathcal{S}^{\mathrm{paired}}_{2,(0,0,0)}$,
                  by \Cref{lem:dbc5}, we have
                  \begin{equation}
                    \label{eq:lmik} \ket{\Phi }
                    -\ket{\psi _{0}}^{\otimes 4} \in M^{\perp }.
                  \end{equation}
            \item
                  For $\ket{\Phi^z} \in \mathcal{S}_{2,( a,b,c)}^{\mathrm{paired}} ,\ket{\Phi^a} \in \mathcal{S}_{2,( a+1,b,c)}^{\mathrm{paired}} ,\ket{\Phi^b} \in
                    \mathcal{S}_{2,( a,b+1,c)}^{\mathrm{paired}} ,\ket{\Phi^c} \in
                    \mathcal{S}_{2,( a,b,c+1)}^{\mathrm{paired}}$,
                  by \Cref{lem:vio4}, \Cref{lem:038o} (\Cref{itm:9iq7}), and \Cref{lem:dbc5},
                  we have
                  \begin{multline}
                    \label{eq:96f6} (
                    -1)^{\operatorname{cr}\left(\ket{\Phi^z}\right)}\ket{\Phi^z}- (
                    -1)^{\operatorname{cr}\left(\ket{\Phi^a}\right)}\ket{\Phi^a} \\
                    - (
                    -1)^{\operatorname{cr}\left(\ket{\Phi^b}\right)}\ket{\Phi^b}- (
                    -1)^{\operatorname{cr}\left(\ket{\Phi^c}\right)}\ket{\Phi^c} \in M^{\perp }.
                  \end{multline}
          \end{itemize}
          We show that there exists functions $w,s:\mathbb{N}^3\to\mathbb{R}$, such that
          \begin{multline}
            \label{eq:irr5} \ket{\Psi ^{*}} -\ket{\psi _{0}}^{\otimes 4} =
            \frac{1}{s( 0,0,0)} \sum _{(\ref{eq:lmik})} \left(\ket{\Phi } -\ket{\psi
              _{0}}^{\otimes 4}\right) \\
            - \sum _{a,b,c} w(a,b,c)\sum _{(\ref{eq:96f6})}\left((
            -1)^{\operatorname{cr}\left(\ket{\Phi^z}\right)} \ket{\Phi^z} -(
            -1)^{\operatorname{cr}\left(\ket{\Phi^a}\right)} \ket{\Phi^a} \right. \\
            \left.  -(-1)^{\operatorname{cr}\left(\ket{\Phi^b}\right)} \ket{\Phi^b} -(
            -1)^{\operatorname{cr}\left(\ket{\Phi^c}\right)} \ket{\Phi^c}\right).
          \end{multline}
          Here the under script indicates that the summation is taken over vectors in \Cref{eq:lmik,eq:96f6}.
          Let $D:=\frac{\binom{n+4}{2}}{\binom{n}{\NumE}\binom{n+4}{\NumE+2}}$.
          Let $s( a,b,c)$ be the number of different paired states with configuration $(
            a,b,c)$:
          \begin{equation}
            s( a,b,c) :=\frac{n!
            }{( n-\NumE -a-b-c)( \NumE -a-b-c) !( a!b!c!)^{2}},
          \end{equation}
          and define
          \begin{equation}
            w(a,b,c):=\frac{a!b!c!}{s(a+1,b,c)s(a,b+1,c)s(a,b,c+1)} f(a+b+c),
          \end{equation}
          where
          \begin{equation}
            f( k) :=
            \begin{cases}
              \frac{1}{\binom{n}{\NumE}} - D, & k=0,          \\
              \frac{k}{( n-\NumE -k+1) (\NumE -k+1)}
              f(k-1)- \frac{D}{\binom{k+2}{2} k!
              } ,
                                              & k\geqslant 1.
            \end{cases}
          \end{equation}
          Comparing the coefficients of vectors on both sides of \Cref{eq:irr5}, we have to prove the following linear equations.
          \begin{align}
            D                                        & = \frac{1-s(0,0,0)}{s(0,0,0)}-w(0,0,0)s(1,0,0)s(0,1,0)s(0,0,1), \label{eq:meen} \\
            \frac{D}{\binom{k+2}{2}\binom{k}{a,b,c}} & = -w( a,b,c) s( a+1,b,c) s( a,b+1,c) s( a,b,c+1) \nonumber                      \\
                                                     & \quad +w(a-1,b,c)s(a-1,b,c)s(a-1,b+1,c)s(a-1,b,c+1) \nonumber                   \\
                                                     & \quad +w(a,b-1,c)s(a,b-1,c)s(a+1,b-1,c)s(a,b-1,c+1) \nonumber                   \\
                                                     & \quad +w(a,b,c-1)s(a,b,c-1)s(a+1,b,c-1)s(a,b+1,c-1).
            \label{eq:s88g}
          \end{align}
          Here we define $k:=a+b+c$, and $w(a,b,c)=s(a,b,c)=0$ if at least one of $a,b,c<0$.
          \Cref{eq:meen} follows from definition.
          It can be verified that the right-hand side of \Cref{eq:s88g} equals
          \begin{equation}
            a!
            b!c!
            \left(-f( k)+\frac{k}{( n-\NumE -k+1) (\NumE -k+1)} f(k-1)\right),
          \end{equation}
          which equals the left-hand side by definition.
  \end{enumerate}
\end{proof}

Notice that the moment vector $\ket{\Psi^{\vec{R}}_{2,\infty}}$ in \Cref{thm:if80} does not depend on the structure of $G$ except the connectivity.
As an immediate corollary, the $(\vec{R},2,\infty)$-moment vectors of $k$-BRA, $k$-UCCS, and $k$-UCCGS are the same, since the underlying graphs $G$, which are a path, a complete bipartite graph, and a complete graph respectively, are all connected.

\begin{corollary}
  Let $\vec{R}$ be defined in \Cref{thm:if80}.
  \begin{equation}
    \ket{\Psi ^{\mathrm{BRA}}_{2,\infty}} =
    \ket{\Psi ^{\mathrm{UCCS}}_{2,\infty}} =
    \ket{\Psi ^{\mathrm{UCCGS}}_{2,\infty}} =
    \ket{\Psi^{\vec{R}}_{2,\infty}}.
  \end{equation}
\end{corollary}

With $\ket{\Psi^{\vec{R}}_{2,\infty}}$ characterized, we can calculate the \nth{2} moment of the cost function for any observables in the $k\to\infty$ limit.
The first case of the main result is stated formally as the following corollary.

\begin{corollary}[Main result, Case 1] \label{cor:7i0g}
  Let $\vec{R}$ be defined in \Cref{thm:if80}, $C(\vec{\uptheta};U^{\vec{R}}_{k},H_{\mathrm{el}})$ be the cost function defined in \Cref{eq:wtfm}.
  Here the observable $H_{\mathrm{el}}$ is an electronic structure Hamiltonian defined in \Cref{eq:dv4e}, with coefficients $\vec{h} =( h_{pq})_{p >q} ,\vec{g} =( g_{pqrs})_{p >q >r >s}$.
  We have
  \begin{equation}
    \lim_{k\to\infty}\Var{C(\vec{\uptheta};U^{\vec{R}}_{k},H_{\mathrm{el}})}=
    \norm{\vec{h}}^2_2 \frac{4\NumE( n-\NumE)}{n( n-1)( n+2)} +\norm{\vec{g}}^2_2
    \frac{2\binom{\NumE}{2}\binom{n-\NumE}{2}}{45\binom{n+2}{6}}.
  \end{equation}
\end{corollary}

\begin{proof}
  Recall that
  \begin{multline}
    \Var{C(\vec{\uptheta};U^{\vec{R}}_{k},H_{\mathrm{el}})} =
    \E{C(\vec{\uptheta};
    U^{\vec{R}}_{k},H_{\mathrm{el}})^2} \\
    = \left(\bra{H_{\mathrm{el}}}\right)^{\otimes 2}
    \ket{\Psi^{\vec{R}}_{2,\infty}}
    =\sum_{p>q}h_{pq}^2\qty(\bra{\hat{a}_{p}^{\dagger } \hat{a}_{q} +h.c.
    })^{\otimes 2}\ket{\Psi^{\vec{R}}_{2,\infty}}                            \\
    + \sum_{p>q>r>s}g_{pqrs}^2
    \qty(\bra{\hat{a}_{p}^{\dagger} \hat{a}_{q}^{\dagger} \hat{a}_{r}
      \hat{a}_{s} +h.c.})^{\otimes 2}\ket{\Psi^{\vec{R}}_{2,\infty}}
    + \text{cross terms}.
  \end{multline}

  The contribution of cross terms is zero.
  To see that, take $\left(\bra{O_1}\otimes\bra{O_2}\right)\ket{\Psi^{\vec{R}}_{2,\infty}}$ with $O_1=\hat{a}_{p}^{\dagger} \hat{a}_{q} +h.c., O_2=\hat{a}_{p}^{\dagger} \hat{a}_{r} +h.c.
  $ ($q\ne r$) as an example.
  Notice that $P:=\prod_{i=1}^{n}\left(\frac{1}{2} \openone_{2^n}^{\otimes 2t}+\frac{1}{2} Z_{i}^{\otimes 2t}\right)$ is the orthogonal projection onto $\mathcal{H}^{\mathrm{even}}_{t}$.
  However,
  \begin{multline}
    P\left(\ket{O_1}\otimes\ket{O_2}\right) =P(O_1\otimes\openone_{2^n}\otimes O_2\otimes\openone_{2^n})\sum_{j,j'\in [2^n]} \ket{j,j,j',j'} \\ =\left(\prod_{i\in[n]\backslash\left\{q\right\}}\left(\frac{1}{2} \openone_{2^n}^{\otimes 4}+\frac{1}{2} Z_{i}^{\otimes 4}\right)\right)(O_1\otimes\openone_{2^n}\otimes O_2\otimes\openone_{2^n}) \\ \underbrace{\left(\frac{1}{2} \openone_{2^n}^{\otimes 4}-\frac{1}{2} Z_{q}^{\otimes 4}\right)\sum_{j,j'\in [2^n]} \ket{j,j,j',j'}}_{=0} = 0.
  \end{multline}
  Hence, $\ket{O_1}\otimes\ket{O_2}\in (\mathcal{H}^{\mathrm{even}}_{t})^{\perp}$.
  Consequently, $\left(\bra{O_1}\otimes\bra{O_2}\right)\ket{\Psi^{\vec{R}}_{2,\infty}}=0$ since $\ket{\Psi^{\vec{R}}_{2,\infty}}\in \mathcal{H}^{\mathrm{even}}_{t}$.
  The same argument can be extended to when one or both of $O_{1},O_{2}$ are double excitation Hermitians, except for the tricky case where both $O_{1},O_{2}$ are double excitation Hermitians and the sets of indices are the same.
  For example, when $O_{1} =\hat{a}_{1}^{\dagger}\hat{a}_{2}^{\dagger}\hat{a}_{3}\hat{a}_{4} +h.c.,O_{2} =\hat{a}_{1}^{\dagger}\hat{a}_{3}^{\dagger}\hat{a}_{2}\hat{a}_{4} +h.c\period$, $\left(\bra{O_1}\otimes\bra{O_2}\right)\ket{\Psi^{\vec{R}}_{2,\infty}}$ may not be 0.
  However, we have assumed that terms like $\hat{a}_{1}^{\dagger}\hat{a}_{3}^{\dagger}\hat{a}_{2}\hat{a}_{4} +h.c\period$ do not appear in $H_{\mathrm{el}}$ to rule out the non-zero cases since they are not essential.

  Finally, by \Cref{thm:if80} and \Cref{prp:q0g8},
  \begin{align}
    \qty(\bra{\hat{a}_{p}^{\dagger } \hat{a}_{q} +h.c.
    })^{\otimes 2}\ket{\Psi^{\vec{R}}_{2,\infty}}
     & =\frac{4\binom{n+4}{2}}{\binom{n}{\NumE}\binom{n+4}{\NumE+2}} \sum _{a=0}^{m-1} \frac{\binom{n-2}{2a}\binom{2a}{a}\binom{n-2-2a}{\NumE
    -1-a}} {\binom{a+3}{2}\binom{a+1}{1}}                                                                                                     \\
     & =\frac{4\NumE( n-\NumE)}{n( n-1)( n+2)},                                                                                               \\
    \qty(\bra{\hat{a}_{p}^{\dagger} \hat{a}_{q}^{\dagger} \hat{a}_{r}
      \hat{a}_{s} +h.c.})^{\otimes 2}\ket{\Psi^{\vec{R}}_{2,\infty}}
     & =
    \frac{4\binom{n+4}{2}}{\binom{n}{\NumE}\binom{n+4}{\NumE+2}}
    \sum _{a=0}^{m-2}
    \frac{\binom{n-4}{2a}\binom{2a}{a}\binom{n-4-2a}{\NumE -2-a}}
    {\binom{a+4}{2}\binom{a+2}{2}}                                                                                                            \\
     & =\frac{2\binom{\NumE}{2}\binom{n-\NumE}{2}}{45\binom{n+2}{6}}.
  \end{align}
  Here $m:=\min( \NumE ,n-\NumE)$.
\end{proof}

Notice that $\lim _{k\to \infty }\Var{C}=1/\operatorname{poly}(n)$ if $h_{pq},g_{pqrs}\in O(1)$.
\section{Proof of main result: Case 2}\label{app:mbe3}

In this section, we prove the polynomial/exponential concentration of the cost function for alternated dUCC \ansatze{} containing only qubit single excitation rotations, with a mild connectivity assumption (\Cref{thm:8pgp} (\Cref{itm:ka9j})).
Examples of such \ansatzes{} include $k$-qubit-UCCS and $k$-qubit-UCCGS.

\begin{theorem}
  \label{thm:2tcb}
  Let $G=(V,E)$ be a connected graph with $|V|=n$ vertices, $\vec{R}$ be a sequence of qubit single excitation rotations $(A^{\mathrm{qubit}}_{uv})_{(u,v)\in E}$.
  Define $m := \min\left\{\NumE, n-\NumE\right\}$.
  Depending on the structure of $G$, the moment vector $\ket{\Psi^{\vec{R}}_{2,\infty}}$ admits one of the following forms.
  \begin{enumerate}
    \item
          \label{itm:fwyh}
          If $G$ is a path or a ring, then $\ket{\Psi^{\vec{R}}_{2,\infty}} =S_{\pi} \ket{\Psi _{2,\infty }^{\mathrm{BRA}}}$ with $\pi\in \mathfrak{S}_n$ and $(\pi(i),\pi(i+1))\in E$ for all $i\in [n-1]$.

    \item
          \label{itm:y7p4}
          If $n=2 \NumE$, $G$ is bipartite, and both parts of $G$ have an even size (i.e.,
          let $V_1\cup V_2$ be the unique partition of vertices such that both $V_1$ and
          $V_2$ are independent sets, $|V_1|$ and $|V_2|$ are both even), then
          $\ket{\Psi^{\vec{R}}_{2,\infty}}$ can be written in the form of \Cref{eq:qydw},
          with
          \begin{equation}
            \label{eq:ztmg} c( a,b,c) =
            \begin{cases}
              D,                                                                  & \text{if } a=b=c=0,                                                    \\
              \frac{D}{3} ,                                                       & \text{if exactly one of } a,b,c\text{ is nonzero, and } a+b+c< \NumE , \\
              \frac{\binom{n}{\NumE} -2}{\binom{n}{\NumE} +2} \cdot \frac{D}{3} , & \text{if exactly one of } a,b,c\text{ is nonzero, and } a+b+c=\NumE ,  \\
              \frac{2}{\binom{n}{\NumE} +2} \cdot \frac{D}{3} ,                   & \text{if exactly two of } a,b,c\text{ is nonzero, and } a+b+c=\NumE ,  \\
              0,                                                                  & \text{otherwise} .
            \end{cases}
          \end{equation}
          And for $\ket{\Phi } \in \mathcal{S}_{2,( a,b,c)}^{\mathrm{paired}}$,
          \begin{equation}
            \label{eq:vbgx}
            \begin{split}
               & \operatorname{sign}\left(\ket{\Phi }\right) \\
               & =
              \begin{cases}
                -( -1)^{\sum _{p< q\in V_{1}} \Phi _{p} \odot \Phi _{q}} , &
                \text{if exactly two of } a,b,c\text{ is nonzero, and } a+b+c=\NumE , \\
                1,                                                         &
                \text{otherwise} .
              \end{cases}
            \end{split}
          \end{equation}
          Here we define $D:=\frac{3\binom{n}{\NumE} +6}{\binom{n}{\NumE}^{3} +4\binom{n}{\NumE}^{2}}$.

    \item
          \label{itm:tany}
          Otherwise, $\ket{\Psi^{\vec{R}}_{2,\infty}}$ is given in \Cref{eq:47vm}.
  \end{enumerate}
\end{theorem}

As an immediate corollary, we can determine the $(\vec{R},2,\infty)$-moment vector of $k$-qubit-UCCS and $k$-qubit-UCCGS, where the underlying graph $G$ is a complete bipartite graph and a complete graph, respectively.

\begin{corollary}
  $\ket{\Psi^{\mathrm{qUCCS}}_{2,\infty}}$ is determined by \Cref{thm:2tcb} (\ref{itm:y7p4}),
  and $\ket{\Psi^{\mathrm{qUCCGS}}_{2,\infty}}$ is determined by \Cref{thm:2tcb} (\ref{itm:tany}).
\end{corollary}

With $\ket{\Psi^{\vec{R}}_{2,\infty}}$ characterized, we can calculate the \nth{2} moment of the cost function for any observables in the $k\to\infty$ limit.
Case 2 of the main result is stated formally as the following corollary.

\begin{corollary}[Main result, Case 2] \label{cor:e79n}
  Let $\vec{R},G$ be defined in \Cref{thm:2tcb}, and $C(\vec{\uptheta};U^{\vec{R}}_{k},H_{\mathrm{el}})$ be the cost function defined in \Cref{eq:wtfm}.
  Here the observable $H_{\mathrm{el}}$ is an electronic structure Hamiltonian defined in \Cref{eq:dv4e}, with coefficients $\vec{h} =( h_{pq})_{p >q} ,\vec{g} =( g_{pqrs})_{p >q >r >s}$.
  \begin{enumerate}
    \item
          \label{itm:ortd}
          If $G$ is a path or a ring, then $\lim_{k\to\infty}\Var{C}$ is the same as in \Cref{cor:7i0g}.
    \item
          \label{itm:agri}
          If $n=2 \NumE$, $G$ is bipartite, and both parts of $G$ have an even size, then,
          \begin{multline}
            \lim _{k\rightarrow \infty }\Var{C(\vec{\uptheta}
            ;U_{k}^{\vec{R}} ,H_{\mathrm{el}} )}=\sum _{p >q}
            h_{pq}^{2}\frac{4\binom{n-2}{\NumE -1}\qty[\binom{n}{\NumE} +2\qty(1-( -1)^{[
                        q\leq \NumE < p] +p-q})]}{\binom{n}{\NumE}^{3} +4\binom{n}{\NumE}^{2}}\\
            +\sum
            _{p >q >r >s} g_{pqrs}^{2}\frac{4\binom{n-4}{\NumE -2}\qty[\binom{n}{\NumE}
                +2\qty(1-( -1)^{[ q\leq \NumE < p\lor s\leq \NumE < r]
                    +p-q+r-s})]}{\binom{n}{\NumE}^{3} +4\binom{n}{\NumE}^{2}} .
          \end{multline}
          Here we define $[ P] =1$ if a proposition $P$ is true and $[ P] =0$ otherwise.
    \item
          \label{itm:v8lq}
          Otherwise,
          \begin{equation}
            \lim _{k\rightarrow \infty }\Var{C(\vec{\uptheta}
            ;U_{k}^{\vec{R}} ,H_{\mathrm{el}} )}=\norm{\vec{h}}_2^2\frac{4\binom{n-2}{\NumE
                -1}}{\binom{n}{\NumE}^{2} +2\binom{n}{\NumE}}
            +\norm{\vec{g}}_2^2\frac{4\binom{n-4}{\NumE -2}}{\binom{n}{\NumE}^{2}
              +2\binom{n}{\NumE}} .
          \end{equation}
  \end{enumerate}
\end{corollary}

\begin{proof}
  Similar to the proof of \Cref{cor:7i0g}, it suffices to evaluate $\left(\bra{O}^{\otimes 2}\right)\ket{\Psi _{2,\infty }^{\vec{R}}}$ for $O=\hat{a}_{p}^{\dagger }\hat{a}_{q} +h.c.
  $ and $O=\hat{a}_{p}^{\dagger }\hat{a}_{q}^{\dagger }\hat{a}_{r}\hat{a}_{s} +h.c.$.

  \Cref{itm:ortd}.
  Case 1 is already proved in \Cref{cor:7i0g}.

  \Cref{itm:agri}.
  Let $D=\frac{3\binom{n}{\NumE} +6}{\binom{n}{\NumE}^{3} +4\binom{n}{\NumE}^{2}}$.
  By \Cref{thm:2tcb} and \Cref{prp:q0g8},
  \begin{equation}
    \begin{split}
       & \qty(\bra{\hat{a}_{p}^{\dagger }\hat{a}_{q} +h.c.
      })^{\otimes 2}\ket{\Psi _{2,\infty }^{\vec{R}}}                          \\
       & =\binom{n-2}{\NumE -1} \cdot \frac{D}{3} -( -1)^{[ q\leq \NumE < p] +
      p-q}\binom{n-2}{\NumE -1} \cdot \frac{2}{\binom{n}{\NumE} +2}\frac{D}{3} \\
       & =\frac{4\binom{n-2}{\NumE -1}\qty[\binom{n}{\NumE} +2\qty(1-( -1)^{
              [ q\leq \NumE < p] +p-q })]}{\binom{n}{\NumE}^{3} +4\binom{n}{\NumE}^{2}},
    \end{split}
  \end{equation}
  \begin{equation}
    \begin{split}
       & \qty(\bra{\hat{a}_{p}^{\dagger }\hat{a}_{q}^{\dagger }\hat{a}_{r}\hat{a}_{s}
      +h.c.})^{\otimes 2}\ket{\Psi _{2,\infty }^{\vec{R}}}                            \\
       & =\binom{n-4}{\NumE -2} \cdot \frac{D}{3} -( -1)^{[ q\leq \NumE < p
              \lor s\leq \NumE < r] +p-q+r-s}\binom{n-4}{\NumE -2} \cdot \frac{2}{
      \binom{n}{\NumE} +2 }\frac{D}{3}                                                \\
       & =\frac{4\binom{n-4}{\NumE -2}\qty[\binom{n}{\NumE} +2\qty(1-( -1)^{
              [ q\leq \NumE < p\lor s\leq \NumE < r] +p-q+r-s })]}{\binom{n}
        {\NumE}^{3} +4\binom{n}{\NumE}^{2}} .
    \end{split}
  \end{equation}

  \Cref{itm:v8lq}.
  By \Cref{thm:2tcb} and \Cref{prp:q0g8},
  \begin{align}
    \qty(\bra{\hat{a}_{p}^{\dagger }\hat{a}_{q} +h.c.
    })^{\otimes 2}\ket{\Psi _{2,\infty }^{\vec{R}}}      & =\frac{4\binom{n-2}{\NumE -1}}
    {\binom{n}{\NumE}^{2} +2\binom{n}{\NumE}} ,                                           \\
    \qty(\bra{\hat{a}_{p}^{\dagger }\hat{a}_{q}^{\dagger }\hat{a}_{r}\hat{a}_{s}
    +h.c.})^{\otimes 2}\ket{\Psi _{2,\infty }^{\vec{R}}} & =\frac{4\binom{n-4}{\NumE
        -2}}{\binom{n}{\NumE}^{2} +2\binom{n}{\NumE}}.
  \end{align}
\end{proof}

Notice that if $\NumE=\Theta(n)$, $\lim _{k\to \infty }\Var{C}=1/\operatorname{poly}(n)$ in case \Cref{itm:ortd}, and $\lim _{k\to \infty }\Var{C}=\exp(-\Theta(n))$ in case \Cref{itm:agri,itm:v8lq}.

\subsection{Proof of Theorem \ref{thm:2tcb} (\ref{itm:fwyh})}

\begin{proof}[Proof of \Cref{thm:2tcb} (\Cref{itm:fwyh})]
  The basis rotation \ansatze{} consists of single excitation rotations acting on
  edges of a path that connects all neighboring qubits, i.e., $E=\qty{( i,i+1) | i\in [
        n-1]}$.
  When acting on two neighboring qubits, single excitations are equivalent to qubit single excitations ($\hat{a}_{i+1}^{\dagger }\hat{a}_{i} =Q_{i+1} Q_{i}$).
  Thus, when $G$ is a path connecting all neighboring qubits, $\ket{\Psi^{\vec{R}}_{2,\infty}}$ is exactly $\ket{\Psi _{2,\infty }^{\text{BRA}}}$.
  Suppose $G$ is an arbitrary path.
  Fix one of the two corresponding permutations $\pi \in \mathfrak{S}_{n}$, such that $( \pi ( i) ,\pi ( i+1)) \in E,\forall i\in [ n-1]$.
  By examining the proof of \Cref{thm:if80}, it can be checked that the moment vector remains unchanged if one replaces the initial state $\ket{\psi _{0}}^{\otimes 4}$ by any paired state in $\mathcal{S}_{2,(0,0,0)}^{\mathrm{paired}}$.
  Thus,
  \begin{equation}
    \label{eq:kb83}
    \begin{split}
      S_{\pi }\ket{\Psi _{2,\infty }^{\text{BRA}}} = S_{\pi }\lim _{k\rightarrow \infty }\prod _{l=1}^{k}\prod _{i=1}^{n-1}\tilde{A}_{i,i+1} \ket{\psi_{0}}^{\otimes 4} & = S_{\pi }\lim _{k\rightarrow \infty }\prod _{l=1}^{k}\prod _{i=1}^{n-1}\tilde{A}^{\mathrm{qubit}}_{i,i+1} \qty(S_{\pi }^{-1}\ket{\psi _{0}}^{\otimes 4})                      \\
                                                                                                                                                                        & =\lim _{k\rightarrow \infty }\prod _{l=1}^{k}\prod _{i=1}^{n-1}\tilde{A}^{\mathrm{qubit}}_{\pi ( i) ,\pi ( i+1)}\ket{\psi _{0}}^{\otimes 4} =\ket{\Psi^{\vec{R}}_{2,\infty}} .
    \end{split}
  \end{equation}

  To extend the conclusion to rings, it suffices to show the equality for the ``standard'' ring, i.e., $E=\{( i,i+1) | i\in [ n-1]\} \cup \{( 1,n)\}$, and argue for arbitrary rings similar to \Cref{eq:kb83}.
  Let $P_{M_{1}} =\tilde{A}_{1n} ,P_{M_{2}} =\tilde{A}_{1n}^{\mathrm{qubit}}$.
  It suffices to prove that $M_{1} \cap \mathcal{H}_{2}^{\mathrm{paired}} =M_{2} \cap \mathcal{H}_{2}^{\mathrm{paired}}$.
  If so, one can replace $A_{1n}^{\mathrm{qubit}}$ by $A_{1n}$ without changing $\E{C^{2}}$.
  Since $A_{i,i+1}^{\mathrm{qubit}}$ can also be replaced by $A_{i,i+1}$ for $i\in [ n-1]$, one concludes that $\ket{\Psi _{2,\infty }^{\vec{R}}} =\ket{\Psi _{2,\infty }^{\mathrm{BRA}}}$ according to \Cref{thm:if80}.
  Equivalently, we prove that the spanning sets are the same for $\left(M_{1} \cap \mathcal{H}_{2}^{\mathrm{paired}}\right)^{\perp } \cap \mathcal{H}_{2}^{\mathrm{paired}}$ and $\left(M_{2} \cap \mathcal{H}_{2}^{\mathrm{paired}}\right)^{\perp } \cap \mathcal{H}_{2}^{\mathrm{paired}}$.
  By comparing \Cref{itm:xojb} and \Cref{itm:rtjl} in \Cref{lem:vio4}, we need to prove for any $\ket{\Phi } \in \mathcal{S}_{2}^{\mathrm{paired}}$ with $\vec{z} :=\bigoplus _{1< i< n} \Phi _{i}$, that (1) $( \Phi _{1} \oplus \vec{z}) \odot ( \Phi _{n} \oplus \vec{z}) \equiv \Phi _{1} \odot \Phi _{n}\pmod{2}$, (2) if $\ket{\Phi _{1}} =\ket{I_{00}} ,\ket{\Phi _{n}} =\ket{I_{11}}$, then $z_{1} +z_{2} \equiv z_{1} +z_{3} \equiv z_{2} +z_{3}\pmod{2}$.

  \begin{enumerate}[(1)]
    \item
          Since $\ket{\Phi }$ is a paired state, we have $\ket{\vec{z} \oplus \Phi _{1}
              \oplus \Phi _{n}} \in \left\{\ket{I_{00}} ,\ket{I_{11}}\right\}$.
          Hence,
          \begin{equation}
            \Phi _{1} \odot \Phi _{n} \equiv ( \Phi _{1} \oplus
              (\vec{z} \oplus \Phi _{1} \oplus \Phi _{n})) \odot ( \Phi _{n} \oplus (\vec{z}
              \oplus \Phi _{1} \oplus \Phi _{n})) \equiv ( \Phi _{1} \oplus \vec{z}) \odot (
            \Phi _{n} \oplus \vec{z})\pmod{2}.
          \end{equation}

    \item
          Since $\ket{\Phi }$ is a paired state, and $\ket{\Phi _{1}} =\ket{I_{00}}
            ,\ket{\Phi _{n}} =\ket{I_{11}}$, we have $\ket{\vec{z}} \in \qty{\ket{I_{00}}
              ,\ket{I_{11}}}$.
          Hence, $z_{1} +z_{2} \equiv z_{1} +z_{3} \equiv z_{2} +z_{3}\pmod{2}$.
  \end{enumerate}
\end{proof}
\subsection{Proof of Theorem \ref{thm:2tcb} (\ref{itm:y7p4})}

\begin{lemma}
  \label{lem:xg0j}
  Let $G,\operatorname{sign},c$ be defined in \Cref{thm:2tcb} (\Cref{itm:y7p4}).
  Suppose $( u,v)$ is an edge of $G$.
  For any $\ket{\Phi } \in \mathcal{S}_{2,(a,b,c)}^{\mathrm{paired}}$, if $c(a,b,c)\neq 0$, then $\operatorname{sign}\left(S_{uv}\ket{\Phi }\right) =( -1)^{\Phi _{u} \odot \Phi _{v}}\operatorname{sign}\left(\ket{\Phi }\right)$.
\end{lemma}

\begin{proof}
  Let $\ket{\Phi '} =S_{uv}\ket{\Phi }$.
  Recall that $G$ is a bipartite graph, and the two parts $V_{1} ,V_{2}$ of $G$ are both even.
  Assume without loss of generality that $u\in V_{1},v\in V_2$.
  By definition, $\operatorname{sign}\left(\ket{\Phi }\right) /\operatorname{sign}\left(\ket{\Phi '}\right) =( -1)^{( \Phi _{u} \oplus \Phi _{v}) \odot \bigoplus _{w\in V_{1} \backslash \{v\}} \Phi _{w}}$.
  We prove by comparing $\operatorname{sign}\left(\ket{\Phi }\right) /\operatorname{sign}\left(\ket{\Phi '}\right)$ and $( -1)^{\Phi _{u} \odot \Phi _{v}}$.

  \begin{itemize}
    \item
          If $\Phi_u=\bar{\Phi}_v$ or $\Phi_u=\Phi_v$, then
          $\Phi _{u} \oplus \Phi _{v} \in \left\{\ket{I_{00}} ,\ket{I_{11}}\right\}$.
          Hence, $\operatorname{sign}\left(\ket{\Phi }\right) /\operatorname{sign}\left(\ket{\Phi '}\right)=1$.
          Meanwhile, $(-1)^{\Phi_u\odot \Phi_v}=1$.

    \item
          If one of $\Phi_u,\Phi_v$ is $\ket{I_{00}}$ or $\ket{I_{11}}$, then by the
          definition of $c( a,b,c)$, at most one of $a,b,c$ is non-zero.
          In such case, we have $\Phi _{p} \odot \Phi _{q} \equiv 0\pmod{2}$ for any $p,q\in V$.
          Hence, $\operatorname{sign}\left(\ket{\Phi}\right) /\operatorname{sign}\left(\ket{\Phi '}\right) =1$.
          Meanwhile, $( -1)^{\Phi _{u} \odot \Phi _{v}} =1$.

    \item
          Otherwise, it must be the case that exactly two of $a,b,c$ are non-zero,
          and $\Phi_u \notin \left\{\Phi_v,\bar{\Phi}_v\right\}$.
          We have $\operatorname{sign}\left(\ket{\Phi }\right) /\operatorname{sign}\left(\ket{\Phi '}\right)=-1$, since $\bigoplus _{w\in V_{1} \backslash u} \Phi _{w}$ must be one of $\Phi _{u} ,\overline{\Phi }_{u} ,\Phi _{v} ,\overline{\Phi }_{v}$, by the fact that $|V_1|$ is even.
          Meanwhile, $( -1)^{\Phi _{u} \odot \Phi _{v}} =-1$.
  \end{itemize}
\end{proof}

\begin{lemma}
  \label{lem:z3zj}
  Let $G,\operatorname{sign},c$ be defined in \Cref{thm:2tcb} (\Cref{itm:y7p4}).
  Suppose $( u,v)$ is an edge of $G$.
  For any $\ket{\Phi^z} \in \mathcal{S}_{2,(a,b,c)}^{\mathrm{paired}}$, with $\ket{\Phi^z_{u}} =\ket{I_{00}} ,\ket{\Phi^z_{v}} =\ket{I_{11}}$, let $\ket{\Phi^a} =F_{uv}^{12}\ket{\Phi^z} ,\ket{\Phi^b} =F_{uv}^{13}\ket{\Phi^z} ,\ket{\Phi^c} =F_{uv}^{23}\ket{\Phi^z}$.
  We have
  \begin{multline}
    \label{eq:lonl} \operatorname{sign}\left(\ket{\Phi^z}\right) c( a,b,c) =\operatorname{sign}\left(\ket{\Phi^a}\right) c( a+1,b,c) \\
    +\operatorname{sign}\left(\ket{\Phi^b}\right) c( a,b+1,c)
    +\operatorname{sign}\left(\ket{\Phi^c}\right) c( a,b,c+1) .
  \end{multline}
\end{lemma}

\begin{proof}
  We prove this by evaluating the values of $\operatorname{sign}$ and $c$ in \Cref{eq:lonl}.
  \begin{itemize}
    \item
          If $a=b=c=0$, the signs of $\ket{\Phi^z} ,\ket{\Phi^a} ,
            \ket{\Phi^b} ,\ket{\Phi^c}$ are all 1, and \Cref{eq:lonl}
          follows from the equality $D=\frac{D}{3} +\frac{D}{3} +\frac{D}{3}$.

    \item
          If exactly one of $a,b,c$ is non-zero, and $a+b+c< \NumE -1$,
          the signs of $\ket{\Phi^z} ,\ket{\Phi^a} ,\ket{\Phi^b} ,
            \ket{\Phi^c}$ are all 1, and \Cref{eq:lonl} follows from the
          equality $\frac{D}{3} =\frac{D}{3} +0+0$.

    \item
          If exactly one of $a,b,c$ is non-zero, and $a+b+c=\NumE -1$,
          the signs of $\ket{\Phi^z} ,\ket{\Phi^a} ,\ket{\Phi^b} ,
            \ket{\Phi^c}$ are all 1.
          To see that, we may assume without loss of generality that $u\in V_{1}$ and $v\in V_{2}$, and $a=\NumE -1$.
          By definition, $\operatorname{sign}\left(\ket{\Phi^z}\right)
            =\operatorname{sign}\left(\ket{\Phi^a}\right) =1$, and
          \begin{equation}
            \operatorname{sign}\left(\ket{\Phi^b}\right) =-( -1)^{\sum _{p< q\in V_{1}} \Phi^b_{p} \odot \Phi^b_{q}} =-( -1)^{\Phi^b_{u} \odot \bigoplus _{w\in
              V_{1} \backslash u} \Phi^b_{w}} =1,
          \end{equation}
          since $\ket{\Phi^b_{v}} \in \left\{\ket{X_{00}} ,\ket{X_{11}}\right\}$ and $\ket{\Phi^b_{w}} \in
            \left\{\ket{I_{01}} ,\ket{I_{10}}\right\}$ for $w\in V_{1} \backslash \{v\}$, and
          $|V_{1} |$ is even.
          Similarly, $\operatorname{sign}\left(\ket{\Phi^c}\right) =1$.
          \Cref{eq:lonl} follows from the equality
          \begin{equation}
            \frac{D}{3} =\frac{2}{\binom{n}{\NumE} +2} \cdot \frac{D}{3} +
            \frac{2}{\binom{n}{\NumE} +2} \cdot \frac{D}{3} +
            \frac{\binom{n}{\NumE} -2}{\binom{n}{\NumE} +2} \cdot \frac{D}{3} .
          \end{equation}

    \item
          Otherwise, we may assume that exactly two of $a,b,c$ are non-zero,
          and $a+b+c=\NumE -1$, since in other cases the coefficients of
          $\ket{\Phi^z} ,\ket{\Phi^a} ,\ket{\Phi^b} ,\ket{\Phi^c}$
          are all zero.
          Assume without loss of generality that $c=0$, $v\in V_{1}$, and $u\in V_{2}$.
          By definition,
          \begin{equation}
            \label{eq:960l}
            \operatorname{sign}\left(\ket{\Phi^a}\right) /\operatorname{sign}\left(\ket{\Phi^b}\right) =( -1)^{\left(\Phi^a_{v} \oplus \Phi^b_{v}\right) \odot \bigoplus _{w\in
              V_{1} \backslash \{v\}} \Phi^a_{w}} =-1,
          \end{equation}
          since $\ket{\Phi^a_{v}} \in \left\{\ket{I_{01}} ,\ket{I_{10}}\right\}$, $\ket{\Phi^b_{v}} \in
            \left\{\ket{X_{00}} ,\ket{X_{11}}\right\}$, $\ket{\Phi^a_{w}} \in \qty{\ket{I_{01}}
              ,\ket{I_{10}} ,\ket{X_{00}} ,\ket{X_{11}}}$ for $w\in V_{1} \backslash \{v\}$,
          and $|V_{1} |$ is even.
          \Cref{eq:lonl} follows from the equality
          \begin{equation}
            0=\frac{2}{\binom{n}{\NumE} +2} \cdot \frac{D}{3} -\frac{2}{
              \binom{n}{\NumE} +2} \cdot \frac{D}{3} +0.
          \end{equation}
  \end{itemize}
\end{proof}

\begin{lemma}
  \label{lem:vsdf}
  Let $G=(V,E)$ be a simple, connected graph with $|V|=n$ vertices and $\max_{v\in V}\deg( v) \geqslant 3$.
  Suppose $\vec{R}$ contains a subsequence of qubit single excitation rotations $(A_{uv}^{\mathrm{qubit}} )_{(u,v)\in E}$.
  Let $M$ be the intersection space of $\vec{R}$ defined in \Cref{cor:9rac} at $t=2$.
  For any $\ket{\Phi } \in \mathcal{S}_{2,( a,b,c)}^{\mathrm{paired}}$, if one of the following two cases happens, then $\ket{\Phi } \in M^{\perp }$.
  \begin{enumerate}
    \item
          Two of $a,b,c$ are non-zero and $2( a+b+c) < n$.
    \item
          All of $a,b,c$ are non-zero.
  \end{enumerate}
\end{lemma}

\begin{proof}
  Fix a vertex $v\in V$ such that $\deg( v) \geqslant 3$, and pick $u_{1} ,u_{2} ,u_{3}$ from the neighborhood of $v$.
  Then $v,u_{1} ,u_{2} ,u_{3}$ form a star motif centered at $v$.

  (1)
  Assume without loss of generality that $a,b >0$.
  Since $2( a+b+c) < n$, there is at least one $\ket{I_{00}}$ or $\ket{I_{11}}$ in $\ket{\Phi }$.
  We may assume that
  \begin{equation}
    \ket{\Phi _{v}} =\ket{I_{01}} ,\ket{\Phi
      _{u_{1}}} \in \left\{\ket{I_{00}} ,\ket{I_{11}}\right\} ,\ket{\Phi _{u_{2}}}
    =\ket{X_{00}} ,\ket{\Phi _{u_{3}}} =\ket{X_{11}} .
  \end{equation}
  This is because by \Cref{cor:h168}, once we prove $\ket{\Phi } \in M^{\perp }$, we can argue that $\ket{\Phi '} \in M^{\perp }$ for any $\ket{\Phi '} \in \mathcal{S}_{2,( a,b,c)}^{\mathrm{paired}}$.
  Consider the vector sequence $\ket{\Phi ^{( 0)}} =\ket{\Phi }$, and
  \begin{align}
    \ket{\Phi^{(1)}} & = S_{vu_{1}} \ket{\Phi},       & \ket{\Phi^{(2)}} & = S_{vu_{2}} \ket{\Phi^{(1)}},                 & \ket{\Phi^{(3)}} & = S_{vu_{1}} \ket{\Phi^{(2)}}, \label{eq:ek71} \\
    \ket{\Phi^{(4)}} & = S_{vu_{3}} \ket{\Phi^{(3)}}, & \ket{\Phi^{(5)}} & = S_{vu_{2}} \ket{\Phi^{(4)}},                 & \ket{\Phi^{(6)}} & = S_{vu_{1}} \ket{\Phi^{(5)}},                 \\
    \ket{\Phi^{(7)}} & = S_{vu_{2}} \ket{\Phi^{(6)}}, & \ket{\Phi^{(8)}} & = S_{vu_{3}} \ket{\Phi^{(7)}}. \label{eq:f91f}
  \end{align}
  It can be verified that $\ket{\Phi ^{( 8)}} =\ket{\Phi }$.
  By \Cref{lem:vio4} (\Cref{itm:rtjl}), the following vectors lie in $\left(M\cap \mathcal{H}_{2}^{\mathrm{paired}}\right)^{\perp } \cap \mathcal{H}_{2}^{\mathrm{paired}} \subseteq M^{\perp }$.
  \begin{align}
    \ket{\Phi ^{( 0)}} -\ket{\Phi ^{( 1)}} , &  & \ket{\Phi ^{( 1)}} -\ket{\Phi ^{( 2)}} , &  & \ket{\Phi ^{( 2)}} +\ket{\Phi ^{( 3)}} , &  & \ket{\Phi ^{( 3)}} +\ket{\Phi ^{( 4)}} , \\
    \ket{\Phi ^{( 4)}} -\ket{\Phi ^{( 5)}} , &  & \ket{\Phi ^{( 5)}} -\ket{\Phi ^{( 6)}} , &  & \ket{\Phi ^{( 6)}} -\ket{\Phi ^{( 7)}} , &  & \ket{\Phi ^{( 7)}} +\ket{\Phi ^{( 8)}} .
  \end{align}
  Combining these vectors to eliminate $\ket{\Phi ^{( 1)}} ,\ket{\Phi ^{( 2)}} ,\dotsc ,\ket{\Phi ^{( 7)}}$, we get $\ket{\Phi } \in M^{\perp }$, since $\ket{\Phi ^{( 0)}} =\ket{\Phi ^{( 8)}} =\ket{\Phi }$.

  (2)
  Similarly, assume that
  \begin{equation}
    \ket{\Phi _{v}} =\ket{I_{01}} ,\ket{\Phi
      _{u_{1}}}= \ket{X_{01}},\ket{\Phi _{u_{2}}}
    =\ket{X_{00}} ,\ket{\Phi _{u_{3}}} =\ket{X_{11}} .
  \end{equation}
  Consider the vector sequence in \Crefrange{eq:ek71}{eq:f91f}.
  By \Cref{lem:vio4} (\Cref{itm:rtjl}), the following vectors lie in $\left(M\cap \mathcal{H}_{2}^{\mathrm{paired}}\right)^{\perp } \cap \mathcal{H}_{2}^{\mathrm{paired}} \subseteq M^{\perp }$.
  \begin{align}
    \ket{\Phi ^{( 0)}} +\ket{\Phi ^{( 1)}} , &  & \ket{\Phi ^{( 1)}} +\ket{\Phi ^{( 2)}} , &  & \ket{\Phi ^{( 2)}} +\ket{\Phi ^{( 3)}} , &  & \ket{\Phi ^{( 3)}} +\ket{\Phi ^{( 4)}} , \\
    \ket{\Phi ^{( 4)}} +\ket{\Phi ^{( 5)}} , &  & \ket{\Phi ^{( 5)}} +\ket{\Phi ^{( 6)}} , &  & \ket{\Phi ^{( 6)}} -\ket{\Phi ^{( 7)}} , &  & \ket{\Phi ^{( 7)}} +\ket{\Phi ^{( 8)}} .
  \end{align}
  Combining these vectors to eliminate $\ket{\Phi ^{( 1)}} ,\ket{\Phi ^{( 2)}} ,\dotsc ,\ket{\Phi ^{( 7)}}$, we get $\ket{\Phi } \in M^{\perp }$, since $\ket{\Phi ^{( 0)}} =\ket{\Phi ^{( 8)}} =\ket{\Phi }$.
\end{proof}

\begin{lemma}
  \label{lem:ls2d}
  Same conditions as \Cref{thm:2tcb} (\Cref{itm:y7p4}).
  For $\ket{\Phi},\ket{\Phi'}\in\mathcal{S}_{2}^{\mathrm{paired}}$, if $\operatorname{conf}(\ket{\Phi})=\operatorname{conf}(\ket{\Phi'})$, then $\operatorname{sign}\left(\ket{\Phi }\right)\ket{\Phi } -\operatorname{sign}\left(\ket{\Phi '}\right)\ket{\Phi '} \in M^{\perp }$.
\end{lemma}

\begin{proof}
  Suppose $\operatorname{conf}(\ket{\Phi})=(a,b,c)$.
  If $(a,b,c)$ satisfies the conditions in \Cref{lem:vsdf}, then $\ket{\Phi},\ket{\Phi'}\in M^{\perp}$, hence $( -1)^{\operatorname{sign}\left(\ket{\Phi }\right)}\ket{\Phi } -( -1)^{\operatorname{sign}\left(\ket{\Phi '}\right)}\ket{\Phi '} \in M^{\perp }$.
  Otherwise, $c(a,b,c)\neq 0$ by definition in \Cref{eq:ztmg}.
  The rest of the proof is similar to that of \Cref{lem:dbc5}, using \Cref{lem:vio4} (\Cref{itm:rtjl}), \Cref{lem:xg0j} and connectivity of $G$.
\end{proof}

\begin{lemma}
  \label{lem:fsi5}
  Same condition as \Cref{thm:2tcb} (\Cref{itm:y7p4}).
  For any $\ket{\Phi } \in \mathcal{S}_{2,( a,b,c)}^{\mathrm{paired}} ,\ket{\Phi '} \in \mathcal{S}_{2,( a',b',c')}^{\mathrm{paired}}$, if one of the following 3 cases happens, we have $\operatorname{sign}\left(\ket{\Phi }\right)\ket{\Phi } -\operatorname{sign}\left(\ket{\Phi '}\right)\ket{\Phi '} \in M^{\perp }$.
  \begin{enumerate}
    \item
          $a+b=a'+b'=\NumE$ and $a,b,a',b'>0$.
    \item
          $a+c=a'+c'=\NumE$ and $a,c,a',c'>0$.
    \item
          $b+c=b'+c'=\NumE$ and $b,c,b',c'>0$.
  \end{enumerate}
\end{lemma}

\begin{proof}
  We prove the first case since the other 2 cases are similar.
  It suffices to prove that for any $0< l< \NumE -1$, if $a=l+1,b=\NumE -a,a'=l,b'=\NumE -a'$, then $\operatorname{sign}\left(\ket{\Phi }\right)\ket{\Phi } -\operatorname{sign}\left(\ket{\Phi '}\right)\ket{\Phi '} \in M^{\perp }$.
  Furthermore, we may assume that there exists $( u,v) \in E$ and $\ket{\Phi^z} \in \mathcal{S}_{2,( l,\NumE -l-1,0)}^{\mathrm{paired}}$ such that $\ket{\Phi } =F_{uv}^{12}\ket{\Phi^z} ,\ket{\Phi '} =F_{uv}^{13}\ket{\Phi^z}$, since one can then use \Cref{lem:ls2d} to argue for general cases.
  By definition of $G$, there exists $R_{j} \in \vec{R}$ such that $R_{j} =A_{uv}^{\mathrm{qubit}}$.
  By \Cref{lem:vio4} (\Cref{itm:rtjl}) and \Cref{lem:hpti}, $\ket{\Phi^z} -\ket{\Phi } -\ket{\Phi '} -F_{uv}^{23}\ket{\Phi^z} \in \left( M_{j} \cap \mathcal{H}_{2}^{\mathrm{paired}}\right)^{\perp } \cap \mathcal{H}_{2}^{\mathrm{paired}} \subseteq M^{\perp }$.
  Since $\operatorname{conf}\left(\ket{\Phi^z}\right) =( l,\NumE -l-1,0) ,\operatorname{conf}\left( F_{uv}^{23}\ket{\Phi^z}\right) =( l,\NumE -l-1,1)$, we have $\ket{\Phi^z} ,F_{uv}^{23}\ket{\Phi^z} \in M^{\perp }$ according to \Cref{lem:vsdf}.
  Hence, $\ket{\Phi } +\ket{\Phi '} \in M^{\perp }$.
  It remains to show that $\operatorname{sign}\left(\ket{\Phi }\right) /\operatorname{sign}\left(\ket{\Phi '}\right) =-1$, which has been proved in \Cref{eq:960l}.
\end{proof}

\begin{proof}[Proof of \Cref{thm:2tcb}
    (\Cref{itm:y7p4})]
  Denote the state specified in \Crefrange{eq:ztmg}{eq:vbgx} by $\ket{\Psi^*}$.
  Use the notations $M,M_j$ from \Cref{cor:9rac}.
  We proceed similarly to the proof of \Cref{thm:if80} by showing that (1) $\ket{\Psi ^{*}} \in M$ and (2) $\ket{\Psi ^{*}} -\ket{\psi _{0}}^{\otimes 4} \in M^{\perp }$.

  \begin{enumerate}[(1)]
    \item
          Obviously, $\ket{\Psi ^{*}} \in \mathcal{H}_{2}^{\mathrm{paired}}$.
          We prove that $\ket{\Psi ^{*}}$ is orthogonal to $\left(M_{j} \cap \mathcal{H}_{2}^{\mathrm{paired}}\right)^{\perp} \cap \mathcal{H}_{2}^{\mathrm{paired}}$, for each $R_{j} \in \vec{R}$.
          If so, we have $\ket{\Psi ^{*}} \in M$ by \Cref{lem:hpti}.
          Recall the spanning set of $\left(M_{j} \cap \mathcal{H}_{2}^{\mathrm{paired}}\right)^{\perp} \cap \mathcal{H}_{2}^{\mathrm{paired}}$ characterized in \Cref{lem:vio4} (\Cref{itm:rtjl}).
          Suppose $R_{j} =A^{\mathrm{qubit}}_{uv}$.
          For any $\ket{\Phi} \in \mathcal{S}_{2,(a,b,c)}^{\mathrm{paired}}$, we check the following two cases.
          \begin{itemize}
            \item

                  $\ket{\Psi ^{*}}$ is orthogonal to $\ket{\Phi} -(-1)^{\Phi _{u} \odot \Phi
                        _{v}}
                    S_{uv}\ket{\Phi}$.
                  In fact, the overlap between these two vectors is proportional to $\operatorname{sign}\left(\ket{\Phi }\right) -( -1)^{\Phi _{u} \odot \Phi _{v}}\operatorname{sign}\left(S_{uv}\ket{\Phi }\right)$, which is 0 by \Cref{lem:xg0j}.

            \item

                  $\ket{\Psi ^{*}}$ is orthogonal to $\ket{\Phi} -F_{uv}^{12}\ket{\Phi}
                    -F_{uv}^{13}\ket{\Phi} - F_{uv}^{23}\ket{\Phi}$.
                  In fact, the overlap between these two vectors is proportional to $\operatorname{sign}\left(\ket{\Phi }\right) c( a,b,c) -\operatorname{sign}\left(F_{uv}^{12}\ket{\Phi }\right) c( a+1,b,c) -\operatorname{sign}\left(F_{uv}^{13}\ket{\Phi }\right) c( a,b+1,c) -\operatorname{sign}\left(F_{uv}^{23}\ket{\Phi }\right) c( a,b,c+1)$, which is 0 by \Cref{lem:z3zj}.

          \end{itemize}
    \item
          Next, we prove $\ket{\Psi ^{*}} -\ket{\psi _{0}}^{\otimes 4} \in M^{\perp }$,
          by expressing $\ket{\Psi ^{*}} -\ket{\psi _{0}}^{\otimes 4}$ as a linear
          combination of vectors in $\left( M\cap \mathcal{H}_{2}^{\mathrm{paired}}\right)^{\perp } \cap
            \mathcal{H}_{2}^{\mathrm{paired}} \subseteq M^{\perp }$.
          Recalls that the union of the spanning sets of each $\left( M_j\cap \mathcal{H}_{2}^{\mathrm{paired}}\right)^{\perp } \cap \mathcal{H}_{2}^{\mathrm{paired}}$ (characterized in \Cref{lem:vio4} (\Cref{itm:rtjl})) spans $\left( M\cap \mathcal{H}_{2}^{\mathrm{paired}}\right)^{\perp } \cap \mathcal{H}_{2}^{\mathrm{paired}}$.
          \begin{itemize}
            \item
                  For any $\ket{\Phi } \in \mathcal{S}_{2,( 0,0,0)}^{\mathrm{paired}}$,
                  by \Cref{lem:ls2d},
                  \begin{equation}
                    \label{eq:y76y} \ket{\Phi }
                    -\ket{\psi _{0}}^{\otimes 4} \in M^{\perp } .
                  \end{equation}
            \item
                  For any $\ket{\Phi^z} \in \mathcal{S}_{2,(
                      0,0,0)}^{\mathrm{paired}} ,\ket{\Phi^a} \in \mathcal{S}_{2,(
                      1,0,0)}^{\mathrm{paired}} ,\ket{\Phi^b} \in \mathcal{S}_{2,(
                      0,1,0)}^{\mathrm{paired}} ,\ket{\Phi^c} \in \mathcal{S}_{2,(
                      0,0,1)}^{\mathrm{paired}}$, by \Cref{lem:vio4}
                  (\Cref{itm:rtjl}) and \Cref{lem:ls2d},
                  \begin{equation}
                    \label{eq:nbp4} \ket{\Phi^z}
                    -\ket{\Phi^a} -\ket{\Phi^b} -\ket{\Phi^c} \in M^{\perp } .
                  \end{equation}
            \item
                  For any $\ket{\Phi ^{a'}} \in \mathcal{S}_{2,(
                      l,0,0)}^{\mathrm{paired}} ,\ket{\Phi ^{a''}} \in \mathcal{S}_{2,(
                      l+1,0,0)}^{\mathrm{paired}}$ with $1\le l\le \NumE-2$,
                  by \Cref{lem:vio4}
                  (\Cref{itm:rtjl}), \Cref{lem:vsdf} and \Cref{lem:ls2d},
                  \begin{equation}
                    \label{eq:pifm} \ket{\Phi ^{a'}}
                    -\ket{\Phi ^{a''}} \in M^{\perp } .
                  \end{equation}
                  And similarly by replace $a$ by $b,c$ (and the corresponding configurations),
                  we have
                  \begin{gather}
                    \ket{\Phi ^{b'}} -\ket{\Phi ^{b''}} \in M^{\perp } ,
                    \label{eq:4djs} \\
                    \ket{\Phi ^{c'}} -\ket{\Phi ^{c''}} \in M^{\perp }
                    \label{eq:y8d4}.
                  \end{gather}
            \item
                  For any $\ket{\Phi ^{az}} \in \mathcal{S}_{2,( \NumE
                      -1,0,0)}^{\mathrm{paired}} ,\ket{\Phi ^{aa}} \in \mathcal{S}_{2,( \NumE
                      ,0,0)}^{\mathrm{paired}},\ket{\Phi ^{ab}} \in \mathcal{S}_{2,(
                      l,\NumE-l,0)}^{\mathrm{paired}},\ket{\Phi ^{ac}} \in \mathcal{S}_{2,(
                      l,0,\NumE-l)}^{\mathrm{paired}}$ with $0<l<\NumE$,
                  by \Cref{lem:vio4}
                  (\Cref{itm:rtjl}), \Cref{lem:ls2d} and \Cref{lem:fsi5},
                  \begin{equation}
                    \label{eq:hk3q} \ket{\Phi ^{az}}-\ket{\Phi ^{aa}}
                    -\operatorname{sign}\left(\ket{\Phi ^{ab}}\right)\ket{\Phi ^{ab}}
                    -\operatorname{sign}\left(\ket{\Phi ^{ac}}\right)\ket{\Phi ^{ac}}.
                  \end{equation}
                  And similarly by replacing $a,b,c$ (and the corresponding configurations), we
                  have
                  \begin{gather}
                    \ket{\Phi ^{bz}}-\ket{\Phi ^{bb}}
                    -\operatorname{sign}\left(\ket{\Phi ^{ab}}\right)\ket{\Phi
                      ^{ab}}-\operatorname{sign}\left(\ket{\Phi ^{bc}}\right)\ket{\Phi ^{bc}} ,
                    \label{eq:x4xm} \\ \ket{\Phi ^{cz}}-\ket{\Phi ^{cc}}
                    -\operatorname{sign}\left(\ket{\Phi ^{ac}}\right)\ket{\Phi
                      ^{ac}}-\operatorname{sign}\left(\ket{\Phi ^{bc}}\right)\ket{\Phi ^{bc}}
                    \label{eq:zhvv}.
                  \end{gather}
          \end{itemize}
          We prove that there exists functions $s,w:\mathbb{N}\to\mathbb{R}$, such that
          \begin{equation}
            \label{eq:7aav}
            \begin{split}
               & \ket{\Psi ^{*}} -\ket{\psi
              _{0}}^{\otimes 4}                                                                                    \\
               & = \frac{1}{s( 0)}\sum_{(\ref{eq:y76y})} \qty(\ket{\Phi }
              -\ket{\psi _{0}}^{\otimes 4})                                                                        \\
               & - w( 0)\sum_{(\ref{eq:nbp4})} \left(\ket{\Phi^z} -\ket{\Phi^a} -\ket{\Phi^b} -\ket{\Phi^c}\right) \\
               & - \sum
              _{l=1}^{\NumE -2} w( l)\left(\sum_{(\ref{eq:pifm})} \qty(\ket{\Phi ^{a'}}
              -\ket{\Phi ^{a''}}) \right.                                                                          \\
               & \hskip 6em \left. +\sum_{(\ref{eq:4djs})} \qty(\ket{\Phi ^{b'}} -\ket{\Phi
              ^{b''}}) +\sum_{(\ref{eq:y8d4})} \left(\ket{\Phi ^{c'}} -\ket{\Phi ^{c''}}\right)\right)             \\
               & - w( \NumE -1)\sum_{(\ref{eq:hk3q})} \qty(\ket{\Phi ^{az}}-\ket{\Phi ^{aa}}
              -\operatorname{sign}\left(\ket{\Phi ^{ab}}\right)\ket{\Phi ^{ab}}
              -\operatorname{sign}\left(\ket{\Phi ^{ac}}\right)\ket{\Phi ^{ac}})                                   \\
               & - w( \NumE
              -1)\sum_{(\ref{eq:x4xm})} \qty(\ket{\Phi ^{bz}}-\ket{\Phi ^{bb}}
              -\operatorname{sign}\left(\ket{\Phi ^{ab}}\right)\ket{\Phi
              ^{ab}}-\operatorname{sign}\left(\ket{\Phi ^{bc}}\right)\ket{\Phi ^{bc}})                             \\
               & - w(
              \NumE -1)\sum_{(\ref{eq:zhvv})} \qty(\ket{\Phi ^{cz}}-\ket{\Phi ^{cc}}
              -\operatorname{sign}\left(\ket{\Phi ^{ac}}\right)\ket{\Phi
                ^{ac}}-\operatorname{sign}\left(\ket{\Phi ^{bc}}\right)\ket{\Phi ^{bc}}) .
            \end{split}
          \end{equation}
          Here, the subscripts indicate that the summation is taken over vectors in \Crefrange{eq:y76y}{eq:zhvv}.

          Abusing the notation $s( \cdot ,\cdot ,\cdot )$ in the proof of \Cref{thm:if80}, we define $s( l) :=s( l,0,0)$.
          $w( \cdot )$ is defined recursively as follows:
          \begin{equation}
            w( l) =
            \begin{cases}
              s( 1)^{-3}\left(s( 0)^{-1} -D\right) ,                      & l=0,                \\
              s( 2)^{-1}\left(s( 0) s( 1)^{2} w( 0) -\frac{D}{3}\right) , & l=1,                \\
              s( l+1)^{-1}\left(s( l-1) w( l-1) -\frac{D}{3}\right) ,     & 2\le l\le \NumE -2, \\
              \frac{1}{6s( \NumE -1) s( \NumE)} \cdot \frac{2}{\binom{n}{\NumE} +2}
              \cdot \frac{3}{\binom{n}{\NumE}^2-2\binom{n}{\NumE}}
              \cdot \frac{D}{3} ,                                         & l=\NumE -1.
            \end{cases}
          \end{equation}
          By comparing coefficients of vectors on both sides of \Cref{eq:7aav}, we have
          to prove the following linear equations:
          \begin{align}
            D                                                                 & = \frac{1}{s( 0)} -s( 1)^{3} w( 0), \label{eq:mvvu}                                                                         \\
            \frac{D}{3}                                                       & = s( 0) s(1)^{2} w( 0) -s( 2) w( 1), \label{eq:5cmg}                                                                        \\
            \frac{D}{3}                                                       & = s( l-1) w( l-1) -s( l+1) w( l),\quad ( 2\le l\le \NumE -2) , \label{eq:o2md}                                              \\
            \frac{D}{3}                                                       & = s( \NumE -2) w( \NumE -2) -\left(\sum _{i=1}^{\NumE -1} s( i,\NumE -i,0)\right)^2 s( \NumE) w( \NumE -1), \label{eq:96ft} \\
            \frac{2}{\binom{n}{\NumE} +2} \cdot \frac{D}{3}                   & = 6s( \NumE -1) s( \NumE) w( \NumE -1), \label{eq:eebw}                                                                     \\
            \frac{\binom{n}{\NumE} -2}{\binom{n}{\NumE} +2} \cdot \frac{D}{3} & = \qty(\sum _{i=1}^{\NumE -1} s( i,\NumE -i,0))^2 s( \NumE -1) w( \NumE -1) .
            \label{eq:kj1u}
          \end{align}
          Notice that \Cref{eq:mvvu,eq:5cmg,eq:o2md,eq:eebw} holds by definition, and one can check that \Cref{eq:kj1u} also holds.
          To prove \Cref{eq:96ft}, we multiply $s( l)$ on both sides of \Cref{eq:o2md} to
          get
          \begin{equation}
            s( l-1) s( l) w( l-1) -s( l) s( l+1) w( l) =\frac{D}{3} s(
            l) ,\quad 2\le l\le \NumE -2.
          \end{equation}
          Thus,
          \begin{equation}
            s( 1) s( 2) w( 1) -s( \NumE -2) s( \NumE -1) w( \NumE
            -2) =\frac{D}{3}\sum _{l=2}^{\NumE -2} s( l) .
          \end{equation}
          And \Cref{eq:96ft} follows.
  \end{enumerate}
\end{proof}
\subsection{Proof of Theorem \ref{thm:2tcb} (\ref{itm:tany})}

\begin{proof}[Proof of \Cref{thm:2tcb} (\ref{itm:tany})]
  It suffices to prove that if $\ket{\Phi } \in \mathcal{S}_{2,(
      a,b,c)}^{\mathrm{paired}}$ and at least two of $a,b,c$ is non-zero, then
  $\ket{\Phi } \in \qty(M\cap
    \mathcal{H}_{2}^{\mathrm{paired}} \cap \bigcap _{\tau \in
      \mathfrak{S}_{4}}\mathcal{H}_{2}^{\tau })^{\perp}$.
  The rest follows immediately from \Cref{lem:vxe0}.
  Since we did not fall into case 1 or 2, it must be $\max_{v\in V}\deg
    v\geqslant 3$, and
  \begin{enumerate}[(1)]
    \item
          $n\neq 2\NumE$,
    \item
          or $G$
          contains an odd ring as a subgraph, and the size of the ring is smaller than $n$,
    \item
          or $G$ is bipartite, but one part of $G$ has an odd size.
  \end{enumerate}
  We prove that $\ket{\Phi } \in M^{\perp }$ in the first 2 cases, and $\ket{\Phi } \in \left(M\cap \mathcal{H}_{2}^{\mathrm{paired}} \cap \bigcap _{\tau \in \mathfrak{S}_{4}}\mathcal{H}_{2}^{\tau }\right)^{\perp}$ in the \nth{3} case.
  It suffices to prove the case where exactly two of $a,b,c$ are non-zero and $2( a+b+c) =n$, since otherwise $\ket{\Phi } \in M^{\perp }$ according to \Cref{lem:vsdf}.
  Assume without loss of generality that $a,b >0$ and $2( a+b) =n$.

  \begin{enumerate}[(1)]
    \item
          If $n\neq 2\NumE$, then $a+b+c\le \min( \NumE ,n-\NumE) < \frac{n}{2}$.
          Hence, $\ket{\Phi } \in M^{\perp }$ by \Cref{lem:vsdf}.
    \item
          Suppose $G$ contains a ring $G'=( V',E')$ as subgraph, where $V'=\{v_{1} ,\dots
            ,v_{r}\} ,E'=\{( v_{i} ,v_{i+1}) | i\in [ r-1]\} \cup \{( v_{1} ,v_{r})\}$, and
          $3\le r\le n$ is an odd number.
          Assume without loss of generality that $\ket{\Phi _{v_{1}}} =\ket{I_{01}} ,\ket{\Phi _{v_{r}}} =\ket{X_{00}}$.
          Consider the vector sequence $\left(\ket{\Phi ^{( t)}}\right)_{0\leqslant t\leqslant
              T}$ where $T=2r-2$ such that
          \begin{equation}
            \ket{\Phi ^{( t)}} =
            \begin{cases}
              \ket{\Phi},                             & t=0,       \\
              S_{t,t+1}\ket{\Phi ^{( t-1)}} ,         & 1\le t<
              r,                                                   \\
              S_{2r-t-2,2r-t-1}\ket{\Phi ^{( t-1)}} , & r\le t< T, \\
              S_{1r}\ket{\Phi
              ^{( t-1)}} ,                            & t=T.
            \end{cases}
          \end{equation}
          It is easy to see that $\ket{\Phi ^{( T)}} =\ket{\Phi }$.
          Construct a vector $\ket{\Phi ^{( t)}} \pm \ket{\Phi ^{( t+1)}} \in M^{\perp }$ for each $0\leqslant t< T-1$ according to \Cref{lem:vio4} (\Cref{itm:rtjl}), one can argue that $\ket{\Phi } \in M^{\perp }$.
    \item
          Suppose the two parts of $G$ are $V_1,V_2$, and $|V_1|$ is odd.
          We first show that for any $\ket{\Phi} \in
            \mathcal{S}_{2,(a,b,c)}^{\mathrm{paired}} ,\ket{\Phi '} \in
            \mathcal{S}_{2,(a',b',c')}^{\mathrm{paired}}$, if $2(a+b) =2(a'+b') =n$ and
          $a,b,a',b' >0$, then
          \begin{equation}
            \label{eq:oyme}
            (-1)^{n_{a}\left(\ket{\Phi} ;V_{1}\right)}\ket{\Phi}
            -(-1)^{n_{a}\left(\ket{\Phi '} ;V_{1}\right)}\ket{\Phi '} \in M^{\perp} .
          \end{equation}
          Here we define $n_{a}\left(\ket{\Phi} ;V_{1}\right) :=n_{01}^{I}\left(\ket{\Phi} ;V_{1}\right) +n_{10}^{I}\left(\ket{\Phi} ;V_{1}\right)$.
          The proof is similar to \Cref{lem:ls2d} and \Cref{lem:fsi5}, hence we only
          provide a proof sketch.
          \begin{itemize}
            \item
                  \Cref{eq:oyme} holds if $a=a'$. ---
                  For any edge $(u,v)$, there exists $R_{j} \in \vec{R}$ such that $R_{j}
                    =A_{pq}^{\mathrm{qubit}}$.
                  By \Cref{lem:vio4} (\Cref{itm:rtjl}) and \Cref{lem:hpti}, we have $\ket{\Phi} -(-1)^{\Phi _{u} \odot \Phi _{v}} S_{uv}\ket{\Phi} \in \left(M\cap \mathcal{H}_{2}^{\mathrm{paired}}\right)^{\perp} \cap \mathcal{H}_{2}^{\mathrm{paired}} \subseteq M^{\perp}$.
                  On the other hand, one checks that $n_{a}\left(\ket{\Phi} ;V_{1}\right) -n_{a}\left(S_{uv}\ket{\Phi} ;V_{1}\right) \equiv \Phi _{u} \odot \Phi _{v}\pmod {2}$.
                  Hence, \Cref{eq:oyme} holds if $\ket{\Phi '} =S_{uv}\ket{\Phi}$.
                  The connectivity of $G$ then implies that \Cref{eq:oyme} holds if $\operatorname{conf}\left(\ket{\Phi}\right) =\operatorname{conf}\left(\ket{\Phi '}\right) =(a,b,0)$.
            \item
                  \Cref{eq:oyme} holds if $a=a'+1$. ---
                  We may assume that there exists $(u,v) \in E$ and $\ket{\Phi^z} \in
                    \mathcal{S}_{2,(a',b,0)}^{\mathrm{paired}}$ such that $\ket{\Phi}
                    =F_{uv}^{12}\ket{\Phi^z} ,\ket{\Phi '} =F_{uv}^{13}\ket{\Phi^z}$.
                  By \Cref{lem:vio4} (\Cref{itm:rtjl}) and \Cref{lem:hpti}, we have $\ket{\Phi^z} -\ket{\Phi} -\ket{\Phi '} -F_{uv}^{23}\ket{\Phi^z} \in \left(M\cap \mathcal{H}_{2}^{\mathrm{paired}}\right)^{\perp} \cap \mathcal{H}_{2}^{\mathrm{paired}} \subseteq M^{\perp}$.
                  By \Cref{lem:vsdf}, $\ket{\Phi^z} ,F_{uv}^{23}\ket{\Phi^z} \in M^{\perp}$.
                  Hence, $\ket{\Phi} +\ket{\Phi '} \in M^{\perp}$.
                  On the other hand, $n_{a}\left(\ket{\Phi} ;V_{1}\right) -n_{a}\left(\ket{\Phi '} ;V_{1}\right) =1$.
                  Hence, \Cref{eq:oyme} holds.
          \end{itemize}

          Remark that when $|V_{1} |$ is even, $-(-1)^{n_{a}\left(\ket{\Phi} ;V_{1}\right)}$ coincides with $\operatorname{sign}\left(\ket{\Phi}\right)$ defined in \Cref{eq:vbgx}, while when $|V_{1} |$ is odd $\operatorname{sign}\left(\ket{\Phi}\right)$ is always -1.

          Next, we show that $\ket{\Phi} \in M^{\perp} +\left(\mathcal{H}_{2}^{\tau}\right)^{\perp}$ with $\tau =(1\ 4)$, which completes the proof.
          Let $\ket{\Phi '} =\left(\PermB_{\tau}\right)^{\otimes n}\ket{\Phi}$.
          It is straightforward to check that $\operatorname{conf}\left(\ket{\Phi '}\right) =(b,a,0)$ and $n_{a}\left(\ket{\Phi} ;V_{1}\right) +n_{a}\left(\ket{\Phi '} ;V_{1}\right) =|V_{1} |\equiv 1 \pmod{2}$.
          Hence, by \Cref{eq:oyme} we have $\ket{\Phi} +\ket{\Phi '} \in M^{\perp}$.
          On the other hand, by definition of $\mathcal{H}_{2}^{\tau}$ we have $\ket{\Phi} -\ket{\Phi '} \in \left(\mathcal{H}_{2}^{\tau}\right)^{\perp}$.
          Thus, $\ket{\Phi} \in M^{\perp} +\left(\mathcal{H}_{2}^{\tau}\right)^{\perp}$.
  \end{enumerate}
\end{proof}
\section{Proof of main result: Case 3 and 4}\label{app:q52u}

In this section, we prove the exponential concentration of the cost function for alternated dUCC \ansatze{} containing both single and double (qubit) excitation rotations, with a mild connectivity assumption (\Cref{thm:8pgp} (\Cref{itm:k2mi,itm:4cl4})).
Examples of such \ansatzes{} include $k$-UCCSD, $k$-UCCGSD, $k$-UpCCGSD, $k$-qubit-UCCSD and $k$-qubit-UCCGSD.

\begin{theorem}
  \label{thm:g006}
  Let $G=(V,E)$ be a connected graph with $|V|=n$ vertices, $\vec{R}$ be a sequence of single excitation rotations $(A_{uv})_{(u,v)\in E}$ concatenated with at least one double excitation rotation $B_{pqrs}$.
  We have $\ket{\Psi ^{\vec{R}} _{2,\infty}}=\ket{\Psi^{\mathrm{qUCCGS}}_{2,\infty}}$.
\end{theorem}

\begin{proof}
  It suffices to prove that if $\ket{\Phi } \in \mathcal{S}_{2,( a,b,c)}^{\mathrm{paired}}$ and at least two of $a,b,c$ are non-zero, then $\ket{\Phi } \in M^{\perp }$.
  The theorem then follows immediately from \Cref{lem:vxe0}.
  Assume without loss of generality that $a,b >0, p>q>r>s$, and
  \begin{equation}
    \label{eq:f6h5} \ket{\Phi_p} = \ket{I_{01}}, \quad \ket{\Phi_q} = \ket{X_{00}},
    \quad \ket{\Phi_r} = \ket{I_{10}}, \quad \ket{\Phi_s} = \ket{X_{11}}.
  \end{equation}
  Denote $\vec{z}_{1} =\underset{a\in ( s,r)}{\bigoplus } \Phi _{a} ,\vec{z}_{2} =\underset{a\in ( r,q)}{\bigoplus } \Phi _{a} ,\vec{z}_{3} =\underset{a\in ( q,p)}{\bigoplus } \Phi _{a}$.
  Consider the following state sequence:
  \begin{equation}
    \label{eq:a5fa}
    \ket{\Phi^{(0)}}=\ket{\Phi},\quad \ket{\Phi ^{( 1)}} =S_{sp}
    S_{rq}\ket{\Phi } ,\quad \ket{\Phi ^{( 2)}} =S_{rq}\ket{\Phi ^{( 1)}} ,\quad \ket{\Phi ^{( 3)}} =S_{sp}\ket{\Phi ^{( 2)}} .
  \end{equation}
  Obviously $\ket{\Phi ^{( 3)}} =\ket{\Phi }$.
  By \Cref{lem:vio4} (\Cref{itm:646w}),
  \begin{equation}
    \label{eq:024f}
    \ket{\Phi ^{( 0)}} -( -1)^{( \Phi _{p} \oplus \vec{z}_{1} \oplus \vec{z}_{3})
        \odot ( \Phi _{q} \oplus \vec{z}_{1} \oplus \vec{z}_{3})}\ket{\Phi ^{( 1)}} \in
    \left(M\cap \mathcal{H}_{2}^{\mathrm{paired}}\right)^{\perp } \cap
    \mathcal{H}_{2}^{\mathrm{paired}} \subseteq M^{\perp } .
  \end{equation}
  By \Cref{lem:dbc5},
  \begin{gather}
    \ket{\Phi ^{( 1)}} -( -1)^{( \Phi _{q}
        \oplus \vec{z}_{2}) \odot ( \Phi _{r} \oplus \vec{z}_{2}) }\ket{\Phi ^{( 2)}}
    \in M^{\perp } ,\label{eq:v3ny}\\
    \ket{\Phi ^{( 2)}} -( -1)^{( \Phi _{p} \oplus
        \vec{z}_{1} \oplus \vec{z}_{2} \oplus \vec{z}_{3} \oplus \Phi _{q} \oplus \Phi
        _{r}) \odot ( \Phi _{s} \oplus \vec{z}_{1} \oplus \vec{z}_{2} \oplus
        \vec{z}_{3} \oplus \Phi _{q} \oplus \Phi _{r}) }\ket{\Phi ^{( 3)}} \in M^{\perp
      } .
    \label{eq:a6dv}
  \end{gather}
  Combining \Cref{eq:024f,eq:v3ny,eq:a6dv} to eliminate $\ket{\Phi ^{( 1)}} ,\ket{\Phi ^{( 2)}}$ we get $\ket{\Phi } \in M^{\perp }$, since $\ket{\Phi ^{(0)}}=\ket{\Phi ^{(3)}}=\ket{\Phi}$.
\end{proof}

For alternated qubit dUCC \ansatzes{}, we make mild assumptions about the topology $G$ as well as the position of qubit double excitation rotations.

\begin{theorem}
  \label{thm:dkm3}
  Let $G=(V,E)$ be a connected graph with $|V|=n$ vertices and $\max_{v\in V}\deg(v)\ge 3$, $\vec{R}$ be a sequence of qubit single excitation rotations $(A^{\mathrm{qubit}}_{uv})_{(u,v)\in E}$ concatenated with at least one qubit double excitation rotation $B^{\mathrm{qubit}}_{pqrs}$.
  Assume, only when $n=2\NumE$ and $G$ is a bipartite graph with two non-empty even parts $V=V_1\cup V_2$, that $\abs{V_{1} \cap \left\{p,q,r,s\right\}}=2$.
  We have $\ket{\Psi ^{\vec{R}} _{2,\infty}}=\ket{\Psi^{\mathrm{qUCCGS}}_{2,\infty}}$.
\end{theorem}

\begin{proof}
  If $n\neq 2\NumE$ or $G$ is not a bipartite graph with two non-empty even parts, then $\dim M = 1$ (\Cref{thm:2tcb} (\Cref{itm:tany})), and it must be $\ket{\Psi^{\vec{R}}_{2,\infty}}=\ket{\Psi^{\text{qUCCGS}}_{2,\infty}}$.
  Otherwise, we show that if $\ket{\Phi } \in \mathcal{S}_{2,( a,b,c)}^{\mathrm{paired}}$ and at least two of $a,b,c$ is non-zero, then $\ket{\Phi } \in M^{\perp }$.
  The theorem then follows immediately from \Cref{lem:vxe0}.
  Similar to the proof of \Cref{thm:g006}, we assume \Cref{eq:f6h5}, and consider the paired state sequence \Cref{eq:a5fa}.
  By \Cref{lem:vio4} (\Cref{itm:qfwn}),
  \begin{equation}
    \label{eq:jhwd}
    \ket{\Phi ^{( 0)}} +\ket{\Phi ^{( 1)}} \in \qty(M\cap
    \mathcal{H}_{2}^{\mathrm{paired}})^{\perp } \cap
    \mathcal{H}_{2}^{\mathrm{paired}} \subseteq M^{\perp } .
  \end{equation}
  By \Cref{lem:ls2d},
  \begin{equation}
    \label{eq:cghj} \ket{\Phi ^{( 1)}}
    +\ket{\Phi ^{( 2)}}, \ket{\Phi ^{( 2)}} +\ket{\Phi ^{( 3)}} \in M^{\perp }.
  \end{equation}
  Combining \Cref{eq:jhwd,eq:cghj} to eliminate $\ket{\Phi ^{( 1)}} ,\ket{\Phi ^{( 2)}}$ we get $\ket{\Phi } \in M^{\perp }$, since $\ket{\Phi ^{(0)}}=\ket{\Phi ^{(3)}}=\ket{\Phi}$.
\end{proof}

As an immediate corollary, the moment vectors of the \nth{2} moment at infinity step of $k$-UCCSD, $k$-UCCGSD, $k$-UpCCGSD, $k$-qubit-UCCSD, and $k$-qubit-UCCGSD are the same.

\begin{corollary}
  We have
  \begin{equation}
    \ket{\Psi ^{\mathrm{UCCSD}}_{2,\infty}} = 
    \ket{\Psi ^{\mathrm{UCCGSD}}_{2,\infty}} = 
    \ket{\Psi ^{\mathrm{UpCCGSD}}_{2,\infty}} = 
    \ket{\Psi ^{\mathrm{qUCCSD}}_{2,\infty}} = 
    \ket{\Psi ^{\mathrm{qUCCGSD}}_{2,\infty}} = 
    \ket{\Psi ^{\mathrm{qUCCGS}}_{2,\infty}}.
  \end{equation}
\end{corollary}

With the moment vector of the \nth{2} moment characterized, we can calculate the \nth{2} moment of the cost function for any observables.
Cases 3 and 4 of the main result are stated formally as the following corollary.

\begin{corollary}[Main result, Cases 3 and 4]
  Let $\vec{R}$ be defined in \Cref{thm:g006,thm:dkm3}, and
  $C(\vec{\uptheta};U^{\vec{R}}_{k},H_{\mathrm{el}})$ be the cost function defined in
  \Cref{eq:wtfm}, where $H_{\mathrm{el}}$ is an electronic structure Hamiltonian
  defined in
  \Cref{eq:dv4e}.
  $\lim_{k\to\infty}\Var{C}$ is the same as in \Cref{cor:e79n} (\Cref{itm:v8lq}).
\end{corollary}

Notice that $\lim _{k\to \infty }\Var{C}=\exp(-\Theta(n))$.
\section{More numerical results}\label{app:y9qk}

To support this conjecture, we further investigate how the variance of the cost function for $k$-UCCSD changes with increasing $n$ when $k$ is small.
In \Cref{fig:ghj3}, we plot the scaling of the cost variance for $k$-UCCSD, with respect to $n \in \{4, 8, \dots, 24\}$ when $k \in \{1, 2, 3, \infty\}$.
Like in main text Figure 2, we only display the results for two fixed observables $\hat{a}^{\dagger}_2 \hat{a}_1 + h.c\period$ and $\hat{a}^{\dagger}_4 \hat{a}^{\dagger}_3 \hat{a}_2 \hat{a}_1 + h.c\period$, and the estimated data points with an exponential function $a \exp (b \cdot n)$.
Unlike the main text Figure 2, the number of electrons is set to $\NumE = n/4$.
It can be seen that the observation of main text Figure 2 also holds for \Cref{fig:ghj3}, indicating that our conclusions for $k$-UCCSD are not affected by the number of electrons $\NumE$.

\begin{figure}[ht]
  \centering
  \includegraphics{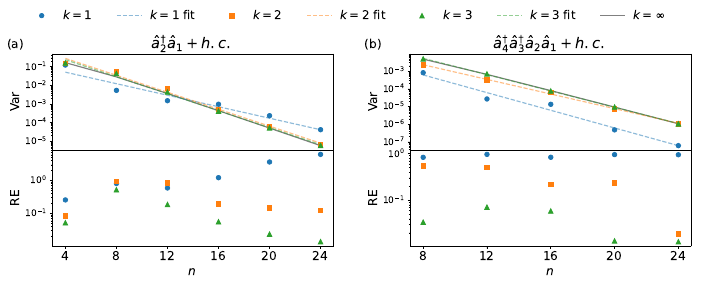}
  \caption{
    \textbf{Scaling of the cost variance and relative error to asymptotic variance for $k$-UCCSD when $\NumE = n/4$.}
    Here we present the cost variance (Var) and the relative error to asymptotic variance (RE) for the $k$ alternations of unitary coupled cluster with singles and doubles ($k$-UCCSD) \ansatze measured by (a) $\hat{a}^{\dagger}_2 \hat{a}_1 + h.c\period$ and (b) $\hat{a}^{\dagger}_4 \hat{a}^{\dagger}_3 \hat{a}_2 \hat{a}_1 + h.c\period$.
    The abbreviation ``h.c\period'' stands for ``Hermitian conjugate''.
    The number of qubit is $n \in \{4, 8, \dots, 24\}$, with the number of electrons given by $\NumE = n/4$.
    For each observable, the illustrated variances for $k = 1, 2, 3$ are estimated from 6000 random samples at each $n$ value, and are represented by blue circles, orange squares, and green triangles, respectively.
    The variances for $k = 1, 2, 3$ are fitted using the function $a \exp(b \cdot n)$, with the resulting fitting curves shown as blue, orange, and green dashed lines.
    The computed asymptotic variances ($k \to \infty$) are depicted as gray solid curves.
    In the lower half of each panel, we display the relative error of the variances at $k = 1, 2, 3$ compared to the asymptotic variances, also represented by blue circles, orange squares, and green triangles.
    The error bars, estimated as $\sqrt{2/(6000-1)} \cdot \mathrm{Var}$, are not shown as they are smaller than the data points.
  }
  \label{fig:ghj3}
\end{figure}

\end{document}